\documentclass[acmsmall,authorversion,nonacm]{acmart}
\setcopyright{none}
\usepackage{booktabs}   
\usepackage{subcaption} 
\usepackage{listings}
\usepackage{todonotes}
\usepackage{physics}
\usepackage{bm}
\usepackage{wasysym}
\usepackage{xspace}
\usepackage{multirow}
\usepackage{subcaption}

\definecolor{forestgreen}{rgb}{0.13, 0.55, 0.13}
\colorlet{myblue}{blue!50!black}
\colorlet{mygreen}{green!40!black}
\colorlet{comgreen}{green!10!gray}

\newcommand{\sdql}{SDQL\xspace}
\newcommand{\lang}{SDQLite\xspace}
\newcommand{\system}{$\nabla$SD\xspace}
\newcommand{\codestyle}{\ttfamily\color{black}}
\newcommand{\commentstyle}{\ttfamily\color{comgreen}}
\newcommand{\defines}{\triangleq}
\newcommand{\codekwstyle}{\ttfamily\bfseries\color{myblue}}
\newcommand{\codemacrostyle}{\ttfamily\bfseries\color{mygreen}}
\newcommand{\code}[1]{\textrm{\codestyle{}#1}}
\newcommand{\codekw}[1]{\textrm{\codekwstyle{}#1}}
\newcommand{\codemacro}[1]{\textrm{\codemacrostyle{}#1}}
\newcommand{\zeromacro}[1]{\codemacro{zero}\code{[}#1\code{]}}

\newcommand{\outertype}{\codemacro{$\otimes$}}

\newcommand{\tmult}[1]{\codemacro{*}$^T$[#1]}

\newcommand{\dfsmooth}{$\text{d}\widetilde{\textsc{f}}$\xspace}

\newcommand{\forwardbegin}{$\mathcal{F} \llbracket$}
\newcommand{\forwardend}{$ \rrbracket$}
\newcommand{\fadttang}[1]{$\mathcal{F}\llbracket#1\rrbracket$}
\newcommand{\fadttype}[1]{\fadttang{#1}}
\newcommand{\fadtctx}[1]{\fadttang{#1}}

\newcommand{\fvadtprime}[1]{#1}
\newcommand{\fvadtgen}[2]{$\mathcal{D}_{#1}\llbracket #2 \rrbracket$}
\newcommand{\fvadttang}[1]{\fvadtgen{\tau}{#1}}
\newcommand{\fvadttype}[1]{\fvadttang{#1}}
\newcommand{\fvadtctx}[1]{\fvadttang{#1}}
\newcommand{\tengengradbegin}[1]{$\mathcal{D}_{#1}\llbracket$}
\newcommand{\tengradbegin}{\tengengradbegin{\tau}}
\newcommand{\tengradend}{$\rrbracket$}
\newcommand{\tengrad}[1]{\fvadttang{#1}}
\newcommand{\smartpara}[1]{\noindent\textbf{#1.}}

\lstdefinelanguage{sdql}{
  morekeywords={if,then,else,let,in,not,%
  true,false,subarr,%
  sum,range,merge,%
  dense,unique,%
  int,real,dense_int,bool,string},%
  emph={%
    zero,gradient,ingrad,onehot,tensor,%
    },emphstyle=\codemacrostyle,
  sensitive,%
  morecomment=[l]//,%
  morecomment=[s]{/*}{*/},%
  morestring=[b]",%
  showstringspaces=false,%
  breaklines=true,%
  mathescape=true,%
  showspaces=false,
  showtabs=false,
  showstringspaces=false,
  breakatwhitespace=true,
  xleftmargin=1em,
  aboveskip=1pt,
  belowskip=1pt,
  lineskip=-0.2pt,
  basicstyle=\codestyle,
  keywordstyle=\codekwstyle,%
  columns=fullflexible,
 commentstyle=\commentstyle,
  escapeinside={(*@}{@*)}
}[keywords,comments,strings]%

\lstMakeShortInline[columns=fixed, keepspaces=true, language=sdql]!

\newcommand{\specialcell}[2][c]{%
  \begin{tabular}[#1]{@{}l@{}}#2\end{tabular}}

  \newcommand{\specialcellc}[2][c]{%
  \begin{tabular}[#1]{@{}c@{}}#2\end{tabular}}

  \newcommand{\tab}{\;\;\;}
  \newcommand{\grammarcomment}[1]{\textit{\small #1}}

\newcommand{\sem}[1]{\llbracket #1 \rrbracket}
\newcommand{\RR}{\mathbb{R}}
\newcommand{\NN}{\mathbb{N}}
\newcommand{\BB}{\mathbb{B}}

\begin{document}

\title{\system{}: Differentiable Programming for Sparse Tensors}

\author{Amir Shaikhha}
\email{amir.shaikhha@ed.ac.uk}
\affiliation{
  \institution{University of Edinburgh}            
  \country{United Kingdom}                    
}
\author{Mathieu Huot}
\email{mathieu.huot@stx.ox.ac.uk}
\affiliation{
  \institution{University of Oxford}            
  \country{United Kingdom}                    
}
\author{Shideh Hashemian}
\email{shideh.hashemian@aut.ac.ir}
\affiliation{
  \institution{Amirkabir University of Technology}            
  \country{Iran}                    
}


\begin{abstract}
Sparse tensors are prevalent in many data-intensive applications, yet existing differentiable programming frameworks are tailored towards dense tensors. 
This presents a significant challenge for efficiently computing gradients through sparse tensor operations, as their irregular sparsity patterns can result in substantial memory and computational overheads. 
In this work, we introduce a novel framework that enables the efficient and automatic differentiation of sparse tensors, addressing this fundamental issue. Our experiments demonstrate the effectiveness of the proposed framework in terms of performance and scalability, outperforming state-of-the-art frameworks across a range of synthetic and real-world datasets. Our approach offers a promising direction for enabling efficient and scalable differentiable programming with sparse tensors, which has significant implications for numerous applications in machine learning, natural language processing, and scientific computing.
\end{abstract}

\keywords{Sparse Tensor Algebra, Automatic Differentiation, Semi-Ring Dictionaries}

\maketitle

\section{Introduction}
Sparse tensors are essential in many scientific and engineering applications, such as natural language processing, computer vision, and graph analytics. Unlike dense tensors, which store all of their elements regardless of their value, sparse tensors only store non-zero values, resulting in significant memory savings and computational efficiency. Sparse tensors also enable efficient representation and manipulation of high-dimensional data structures, which are often encountered in modern machine learning and scientific computing, such as sparse tensors representing the frequency of words in a document or corpus in natural language processing,  adjacency matrices of large and sparse graphs in network/relational analysis, or sparse user-item interaction matrices for collaborative filtering in recommender systems.
This has inspired recent interest in developing better support for sparse tensors~\cite{DBLP:journals/pieee/StroutHO18,kjolstad:2017:taco,DBLP:journals/cgf/TangSKPLP20}.

Automatic differentiation (AD) is a fundamental technique in modern machine learning and scientific computing that enables efficient computation of the gradient of a function. This is crucial for optimization, parameter estimation, and many other applications in which gradient-based optimization methods are employed. While AD tools for dense tensors are abundant and well-established, the lack of efficient AD tools for sparse tensors hinders the wider adoption of these techniques and represents a major research challenge in this field. Libraries such as TensorFlow, PyTorch, and JAX provide efficient and scalable implementations of gradient computation for dense tensor operations, but their support for sparse operations is limited~\cite{tensorflow:sp:issue,jax:sp:issue,pytorch:sp:issue}. As a result, there have been various efforts to manually provide differentiation for particular sparse tensor kernels~\cite{nytko2022optimized}.

AD for sparse tensor algebra is more challenging than for dense tensor algebra for several reasons. Firstly, the structure of sparse tensors is more complex than that of dense tensors, and their sparsity patterns are often irregular and vary across different operations. This can make it challenging to efficiently propagate gradients through the computation graph and to identify which elements of the sparse tensor are relevant for the gradient computation. Secondly, sparse tensor operations often require the use of specialized data structures and algorithms, such as compressed sparse row (CSR) or compressed sparse column (CSC) formats, which are not natively supported by most AD-enabled frameworks or have incomplete AD rules.

This paper presents \system{}, a novel differentiable programming framework that supports the automatic differentiation of arbitrary sparse computations. To the best of our knowledge, \system{} is the first framework that provides this capability.
As opposed to the existing frameworks that offer AD support for a limited number of sparse kernels~\cite{nytko2022optimized,tensorflow2015-whitepaper,Paszke_PyTorch_An_Imperative_2019,jax2018github}, \system{} allows the AD of an arbitrary sparse computation expressible in tensor algebra. 

The key insight is to \textit{perform differentiation over a logical representation of a sparse tensor}. 
This means that there is a clear separation of concerns between the semantics of differentiation over a program on the one hand, and optimizations and data layout representations on the other (cf. Figure~\ref{fig:archhigh}). To see the simple example of the dot product of two vectors, see Example 1 (cont.) for its logical form before differentiation in \S\ref{sub:sparse-tensors-semiring-dico}, after differentiation in \S\ref{sub:tensorized-fad}, its physical coordinate format (COO) representation in \S\ref{sub:storage-comp}, and in \S\ref{sub:dps-and-cpp} for the generated code.

The physical representation of sparse tensors (e.g., COO/CSR/CSC) involves multiple arrays storing a compressed representation of the matrix (cf. Figure~\ref{fig:bowexample:2csr}). The computations over such representations involve \textit{imperative} nested loops over these arrays. However, our logical representation uses a nested dictionary, where sparse computations are expressed \textit{functionally} as nested summations over them. The logical representation can be later fused with a concrete physical storage format (cf. Section~\ref{sub:storage-comp}).

In more detail, the contributions of this paper are as follows:


\begin{figure}
    \centering
    \includegraphics[width=0.96\linewidth]{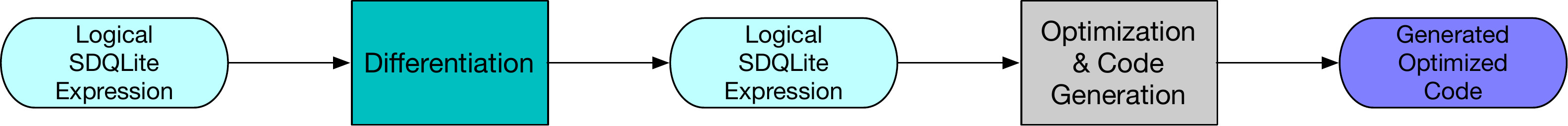}
    \caption{The high-level overview of transformations in \system{}.}
    \label{fig:archhigh}
\end{figure}

\begin{itemize}
    \item We present \system{}, the first framework with systematic support for the automatic differentiation of sparse tensors. \system{} is based on \lang{}~\cite{DBLP:journals/pacmpl/ShaikhhaHSO22,schleich2022optimizing}, a recently introduced intermediate language that can express sparse tensor workloads by \textit{separating the tensor computations from the storage specifications} (Section~\ref{sec:overview}). 
    \item We introduce a novel \textit{tensorized forward-mode AD} that computes the gradients in a batch (Section~\ref{sec:diff}). Our automatic differentiation transform is over the logical part of the language, which we call Logical \lang{}, without worrying about the physical storage formats. 
    \item The differentiated program is then optimized by leveraging the following AD-agnostic transformations: (1) \textit{sparsity-aware rewrite rules}, (2) composing with the physical storage formats, and (3) algebraic rewrite rules applied in a cost-based manner using \textit{equality saturation}~\cite{tate2009equality,DBLP:journals/pacmpl/WillseyNWFTP21} (Section~\ref{sec:opt}). In addition, \system{} performs low-level transformations for removing unnecessary intermediate tensors appearing in nested loops before generating low-level C++ code.
    \item We show the correctness of our approach (Section~\ref{sec:semantics}). That is, AD is a well-typed transformation, computes derivatives  of programs, and our optimizations are sound with respect to our denotational semantics.
    \item We experimentally validate the effectiveness of \system{} in comparison with the state-of-the-art frameworks (Section~\ref{sec:exp}). We demonstrate that \system{} scales the gradient computation to large matrices with many zero elements over both real-world and synthetic datasets.
\end{itemize}

\section{Background}
\label{sec:bg}
\subsection{Automatic Differentiation}

Automatic differentiation (AD) is a powerful and widely used technique in machine learning and scientific computing that enables efficient computation of the gradient of a function. The gradient is a crucial quantity in many optimization, parameter estimation, and machine learning algorithms, and its computation is often a bottleneck in the training process. AD provides a computationally efficient and accurate way of computing the gradient by breaking down a function into a series of elementary operations and applying the chain rule to compute the derivatives of each operation. The result is an exact gradient that is computed with a similar computational cost as the original function, with no need for approximate methods such as finite differences, or manual derivation.

\smartpara{Forward Mode}
One method of computing the gradient of a function is the forward-mode AD which involves computing the derivatives of each operation in the forward direction through the computational graph. The method starts with the input variables and propagates the values and their derivatives through the graph, one operation at a time, until the output variable is reached. At each operation, the derivative of the output variable with respect to each input variable is computed using the chain rule, and these derivatives are stored in a computational graph that can be used to compute the gradient of the function.

\smartpara{Example 1 - Vector Dot Product} Consider the following function:
\[f([x_1, x_2], [y_1, y_2]) = x_1y_1 + x_2y_2\]

\noindent which takes two pairs of input variables $x_1$, $x_2$ and $y_1$, $y_2$ and computes the dot product of the vectors $x = [x_1, x_2]$ and $y = [y_1, y_2]$. To compute the gradient of $f$ with respect to the inputs using forward-mode AD, we start by converting the program into ANF~\cite{DBLP:conf/pldi/FlanaganSDF93}:

\begin{tabular}{r l}
$f([x_1, x_2], [y_1, y_2])=$ 
& \texttt{let} $t_1 = x_1y_1$ \\
& \texttt{let} $t_2 = x_2y_2$ \\
& \texttt{let} $t_3 = t_1 + t_2$ \\
& $t_3$ \\
\end{tabular}

The forward-mode AD lifts every variable to a dual number by associating a tangent variable $v'$ to each input and intermediate variable $v$. Then, each intermediate tangent variable is computed by following the chain rule. In the previous example, the function $f$ is transformed as follows:

\begin{tabular}{r l}
$f'([x_1, x_2], [y_1, y_2], [x_1', x_2'], [y_1', y_2'])=$  
& \texttt{let}  $t_1 = x_1y_1$ \\
& \texttt{let} $t_1' = x_1'y_1 + x_1y_1'$ \\
& \texttt{let} $t_2 = x_2y_2$ \\
& \texttt{let} $t_2' = x_2'y_2 + x_2y_2'$ \\
& \texttt{let} $t_3 = t_1 + t_2$ \\
& \texttt{let} $t_3' = t_1' + t_2'$ \\
& $t_3'$ \\
\end{tabular}

To compute the partial derivative of $f$ with respect to each input we need to set the corresponding tangent variable to $1$ and the other input tangent variables to $0$. For example, the gradient of $f$ with respect to the first vector is computed by the following partial derivative computations:

\begin{tabular}{r l l}
$f'([a_1,a_2],[b_1,b_2],[1,0],[0,0]) \to^*$ & $b_1 =$ & $\frac{\partial \llbracket f\rrbracket}{\partial x_1}(a_1,a_2,b_1,b_2)$ \\
$f'([a_1,a_2],[b_1,b_2],[0,1],[0,0]) \to^*$ & $b_2 =$ & $\frac{\partial \llbracket f\rrbracket}{\partial x_2}(a_1,a_2,b_1,b_2)$ \\
\end{tabular}

\smartpara{Reverse-mode AD}
Forward mode AD is computationally expensive for the derivative computation of scalar-valued functions with tensor inputs, which among other use cases appear in training machine learning models by optimizing an objective function. This is due to the fact that when differentiating a program representing a function $\mathbb{R}^n\to\mathbb{R}$, as is often the case in these contexts, one needs $n$ runs of the program transformed by forward-mode to obtain the whole gradient.
The reverse-mode technique, which computes the gradient of such functions in one run, is then more appropriate, and is massively used in deep learning frameworks~\cite{tensorflow2015-whitepaper,Paszke_PyTorch_An_Imperative_2019,jax2018github}. There has been recent interest on bridging the gap between theoretical correctness guarantees and more practical, efficient implementations with complexity guarantees \cite{DBLP:journals/pacmpl/KrawiecJKEEF22,DBLP:journals/pacmpl/SmedingV23,DBLP:journals/corr/abs-2212-09801}.

\smartpara{Example 1 (cont.)} Consider a generalization of the previous function, where $\cdot$ denotes the dot product of two vectors:
 \[f(V_1, V_2) = V_1 \cdot V_2\]

\noindent If each input vector has $m$ elements, then the cost of forward-mode AD is $O(m^2)$ as it requires $m$ forward passes, each costing $O(m)$. However, reverse-mode AD can compute the gradient by one forward pass to compute the primal values and one reverse pass to compute the gradient values, resulting in an $O(m)$ overall complexity.  

\smartpara{Vector Forward-mode AD}
Inspired by the use cases that require the computation of the full Jacobian matrix~\cite{more2006levenberg}, there have been efforts on batch computations of forward-mode AD~\cite{DBLP:journals/oms/KhanB15,DBLP:journals/pacmpl/ShaikhhaFVJ19}. It has been experimentally demonstrated that by leveraging rewrite rules, one can even recover the asymptotic performance of reverse-mode AD on vectorized forward-mode AD~\cite{DBLP:journals/pacmpl/ShaikhhaFVJ19}. 

\begin{figure}
    \centering
    \includegraphics[width=\linewidth]{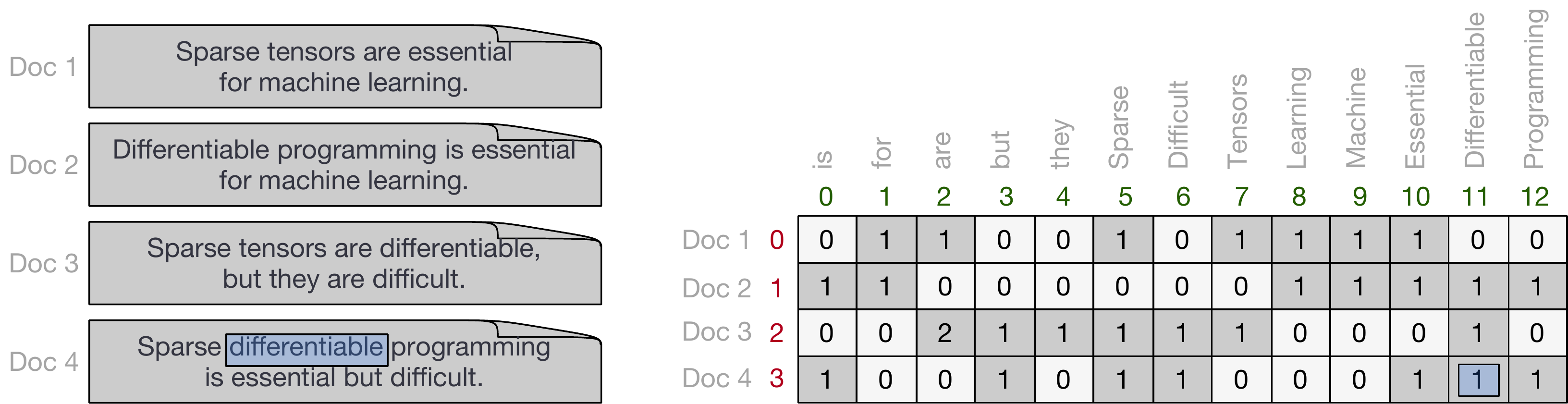}
    \caption{The bag-of-words representation of documents using a matrix with documents as rows and words as columns. Note that the matrix has many zero elements, i.e., it is a sparse matrix.}
    \label{fig:bowexample:1}
\end{figure}

\subsection{Sparse tensors and semi-ring dictionaries}
\label{sub:sparse-tensors-semiring-dico}

\smartpara{Sparse Tensors} Sparse tensors are a type of data structure that is commonly used to represent high-dimensional data that have a majority of zero values. Sparse tensors have a compact representation that only stores the non-zero values and their corresponding indices, which makes them more memory-efficient than dense tensors for large-scale data. Sparse tensors are used in many domains, including natural language processing, computer vision, and scientific computing. For example, in natural language processing, sparse tensors can be used to represent text data as a bag-of-words or term frequency-inverse document frequency (TF-IDF) matrix, where the rows correspond to documents and the columns correspond to words (Figure~\ref{fig:bowexample:1}).

Sparse tensors can be manipulated using a variety of specialized algorithms and data structures, such as compressed sparse row (CSR) and compressed sparse column (CSC) formats, which enable efficient matrix-vector multiplication and other operations. However, the irregular sparsity patterns of sparse tensors pose significant challenges for automatic differentiation.

\smartpara{Example 1 (cont.)}
In the previous example, if the majority of the elements of the input vectors of sizxe $m$ are zeros (the number of non-zero elements, denoted by $nnz$, is such that $nnz<<m$), one can use the CSR representation shown in Figure~\ref{fig:bowexample:2csr}. In this representation, the array !pos! is a compressed representation of rows, whereas !idx! and !val! show the columns and values of non-zero elements. For example, !pos(0)=0,pos(1)=7! depicts that the indices 0 to 6 of !idx!/!val! correspond to the column/value of the elements in the row=0 of the matrix, and 7 to 14 for the positions of the row=1. However, existing linear algebra frameworks do not support gradients over this representation. Thus, rather than computing the gradient in $O(nnz)$, they compute it over the dense representation in $O(m)$.

\begin{figure}
    \centering
    \begin{subfigure}[t]{\linewidth}
    \includegraphics[width=\linewidth]{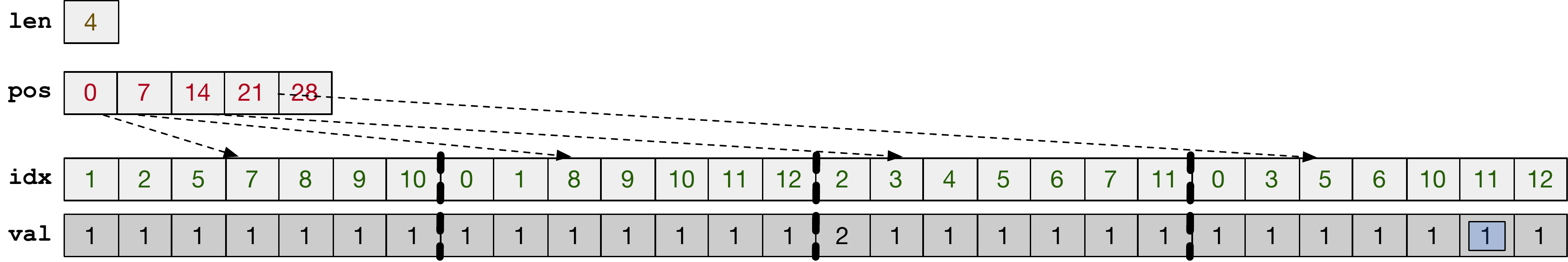}
    \subcaption{Compressed sparse row (CSR) representation.}
    \label{fig:bowexample:2csr}
    \end{subfigure}
    \begin{subfigure}[t]{0.7\linewidth}
    \includegraphics[width=\linewidth]{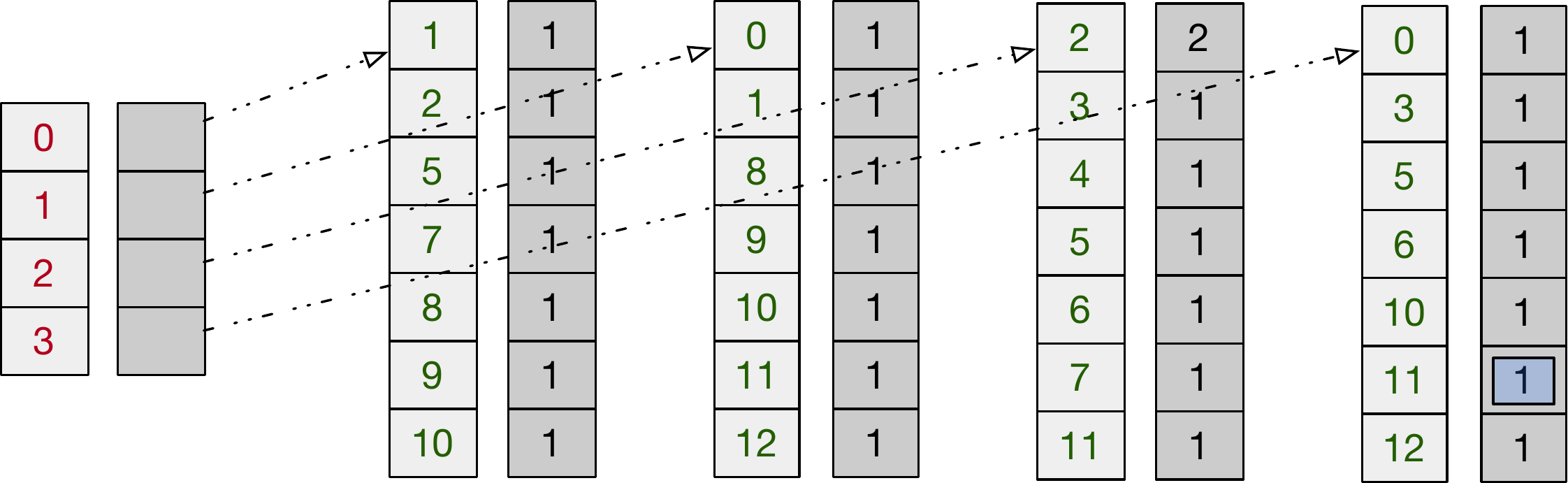}
    \subcaption{The logical semi-ring dictionary representation.}
    \label{fig:bowexample:2sd}
    \end{subfigure}
    \vspace{-0.3cm}
    \caption{Sparse representations of the previous bag-of-words matrix.}
    \label{fig:bowexample:2}
\end{figure}

\smartpara{Semi-ring Dictionaries} Semi-ring dictionaries are data structures that subsume sets, multisets, and dense/sparse tensors~\cite{DBLP:journals/pacmpl/ShaikhhaHSO22}. A semi-ring is a set with two binary operations that satisfy certain axioms, such as associativity, distributivity, and commutativity. For example, the set of non-negative integers with addition and multiplication forms a semi-ring, and the set of Booleans with logical $\vee$ and $\wedge$ forms another semi-ring. Semi-ring dictionaries are designed to represent sparse tensors as key-value pairs. In sparse vectors, the keys correspond to the vector indices and the values correspond to the non-zero elements. In sparse matrices, the keys correspond to the row indices and the values correspond to the sparse vector associated with that row. The semi-ring operations are then defined in terms of the corresponding operations on the values, such as addition or multiplication. The multiplication operator for semi-ring dictionaries has a semantics of tensor outer product, as can be observed in the next example.

\smartpara{Example 2 - Scalar-Vector Product} Consider the scalar-vector product between a scalar value !s! and a vector value !V! represented using a semi-ring dictionary. The equivalent semi-ring dictionary representation is !s * V!, where !*! has the semantics of tensor outer product.

\smartpara{\sdql} SDQL~\cite{DBLP:journals/pacmpl/ShaikhhaHSO22} is a functional language for querying against semi-ring dictionaries. SDQL is expressive enough to capture database queries and linear algebra expressions; this makes it appropriate as an intermediate language for hybrid database and machine learning workloads. SDQL provides the following constructs for manipulating semi-ring dictionaries: 
\begin{enumerate} 
    \item !dict(k)! accesses the value associated with the key !k! in !dict!, and if the key does not exist, it returns the addition identity of semi-ring ($0$ in the case of real and natural numbers).
    \item !{k -> v}! constructs a dictionary with a single key-value pair of !k! and !v!.
    \item !sum(<k,v> in dict) f(k,v)! folds over each key-value pair of !dict! and computes the summation of !f(k,v)! starting from the addition identity of the corresponding semi-ring.
\end{enumerate}

\smartpara{Example 1 (cont.)} The equivalent \sdql expression for $V1 \cdot V2$ can be one of the following two:

\begin{tabular}{l c}
\begin{lstlisting}
sum(<i, a> in V1) a * V2(i)
\end{lstlisting} & \hspace{2cm}
\begin{lstlisting}
sum(<i, a> in V2) V1(i) * a
\end{lstlisting}
\end{tabular}

\noindent
The preferred choice depends on the number of non-zero elements of !V1! and !V2!. If !V1! (resp. !V2!) has fewer non-zero elements, the left (resp. right) variant is more efficient. Otherwise, if both have the same number of non-zero elements (e.g., both are dense), both variants have the same performance.

\smartpara{\lang{}} \lang{}~\cite{schleich2022optimizing} is a dialect of SDQL tailored for sparse tensor processing; it restricts SDQL to the types required for sparse tensors while extending it with constructs required for different sparse storage formats (e.g., CSR, CSC). The following constructs are central to \lang{}:

\begin{enumerate} 
    \item !(st:en)! specifies a dense array holding the range of numbers from !st! to !en! (excluding).
    \item !arr(st:en)! specifies the sub-array of !arr! ranging from !st! to !en! (excluding).
    \item Additional annotations for guiding rewrite rules such~\cite{schleich2022optimizing}.
\end{enumerate}



\section{Languages}
\label{sec:overview}

In this section, we give an overview of the languages used in \system{}.
We divide \lang{} in two. The smaller fragment, Logical \lang{}, on which AD will be performed, and the full fragment, Physical \lang, which augments the logical subset of the language with the constructs for expressing the different sparse storage formats. The grammar and most important typing rules of these languages are shown in Figure~\ref{fig:langs}.

\begin{figure*}[t]
\setlength{\tabcolsep}{0.3em}
\centering
\begin{tabular}{|l c l|l|c|}
\hline
\multicolumn{3}{|c|}{\textbf{Construct}} & \multicolumn{1}{c|}{\textbf{Description}} & \textbf{Language} \\\hline \hline
!e! & \mbox{::=} & !sum(<x,y> in e) e!  & \grammarcomment{Dictionary Aggregation} & \multirow{10}{*}{\specialcellc{Logical \\ \lang}} \\
& $\mid$ & !{ e -> e }! \tab  $\mid$ \tab !{ }! \tab  $\mid$ \tab !e(e)! & \grammarcomment{Singleton/Empty Dictionary, Lookup} &\\
& $\mid$ & !let x = e in e! \tab $\mid$ \tab !x! & \grammarcomment{Variable Binding \& Access} &\\
& $\mid$ & !not e! \tab $\mid$ \tab !if e then e! & \grammarcomment{Negation, Conditional} &\\
& $\mid$ & !e + e! \tab $\mid$ \tab !e * e! & \grammarcomment{Addition, Multiplication} &\\
& $\mid$ & !n! \;\; $\mid$ \;\; !r! \;\; $\mid$ \;\; !false! \;\; $\mid$ \;\; !true! & \grammarcomment{Numeric and Boolean Constants} &\\ 
& $\mid$ & !op(e)! & \grammarcomment{Real unary operation such as $\cos,\exp,\sin$} & \\
& $\mid$ & !e = e! & \grammarcomment{Equality of discrete types: bool and int} & \\ 
\cline{1-4}
!T! & \mbox{::=} & !I! \;\; $\mid$ \;\; !bool! \;\; $\mid$ \;\; !D! &  \grammarcomment{Scalar and Tensor Types} &\\
!D! & \mbox{::=}  & !real! \;\; $\mid$ \;\; !{ I -> D }!  & \grammarcomment{Dictionary (Tensor) Type} &\\
!I! & \mbox{::=} & !int! & \grammarcomment{Index Types (Only Integer)} &\\
\hline 
\hline
!e! & \mbox{::=} & !...!  & \grammarcomment{Source \lang Expressions} & \multirow{5}{*}{\specialcellc{Physical \\ \lang}} \\
& $\mid$ & !e:e! \tab $\mid$ \tab !e(e:e)! & \grammarcomment{Range and Sub-Array} &\\
& $\mid$ & \textit{cf.~\cite{schleich2022optimizing}} & \grammarcomment{Other Constructs} &\\
\cline{1-4}
!T! & \mbox{::=} & !...! &  \grammarcomment{Source \lang  Types} &\\
!D! & \mbox{::=} & !...! \;\; $\mid$ \;\; !int! &  \grammarcomment{Tensor Types with Integer Value} &\\ 
!I! & \mbox{::=} & !...! \;\; $\mid$ \;\; !dense_int! &  \grammarcomment{Index Types (+ Dense Integer)} &\\ \hline
\hline 
\multicolumn{3}{|l}{Definition of \codemacro{zero}\code{[D]}:} & \multicolumn{2}{|l|}{Definition of \outertype{}:}\\
\multicolumn{3}{|l}{\zeromacro{\codekw{real}} $\defines$ \code{0}} &
\multicolumn{2}{|l|}{\codekw{real} \outertype{} \code{D} $\defines$ \code{D}} \\ 
\multicolumn{3}{|l}{\zeromacro{\code{\{I -> D\}}} $\defines$ \code{\{ \}}} &
\multicolumn{2}{|l|}{\code{\{I -> D1\}} \outertype{} \code{D2} $\defines$ \code{\{I -> D1 \outertype{} D2\}}} \\  \hline
\multicolumn{3}{|l}{Definition of \codekw{let}-tupling:} & \multicolumn{2}{|l|}{Definition of \codemacro{tensor} \code{n}:} \\
\multicolumn{3}{|l}{
\codekw{let} \code{<v1, ..., vn> = <e1, ..., en>} \codekw{in} \code{e}
}  & \multicolumn{2}{|l|}{\codemacro{tensor} \code{0} $\defines$ \codekw{real}}\\
\multicolumn{3}{|l}{ $\defines$
\codekw{let} \code{v1=e1} \codekw{in} \code{...}
\codekw{let} \code{vn=en} \codekw{in} \code{e}}  & \multicolumn{2}{|l|}{\codemacro{tensor} \code{n} $\defines$ \code{\{} \codekw{int} \code{->} \codemacro{tensor} \code{(n-1) \}} 
}\\ \hline

\end{tabular}

\begin{tabular}{|c|}
\hline
\begin{tabular}{c}
$\Gamma \vdash$ !e1!: !tensor n+1! $\quad$ $\Gamma$, !k!: !int!, !v!: !tensor n! $\vdash$ !e2!: !D!\\\hline
$\Gamma \vdash$ !sum(<k,v> in e1) e2!: !D!
\end{tabular}
\hspace{0.48cm}
\begin{tabular}{c}
$\Gamma \vdash$ !e1!: !tensor n+1! $\quad$ $\Gamma \vdash$ !e2!: !int!\\\hline
$\Gamma \vdash$ !e1(e2)!: !tensor n!
\end{tabular} \\
\begin{tabular}{c}
$\Gamma \vdash$ !k!: !int! $\,\,$ $\Gamma \vdash$ !v!: !tensor n!\\\hline
$\Gamma \vdash$ !{ k -> v }!: !tensor (n+1)!
\end{tabular}\hspace{0.5cm}
\begin{tabular}{c}
$\Gamma \vdash$  !e1!: !D1! $\quad$ $\Gamma \vdash$ !e2!: !D2! \\\hline
$\Gamma \vdash$  !e1 * e2!: !D1! \outertype{} !D2!
\end{tabular}
\\ \hline
\end{tabular}

\caption{Grammar and a subset of typing rules of the languages used in \system{}.}
\label{fig:langs}
\end{figure*}

\smartpara{Logical \lang{}} Initially, the program is expressed in a subset of \lang{} that is sufficient for expressing tensor programs at the logical level, i.e., without worrying about the storage format. 
Thus, at this stage, we do not require the support for dense arrays. Furthermore, there is no need for expressing tuples. Nevertheless, for convenience, we use tupled let-binding as a syntactic sugar for multiple let bindings.

Logical \lang{} is expressive enough to capture Einstein summations~\cite{DBLP:journals/pacmpl/ShaikhhaHSO22}. For example, the matrix-matrix product for two matrices !M1! and !M2! (both represented as nested dictionaries) is expressed as:

\begin{lstlisting}
sum(<i,row> in M1) { i ->
  sum(<j, v1> in row)
    sum(<k, v2> in M2(j)) { k ->
      v1 * v2 } }
\end{lstlisting}

\noindent However, the expressiveness goes beyond Einstein summations. As an example, !map! of function !f! over the values of a tensor !e! is expressed as !sum(<k,v> in e) { k -> f(v) }!. 


\smartpara{Physical \lang{}} The storage specifications require dense-array-related constructs such as range (!st:end!) and subarray (!arr(st:end)!). Thus, after combining the differentiated program with the storage specification, we need to include these constructs for the Physical \lang{}. We go back to this intermediate language in Section~\ref{sec:opt}. The next section focuses primarily on Logical \lang{} and the differentiation rules over its constructs.


\section{Differentiation}
\label{sec:diff}
In this section, we present the differentiation transformations applied to Logical \lang{} expressions. First, for exposition purposes, we present a variant of traditional forward-mode AD (FAD).
Then, we show a tensorized FAD that not only subsumes the traditional FAD, but also computes gradients more efficiently.
Finally, we show the high-level API exposed to the programmers.

\subsection{Scalar Forward-Mode Transformation}
Traditional FAD uses dual numbers to compute the tangent (derivative) component along with the actual (original) computation. We refer to it as scalar FAD because for each scalar expression, it stores a scalar tangent component.

Similar to other FAD frameworks, \system{} precedes the differentiation transformation with an ANF conversion~\cite{DBLP:conf/pldi/FlanaganSDF93}.
This allows for sharing sub-expressions and avoids duplication of computations for non-unary constructs such as multiplication.
Logical \lang{} does not allow function definitions nor higher-order functions; all functions need to be inlined~\cite{paszke2021getting}.

\smartpara{Scalar Constructs} Figure~\ref{fig:forward_diffrules} shows the forward-mode transformation rules.
As opposed to existing functional AD systems, \system{} does not use explicit pair construction and projection for dealing with dual numbers.
Instead, the \fadttang{} construct only computes the tangent part of differentiation and refers to the expressions in the ANF transformed program for primal components (cf. the rule for let binding).
This avoids the need to extend the target language of differentiation with pairing constructs. Furthermore, this makes the differentiation rules simpler. For every unary real operation !op!, we assume that the language has a unary real operation !op'! representing its derivative.
Finally, the differentiation for all discrete types (!int! and !bool!) is !0!.

\begin{figure}
\begin{tabular}{|lcl|}
\hline
 & & \vspace{-0.2cm} \\
\multicolumn{3}{|c|}{
\begin{tabular}{c|c}
\begin{tabular}{l c l}
\multicolumn{3}{c}{\framebox[1.2\width]{\fadttype{\code{T}}} Forward mode on Types} \\
 & & \vspace{-0.3cm} \\
   \fadttype{\code{ D }}  & = &
   \code{ D } \\ 
   \fadttype{\codekw{ bool }}  & = &
   \codekw{ real } \\ 
   \fadttype{\codekw{ int }}  & = &
   \codekw{ real } \\ 
\end{tabular} &
\begin{tabular}{l c l}
\multicolumn{3}{c}{\framebox[1.2\width]{\fadtctx{\Gamma}} Forward mode on Context} \\
 & & \vspace{-0.3cm} \\
   \fadtctx{\emptyset}  & = &
   $\emptyset$ \\ 
   \fadtctx{\Gamma, \code{x:T}}  & = &
   \fadtctx{\Gamma}$, \code{x:T}, \code{x':}$\fadttype{\code{T}} \\ 
\end{tabular}
\end{tabular}
} \\
& & \vspace{-0.2cm}\\
\multicolumn{3}{|c|}{\framebox[1.2\width]{\fadttang{\code{e}}} Forward mode on Expressions} \\
& & \vspace{-0.2cm}\\
\multicolumn{3}{|c|}{Invariant: If $\Gamma \vdash$ \code{e: T}, then \fadtctx{\Gamma} $\vdash$ \fadttang{\code{e}}\code{:} \fadttype{\code{T}}} \\
& & \vspace{-0.2cm}\\
   \fadttang{\code{ \codekw{sum}(<k,v> \codekw{in} e1) e2 }}   & = & \code{\codekw{sum}(<k,v> \codekw{in} \fvadtprime{\code{e1}})} \\
   & & \code{ \codekw{let} <k', v'> = <0, \fadttang{\code{e1(k)}}> \codekw{in} } \\
   & & \code{ \fadttang{\code{e2}}} \\
\fadttang{\code{ \codekw{let} x = e1 \codekw{in} e2 }}  & = &
   \code{\codekw{let} <x, x'> = <e1, \fadttang{\code{e1}}> \codekw{in}}
   \\ 
   & & \code{\fadttang{\code{e2}}
   } \\
   \fadttang{\code{ \codekw{if} e1 \codekw{then} e2 }}  & = &
   \code{\codekw{if} \fvadtprime{\code{e1}} \codekw{then} \fadttang{\code{e2}}
   } \\
   \fadttang{\code{ e1(e2) }}  & = &
   \fadttang{\code{e1}}\code{(}\fvadtprime{\code{e2}}\code{)} \\
   \fadttang{\code{ \{ e1 -> e2 \} }}  & = &
   \code{\{ \fvadtprime{\code{e1}} -> \fadttang{\code{e2}} \}} \\ 
   \fadttang{\code{ e1 * e2 }}  & = &
   \code{\fvadtprime{\code{e1}} * \fadttang{\code{e2}} +
   \fadttang{\code{e1}} * \fvadtprime{\code{e2}}
   } \\
   \fadttang{\code{ e1 + e2 }}  & = &
   \code{\fadttang{\code{e1}} + \fadttang{\code{e2}}
   } \\
   \fadttang{\code{ op(e) }}  & = &
   \code{op'(e) * }\fadttang{\code{e}} \\
   \fadttang{\code{ x }}  & = &
   \code{x'} \\
   \fadttang{\code{ r }} & = &
   \code{0} \\
   \fadttang{\code{ n }} = \fadttang{\codekw{false}} = \fadttang{\codekw{true}} & = &
   \code{0} \\
\hline
\end{tabular}
\caption{Forward-mode AD transformation rules for \lang expressions. In order to maximize sharing, all sub-expressions need to be let-bound, i.e., ANF transformation should be applied before differentiation. 
} 
\label{fig:forward_diffrules}
\end{figure}

\smartpara{Example 1 (cont.)} Consider again the case of the dot-product of two unrolled vectors of size two initially used in Section~\ref{sec:bg}. Applying differentiation over the ANF transformed program in \lang{} is as follows:

\begin{tabular}{l l l}
$\mathcal{F}\Biggl\llbracket$
&
\begin{lstlisting}
let t1 = x1*y1 in
let t2 = x2*y2 in
let t3 = t1+t2 in
t3
\end{lstlisting} &
$\Biggl\rrbracket$
\end{tabular}

\noindent After applying the rules in Figure~\ref{fig:forward_diffrules}, we obtain the following program:

\begin{lstlisting}
let <t1, t1'> = <x1*y1, x1*y1' + x1'*y1> in
let <t2, t2'> = <x2*y2, x2*y2' + x2'*y2> in
let <t3, t3'> = <t1+t2, t1'+t2'> in
t3'
\end{lstlisting}

\noindent Note that the let-binding constructs are syntactic sugar; there is no pair created (cf. Figure~\ref{fig:langs}).

\smartpara{Tensor Constructs} By choosing not to incorporate pairs in the language, we have eliminated the option of differentiating vectors as vectors of pairs (i.e., arrays of structs). Instead, an expression of type !tensor n! is differentiated as an expression with the same type. One of the interesting tensor-based differentiation rules is our rule for summation, where we need to access the corresponding element from the differentiated range, as we can observe in the following example.

\smartpara{Example 1 (cont.)} Let us go back to the dot-product for two vectors !V1! and !V2!. The differentiation transformation is expressed as follows:

\begin{tabular}{l l l}
\forwardbegin{} &
\begin{lstlisting}
sum(<i,a> in V2) V1(i) * a
\end{lstlisting} &
\forwardend{}
\end{tabular}

\noindent Applying differentiation rules results in the following program:

\begin{lstlisting}
sum(<i,a> in V2) 
  let <i',a'> = <0,V2'(i)>
  V1(i) * a' + V1'(i) * a
\end{lstlisting}

In order to compute the gradient of this function with respect to one of its vector inputs, say !V1!, we need to repeatedly set !V1'! into a one-hot vector that is !1! at index !j! and !0! everywhere else, and set !V2'! to be the zero vector.
This requires multiple rounds of running the forward-mode AD for different one-hot vectors, which is computationally expensive. Previous research~\cite{DBLP:journals/pacmpl/ShaikhhaFVJ19} has shown how this can be optimized by wrapping the vector construction around the forward-mode AD and applying loop optimizations. 

\smartpara{Example 2 (cont.)} Consider the case of scalar-vector product, represented as !s * V! in \lang{}. 
Applying the scalar FAD rules on this program results in:

\begin{lstlisting}
s * V' + s' * V
\end{lstlisting}

\noindent In the case of differentiation with respect to !V!, similar to the dot-product example, one needs to repeatedly pass all different one-hot vectors as !V'!.
However, the differentiation with respect to !s! can be done by setting !s'! to 1 and !V'! to !zero[tensor 1]!.

Next, we show an alternative differentiation transformation that enables native tensorized forward-mode AD.


\subsection{Tensorized Forward-Mode Transformation}
\label{sub:tensorized-fad}

As tensors are first-class citizens in \lang{}, one can directly express the differentiation with respect to tensor variables of type $\tau$, represented by \tengrad{}. This means that the derivative of an expression of type !tensor n! with respect to a tensor of type !tensor m!, will be a !tensor (n + m)!, which is the same as the outer tensor product type (cf. Figure~\ref{fig:diffrules}, the rules for types).

Figure~\ref{fig:diffrules} shows the rules for tensorized FAD.
They generalize the rules for scalar forward-mode AD, which one recovers by setting $\tau$ to be !real!. The key differences are in the rules for constant reals and multiplication. Rather than just returning a real-valued !0!, tensorized FAD returns the zero value of type $\tau$ represented as !zero[$\tau$]!.
For multiplication, if $\tau$ is a tensor type with a non-zero order, the first term still computes the multiplication of !e1! and the differentiation of !e2!. However, the second term requires re-arranging the indices of the tensors. This complication can be avoided by only allowing for the multiplication of real numbers in the input program. This is achieved by applying multiplication normalization (cf. Section~\ref{sec:normalization}).

\begin{figure}
\begin{tabular}{|lcl|}
\hline
 & & \vspace{-0.2cm} \\
\multicolumn{3}{|c|}{
\begin{tabular}{c|c}
\begin{tabular}{l c l}
\multicolumn{3}{c}{\framebox[1.2\width]{\fvadttype{\code{T}}} Tensorized FAD on Types} \\
 & & \vspace{-0.3cm} \\
   \fvadttype{\code{ D }}  & = &
   $\code{D}~~~\outertype~~~\tau$ \\ 
   \fvadttype{\codekw{ bool }}  & = &
   \codekw{real } \\ 
   \fvadttype{\codekw{ int }}  & = &
   \codekw{real } \\ 
\end{tabular} &
\begin{tabular}{l c l}
\multicolumn{3}{c}{\framebox[1.2\width]{\fvadtctx{\Gamma}} Tensorized FAD on Context} \\
 & & \vspace{-0.3cm} \\
   \fvadtctx{\emptyset}  & = &
   $\emptyset$ \\ 
   \fvadtctx{\Gamma, \code{x:T}}  & = &
   \fvadtctx{\Gamma}$, \code{x:T}, \code{x':}$\fvadttype{\code{T}} \\ 
\end{tabular}
\end{tabular}
} \\
& & \vspace{-0.2cm}\\
\multicolumn{3}{|c|}{\framebox[1.2\width]{\fvadttang{\code{e}}} Tensorized FAD on Expressions} \\
& & \vspace{-0.2cm}\\
\multicolumn{3}{|c|}{Invariant: If $\Gamma \vdash$ \code{e: T}, then \fvadtctx{\Gamma} $\vdash$ \fvadttang{\code{e}}\code{:} \fvadttype{\code{T}}} \\
& & \vspace{-0.2cm}\\
   \fvadttang{\code{ \codekw{sum}(<k,v> \codekw{in} e1) e2 }}   & = & \code{\codekw{sum}(<k,v> \codekw{in} \fvadtprime{\code{e1}})} \\
   & & \code{ \codekw{let} <k', v'> = <0, \fvadttang{\code{e1(k)}}> \codekw{in} } \\
   & & \code{ \fvadttang{\code{e2}}} \\
\fvadttang{\code{ \codekw{let} x = e1 \codekw{in} e2 }}  & = &
   \code{\codekw{let} <x, x'> = <e1, \fvadttang{\code{e1}}> \codekw{in}}
   \\ 
   & & \code{\fvadttang{\code{e2}}
   } \\
   \fvadttang{\code{ \codekw{if} e1 \codekw{then} e2 }}  & = &
   \code{\codekw{if} \fvadtprime{\code{e1}} \codekw{then} \fvadttang{\code{e2}}
   } \\
   \fvadttang{\code{ e1(e2) }}  & = &
   \fvadttang{\code{e1}}\code{(}\fvadtprime{\code{e2}}\code{)} \\
   \fvadttang{\code{ \{ e1 -> e2 \} }}  & = &
   \code{\{ \fvadtprime{\code{e1}} -> \fvadttang{\code{e2}} \}} \\ 
   \fvadttang{\code{ e1 * e2 }}  & = &
   \code{\fvadtprime{\code{e1}} * \fvadttang{\code{e2}} +
   \fvadttang{\code{e1}} \tmult{$\tau$} \fvadtprime{\code{e2}}
   } \\
   \fvadttang{\code{ e1 + e2 }}  & = &
   \code{\fvadttang{\code{e1}} + \fvadttang{\code{e2}}
   } \\
   \fvadttang{\code{ op(e) }}  & = &
   \code{op'(e) * }\fvadttang{\code{e}} \\
   \fvadttang{\code{ x }}  & = &
   \code{x'} \\
   \fvadttang{\code{ r }}  & = &
   \zeromacro{$\tau$} \\
   \fvadttang{\code{ n }} = \fvadttang{\codekw{false}} = \fvadttang{\codekw{true}} & = &
   \code{0} \\ && \\
\code{e1\tmult{\codekw{real}} e2} & $\defines$ & !e1 * e2! \\
\code{e1\tmult{\codemacro{tensor} n} e2} & $\defines$ & !sum(<i_1,r_2> in e1) ...! \\
!e1: tensor (m + n)!& & !  sum(<i_m,v> in r_m)) !\\
& & !    {i_1 -> ... { i_m -> 1 } ... } * e2 * v! \\

\hline
\end{tabular}
\caption{Tensorized Forward-mode Automatic Differentiation (FAD) rules for \lang expressions. The type $\tau$ needs to be a tensor type, i.e., it follows the grammar of \code{D}. The rules of Figure~\ref{fig:forward_diffrules} are the special case of \fadttang{}=\fvadtgen{\codekw{real}}{}. }
\label{fig:diffrules}
\end{figure}

\smartpara{Example 1 (cont.)} In our running example, tensorized FAD transformation is represented as
\tengradbegin{}!sum(<i,a> in V2) V1(i) * a!\tengradend{}. 

By applying the rules in Figure~\ref{fig:diffrules} we have:

\begin{lstlisting}
sum(<i,a> in V2) 
  let <i',a'> = <0,V2'(i)>
  V1(i) * a' + V1'(i) * a
\end{lstlisting}

\noindent As one can see, this program looks identical to the version generated by scalar FAD. The difference is in the type of !V1'! and !V2'!. In scalar FAD, their type is the same as !V1! and !V2!, i.e., !tensor 1!. 
In tensorized FAD, their type is !tensor 2!. Thus, the multiplications in the last expression are now scalar-vector multiplications.

As opposed to scalar FAD, we need to assemble the zero and one-hot vectors of all iterations together; instead of passing them vector by vector, we pass them as an entire matrix in which each row represents one of the one-hot vectors. The definition of !onehot[tensor 1]! in Figure~\ref{fig:api} specifies how one can build such a matrix for variable !x!.

\smartpara{Example 3} Consider the case of computing the trace of the matrix !M!. The tensorized FAD over this expression in \lang{} is represented as 
\tengradbegin{} !sum(<i,r> in M) r(i)!\tengradend{}.


\noindent For each row !r! of matrix !M! at index !i!, we compute the summation of diagonal elements specified by !r(i)!. After our tensorized FAD transformation, we obtain the following program:

\begin{lstlisting}
sum(<i,r> in M) 
  let <i',r'> = <0,M'(i)>
  r'(i)
\end{lstlisting}

\noindent To compute the one-hot input for a matrix, we need to generalize the case of a vector. In the case of differentiating with respect to a vector (!tensor 1!), we were passing a matrix (!tensor 2!) for the one-hot inputs. Here, in the case of differentiating w.r.t. a  matrix (!tensor 2!), we need to pass an order-4 tensor (!tensor 4!). The definition of !onehot[tensor 2]! can also be found in Figure~\ref{fig:api}.


\begin{figure}[t]
\begin{tabular}{|r c l|}
\hline
!gradient e (x: $\tau$)! & $\defines$ & !let <v1',..., vn'> = <ingrad v1 x,...,ingrad vn x> in! \\
& & \tengrad{\code{e}} \\
\textbf{where} & & $\{$!v1, ..., vn!$\} = FVS($!e!$)$\\
!ingrad (v: D) x! & $\defines$ & !zero[D]! \\
!ingrad (x: D) x! & $\defines$ & !onehot[D] x! \\
!onehot[real] x! & $\defines$ & !1! \\
!onehot[tensor 1] x! & $\defines$ & !sum(<i,_> in x) { i -> {i -> 1} } ! \\
!onehot[tensor 2] x! & $\defines$ & !sum(<i,v> in x) { i ->! \\
& & !  sum(<j,_> in v) {j ->! \\
& & !    {i -> {j -> 1} } } } ! \\
!onehot[tensor n] x! & $\defines$ & !sum(<i_1,v_2> in x) { i_1 -> ! \\
if !n > 2!& & !  sum(<i_2,v_3> in v_2) { i_2 ->! \\
& & !   ... sum(<i_n,_> in v_n)) { i_n ->!\\
& & !         {i_1 -> ... { i_n -> 1 } ... }! \\ \hline
\end{tabular}
\caption{The \codemacro{gradient} API exposed by \system{} to the programmer.}
\label{fig:api}
\end{figure}

\subsection{Putting it All Together}
Programmers who use machine learning frameworks need not concern themselves with the source-to-source transformations employed in the background. To accomplish this, \system{} provides a high-level API through the !gradient! macro, which accepts two inputs. The first input is the \lang{} expression to be differentiated, while the second input is the free variable with respect to which we perform the differentiation.

Consider the case of computing the gradient of the expression !e! with respect to !x!. This is represented as !gradient e x!.
The !gradient! macro generates let-bindings for the differentiation components of all free variables of the expression !e! ($FVS(\code{e})$). 
The RHS of let-binding for a free variable !v! is !ingrad v x!. The !ingrad! macro is responsible for computing the input zeros or one-hots, depending on whether variable !v! is different than !x! or is the same.
In the former case, the !zero! macro is used by passing the type of !v!. In the latter case, the !onehot! macro is used. For an input variable of type !tensor n!, the !onehot! assembles a sparse !tensor (2*n)! where the !n! indices of the input variable !x! are repeated twice so that the diagonals of the hypercube are set to !1!.

\smartpara{Example 1 (cont.)} Our running example is specified by the high-level API as follows:

\begin{lstlisting}
gradient (sum(<i, a> in V1) a * V2(i)) V2
\end{lstlisting}

\noindent This is expanded to the following expression:

\begin{lstlisting}
let V1' = zero[{int -> real}]
let V2' = onehot[{int -> real}] V2
sum(<i, a> in V1)
  let <i', a'> = <{}, V1'(i)> in
  a * V2'(i) + a' * V2(i)
\end{lstlisting}

\noindent Which is, in turn, expanded to:

\begin{lstlisting}
let V1' = {}
let V2' = sum(<i,_> in V2) {i -> {i -> 1}}
sum(<i, a> in V1)
  let <i', a'> = <{}, V1'(i)> in
  a * V2'(i) + a' * V2(i)
\end{lstlisting}

\noindent
In the next section, we see how this expression will be optimized. Note that by simple constant propagation, the asymptotic complexity will already be $O(n)$ on this example. In fact, we will only perform general optimizations which are not AD-specific.



\section{Performance}
\label{sec:opt}
In this section, we review the techniques we employ to improve performance. While Figure~\ref{fig:archhigh} showed the simplified  overview of our approach, Figure~\ref{fig:arch} presents in more detailed the series of transformations applied in \system{}. 

\begin{figure}
    \centering
    \includegraphics[width=0.9\linewidth]{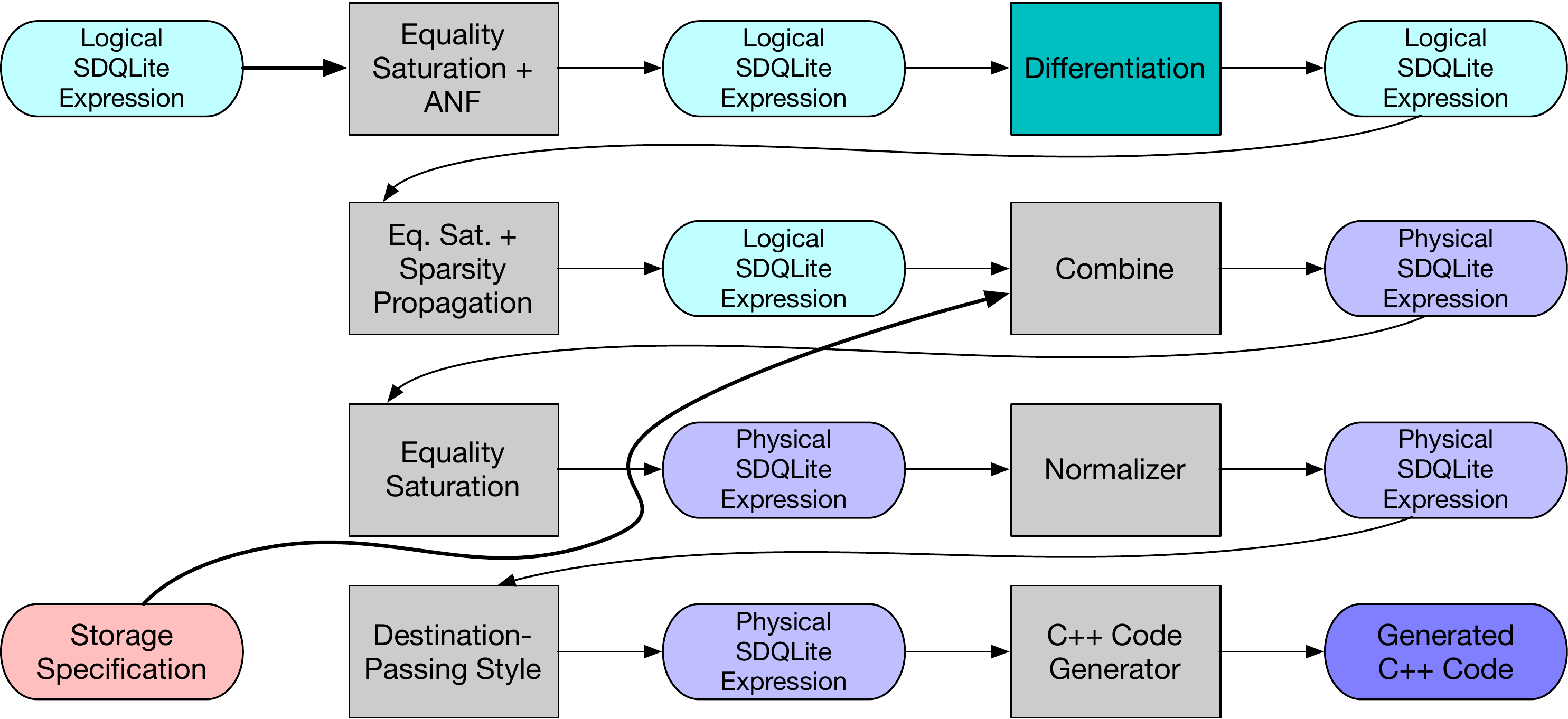}
    \caption{The transformations in \system{}.}
    \label{fig:arch}
\end{figure}

\subsection{Pre-differentiation}
\smartpara{Equality Saturation}
First,  \system{} applies algebraic rewrite rules over the input program to leverage optimizations such as factorization, loop fusion, etc~\cite{schleich2022optimizing}. We use equality saturation in order not to worry about phase ordering problems and making sure that the rewrite rules are applied globally~\cite{tate2009equality}. 
We use EGG~\cite{DBLP:journals/pacmpl/WillseyNWFTP21}, a state-of-the-art implementation of equality saturation that has been successfully used for \lang{} in the context of flexible storage specification for tensor programs~\cite{schleich2022optimizing}. We rely on the cost models and the algebraic rewrite rules specified in~\cite{schleich2022optimizing}; there is no need to specify any AD-specific cost model or rewrite rule.

\smartpara{ANF}
Then, \system{} applies an ANF transformation~\cite{DBLP:conf/pldi/FlanaganSDF93}, which ensures that sub-expressions are simple expressions, i.e. constant values or variable references. This is achieved by using a let-binding for the sub-expressions if they are not already simple expressions. The ANF transformed program is then fed into the differentiation transformation presented in Section~\ref{sec:diff}.

\begin{figure}[t]
\begin{tabular}{|r c l|l|}
\hline
\multicolumn{3}{|c|}{Rewrite Rule} & Condition \\
\hline
!zero[D1] * e2! & $\leadsto$ & !zero[D1! \outertype{} !D2]! & if !e2: D2!\\ \hline
!e1 * zero[D2]! & $\leadsto$ & !zero[D1! \outertype{} !D2]! & if !e1: D1!\\ \hline
!e1 + zero[D1]! & $\leadsto$ & !e1! & must have !e1: D1!\\ \hline
!zero[D1] + e2! & $\leadsto$ & !e2! & must have !e2: D1!\\ \hline
!zero[tensor n](e2)! & $\leadsto$ & !zero[tensor (n-1)]! & must have !n > 0!\\ \hline
!let x=zero[D] in e2! & $\leadsto$ & !e2[x$\rightarrow$zero[D]]! & \\ \hline
!sum(<k,v> in zero[D])e2!& $\leadsto$ & !zero[D2]!& if !e2:D2! \\ 
\hline \hline
!e1 * e2! & $\leadsto$ & !sum(<i_1,v_2> in e1) ... ! & if !e1: tensor n!\\
& & !  sum(<i_n,v> in v_n)) ! & and !n > 0! \\
& & !    {i_1->...{i_n-> v*e1}...}! & \\ \hline
!e1 * e2! & $\leadsto$ & !sum(<i_1,v_2> in e2) ... ! & if !e1: real!\\
& & !  sum(<i_n,v> in v_n)) ! & and !e2: tensor n! \\
& & !    {i_1->...{i_n-> e1*v}...}! & and !n > 0! \\ \hline 
\end{tabular}
\caption{Transformations used in \system{}. The first group of rewrite rules are related to sparsity propagation. The second group correspond to multiplication normalization.}
\label{fig:trans}
\end{figure}

\subsection{Post-Differentiation}
\smartpara{Sparsity Propagation}
After differentiation, many intermediate highly sparse values (e.g., zero tensors) are constructed. Even though expressing local rewrite rules for simplifying them is possible, these programs are very large and optimizing them with equality saturation requires a large search space. Thus, we applied these sparsity propagation rules as a separate pass. The first group of rewrite rules in Figure~\ref{fig:trans} correspond to the sparsity propagation rules. Finally, we again pass the program to equality saturation.

\smartpara{Example 1 (cont.)} In the dot-product example, by assuming the second variant (iteration over the second vector), the differentiated program with respect to the first vector is as follows:

\begin{lstlisting}
let V1' = sum(<i, a> in V1) {i -> {i -> 1}} 
let V2' = {}
sum(<i, a> in V2)
  let <i', a'> = <{}, V2'(i)> in
  V1(i) * a' + V1'(i) * a
\end{lstlisting}

\noindent The sparsity propagation optimization propagates !V2'! and simplifies the relevant expressions (e.g., !V1(i) * a'!) as follows:

\begin{lstlisting}
let V1' = sum(<i, a> in V1) {i -> {i -> 1}}
sum(<i, a> in V2)
  V1'(i) * a
\end{lstlisting}

\noindent After applying equality saturation, if !V1! is a dense vector, the program is optimized and reduces to simply !V2!, a slightly less optimized version of which is:

\begin{lstlisting}
sum(<i, a> in V2)
  { i -> a }
\end{lstlisting}





\begin{figure}
    \centering
    \begin{tabular}{|l|l|} \hline
    \specialcell{\textbf{Storage Format}\\ \textbf{(Inputs)}} & \textbf{Physical \lang{} Definition} \\ \hline
        \specialcell{Vector COO\\\code{(len,row,val)}} & 
        \begin{lstlisting}
sum(<_, i> in (0:len)) 
  { unique(row(i)) -> val(i) }
        \end{lstlisting}
        \\ \hline
        \specialcell{Vector Dense\\\code{(len,V)}} &         \begin{lstlisting}
sum(<_, i> in (0:len)) 
  { unique(i) -> V(i) }
        \end{lstlisting} \\ \hline
        \specialcell{Matrix CSR\\\code{(len,pos,idx,val)}} & \begin{lstlisting}
sum(<_, i> in (0:len)) 
  { unique(i) -> 
      sum(<p, j> in idx(pos(i):pos(i+1))) 
        { unique(j) -> val(p) }  }
        \end{lstlisting} \\ \hline
        \specialcell{Matrix Dense Row Major \\\code{(rows,cols,M)}}  & 
\begin{lstlisting}
sum(<_, i> in (0:rows)) 
  { unique(i) -> 
      sum(<_, j> in (0:cols)) 
        { unique(j) -> M(i*cols+j) }  }
\end{lstlisting}
        \\ \hline
    \end{tabular}
    \caption{The functional implementation of various sparse/dense formats in Physical \lang{}. Matrix CSC and Dense Column Major are expressed similarly to Matrix CSR and Dense Row Major, respectively.}
    \label{fig:storage_spec}
\end{figure}

\subsection{Storage Composition} 
\label{sub:storage-comp}

After optimizing the differentiated expression, \system{} uses the storage specifications as the definition for each input tensor. 
Figure~\ref{fig:storage_spec} shows the specification of several storage formats in \lang{}. 
\system{} generates let-bindings that bind every input tensor to the expression specifying its physical storage format.
However, this makes the performance even worse, due to unnecessary intermediate tensor allocations. Luckily, previous research~\cite{schleich2022optimizing} showed how equality saturation can not only remove these intermediate tensors but also recover well-known algorithms for sparse tensor computations.

\smartpara{Example 1 (cont.)} In the previous example, if we assume that !V1! is a dense vector and !V2! is a sparse vector using the COO representation, combining the storage specification with the differentiated program results in:

\begin{lstlisting}
let V1 = sum(<_, i> in (0:V1_len)) { unique(i) -> V1_V(i) } in
let V2 = sum(<_, i> in (0:V2_len)) { unique(V2_row(i)) -> V2_val(i) } in
sum(<i, a> in V2)
  { i -> a }
\end{lstlisting}

\noindent By applying equality saturation, \system{} returns the following program that does not introduce unnecessary intermediate tensors, as expected from \cite{schleich2022optimizing}:

\begin{lstlisting}
sum(<_, i> in (0:V2_len))
  { V2_row(i) -> V2_val(i) }
\end{lstlisting}

\subsection{Normalization}
\label{sec:normalization}
\system{} performs additional lower-level transformations on the optimized storage-format-aware program. First, multiplication normalization rewrites tensor outer products into summation expressions with scalar multiplications. The rewrite rules for multiplication normalization are shown in the second group of rules of Figure~\ref{fig:trans}. Once combined with loop fusion rules, multiplication normalization removes unnecessary intermediate tensors. The second normalization involves ANF transformations which have been previously cancelled by equality saturation and other transformations. 

\smartpara{Example 5 - BATAX} The BATAX kernel~\cite{DBLP:journals/toms/NelsonBSJN15} is represented as $\beta A^TAx$ where $\beta$ is a scalar value, $A$ a matrix, and $x$ a vector. This kernel is expressed in \lang as follows:

\begin{lstlisting}
sum(<i, r> in A)
  sum(<j, v1> in r)
    sum(<k, v2> in r)
      { j -> ((beta * v1) * v2) * (x(k)) }
\end{lstlisting}

\noindent After differentiating with respect to !x!, post-differentiation optimizations, \system produces:

\begin{lstlisting}
beta * (sum(<i, r> in A) r * r)
\end{lstlisting}

\noindent Considering a CSR representation for !A! and a dense representation for !x!, after applying storage composition and equality saturation we derive:

\begin{lstlisting}
beta * sum(<_, i> in (0:A_len))
  let r = sum(<p, j> in A_VCol(A_VRow(i):A_VRow((i + 1))))
    { j -> A_VVal(p) } in 
  r * r
\end{lstlisting}

\noindent The multiplication of !beta! with the result of summation is a scalar-vector product. Also, 
the last expression !r * r! corresponds to a vector outer product. The multiplication normalization rewrites both expressions,  the result of which is as follows:

\begin{lstlisting}
sum(<_, i> in (0:A_len))
  let r = sum(<p, j> in A_VCol(A_VRow(i):A_VRow((i + 1))))
    { j -> A_VVal(p) } in 
  sum(<i1, v1> in r) { i1 -> 
    sum(<i2, v2> in r)
      { i2 -> beta * v1 * v2 } }
\end{lstlisting}

\subsection{Destination-Passing Style \& C++ Code Generation}
\label{sub:dps-and-cpp}
As final step, \system{} generates C++ code. The index and scalar types !int! and !real! translate into !size_t! and !double!, respectively. The expressions of type !{ int -> real }! are translated to objects of type !dict_type<size_t, double>!. Nested dictionaries are recursively translated, e.g., !{ int -> { int -> real } }! $\leadsto$ !dict_type<size_t, dict_type<size_t, double>>!. We use the robinhood dictionary~\cite{rbdict} for the C++ runtime. Dictionaries with a key type !dense_int! are translated to !arr_type! instead of !dict_type!. For example, !{ dense_int -> real }! translates into !arr_type<double>! and !{ dense_int -> int }! into !arr_type<size_t>!.

The C++ code generation for most constructs of \lang{} is straightforward. The construct !dict(key)! is translated into !dict[key]!, which corresponds to a dictionary or array lookup in C++. The addition and multiplication constructs are converted to the same primitives in C++. This requires the support for !+! and !*! over dictionary types, which is provided by our C++ runtime.
The !sum! construct is translated into a for-loop. Similarly, nested summations are translated into nested for-loops. If the range expression of a summation is the range construction !(st:en)! or sub-array !arr(st:en)!, the translation produces a standard for-loop. However, if it is a dictionary, the translation generates a for-each construct over the dictionary.
Another interesting case is the code generation for the singleton dictionary construct. This construct is mostly used inside a summation. In this case, it is translated to a dictionary update as follows:

\begin{tabular}{l c l}
\begin{lstlisting}
let res = 
  sum(<k,v> in dict)
    { f(k,v) -> g(k,v) }
\end{lstlisting} & $\leadsto$ &
\begin{footnotesize}
\begin{lstlisting}[language=c++]
dict_type<size_t, double> res;
for(auto& kv : dict) { 
  res[f(kv.first, kv.second)] += g(kv.first, kv.second);
}
\end{lstlisting}
\end{footnotesize}
\end{tabular}

\noindent Note that the above translation works if !g(k,v)! outputs a dictionary. This is because !+=! is also overloaded by dictionaries.

\smartpara{Example 1 (cont.)} The generated C++ code for the differentiated program is as follows:

\begin{footnotesize}
\begin{lstlisting}[language=c++]
void VVD_wrt_V1(
  // inputs for the dense vector V1
  arr_type<double>& V1_V, size_t V1_len, 
  // inputs for the sparse vector V2
  arr_type<size_t>& V2_VRow, arr_type<double>& V2_VVal, size_t V2_len, 
  // destination dictionary
  dict_type<size_t, double>& result) {
  for(size_t i = 0; i < V2_len; i++) {
    result[V2_VRow[i]] += V2_VVal[i];
} }
\end{lstlisting}
\end{footnotesize}

In order to generate efficient C++, \system{} leverages the destination-passing style (DPS) technique~\cite{minamide1998functional} in two ways. 
First, the generated function is provided with the destination object to store the final results~\cite{shaikhha2017destination}.
Second, 
The destination-passing style transformation removes intermediate tensors created in inner loops. This is achieved by pushing all the singleton dictionary constructions into the inner loops to avoid unnecessary intermediate dictionary constructions~\cite{shaikhha2022deep}.

\smartpara{Example 4 (cont.)} In the BATAX kernel, the C++ code generation produces the following program:

\noindent
\begin{footnotesize}
\begin{lstlisting}[language=c++]
void BATAX_stg_wrt_X(
  // input scalar
  double& beta_S,
  // input for the CSR matrix A
  arr_type<size_t>& A_VRow, arr_type<size_t>& A_VCol, arr_type<double>& A_VVal, size_t A_len, 
  // input for the dense vector X
  arr_type<double>& X_V, size_t X_len, 
  // destination dictionary
  dict_type<size_t, dict_type<size_t, double>>& result) {
  for(size_t i = 0; i < A_len; i++) {
    dict_type<size_t, double> r;
    for(size_t p = A_VRow[i]; p < A_VRow[(i + 1)]; p++) {
      size_t j = A_VCol[p];
      r[j] += A_VVal[p];
    }
\end{lstlisting}
\begin{lstlisting}[language=c++,backgroundcolor=\color{red!10}]
    dict_type<size_t, dict_type<size_t, double>> tmp1;
    for(auto& iv1 : r) {
      dict_type<size_t, double> tmp2;
      for(auto& iv2 : r) {
        tmp2[iv2.first] += beta_S * (iv1.second * iv2.second);
      }
      tmp1[iv1.first] += tmp2;
    }
    result += tmp1;
\end{lstlisting}
\begin{lstlisting}[language=c++]
} }
\end{lstlisting}
\end{footnotesize}

\noindent Note that the last inner loop (highlighted) requires constructing two intermediate dictionaries !tmp1! and !tmp2!. One can use the associativity of semi-ring dictionaries in order to reorder the insertions so that all the insertions happen in the most-inner loop directly into !result!: 

\begin{footnotesize}
\begin{lstlisting}[language=c++,backgroundcolor=\color{blue!10}]
    for(auto& iv1 : r) {
      for(auto& iv2 : r) {
        result[iv1.first][iv2.first] += beta_S * (iv1.second * iv2.second);
    } }
\end{lstlisting}
\end{footnotesize}

\noindent This removes the need to allocate the intermediate dictionaries !tmp1! and !tmp2!. We will see the impact of this optimization in Section~\ref{sec:exp}.

\section{Proofs}
\label{sec:semantics}

In this section, first, we prove the type soundness of the automatic differentiation. Then, we show that programs in Logical \lang denote differentiable functions for a generalization of differentiability that conservatively extends the one between Euclidean spaces. Finally, we show the correctness of both scalar and tensorized forward-mode AD.

\subsection{Type Soundness of Differentiation}

\begin{lemma}[Well-typedness of \fvadttang{}]
\label{lem:well_typedness_D}
	If $\Gamma\vdash $ \code{e:T} then \fvadtctx{\Gamma}$\vdash$ \fvadttang{\code{e}}: \fvadttype{\code{T}}.
\end{lemma}
\begin{proof}
	By induction on the typing tree of \code{e} and case analysis. The interesting cases are:

\noindent \textbf{Case \code{*}:} Let us consider !e1 * e2: tensor m! and $\tau=$ !tensor n!. We need to prove that \tengrad{\code{e1 * e2}}!: tensor (m+n)!. Based on the induction hypothesis and typing rules, we have the following:
\begin{itemize}
    \item !e1: tensor m1! $\Rightarrow$ \tengrad{\code{e1}}!: tensor (m1+n)!
    \item !e2: tensor m2! $\Rightarrow$ \tengrad{\code{e2}}!: tensor (m2+n)!
    \item !e1: tensor m1, e2: tensor m2! $\Rightarrow$ !e1 * e2: tensor (m1+m2)! $\Rightarrow$ !m1+m2 = m!
\end{itemize}

By considering the sub-expressions of \tengrad{\code{e1 * e2}} we have:

\begin{itemize}
    \item !e1 * !\tengrad{\code{e2}}!: tensor (m1+(m2+n))! $\Rightarrow$ !e1 * !\tengrad{\code{e2}}!: tensor ((m1+m2)+n)!
    \item \tengrad{\code{e1}} \tmult{$\tau$} !e2: tensor ((m1+m2)+n)!
    
\end{itemize}

Thus, the addition of these two sub-expressions also has the same type, i.e., !tensor ((m1+m2)+n)!. As !m1+m2 = m!, the output type is !tensor (m+n)!.

\noindent \textbf{Case \codekw{sum}:} Let us consider $\Gamma\vdash$ !sum(<k, v> in e1) e2: tensor m! and $\tau=$ !tensor n!.
We need to prove that \tengrad{\Gamma} $\vdash$
\tengrad{\codekw{sum}\code{(<k, v> in e1) e2}}!: tensor (m+n)!. Based on the induction hypothesis and typing rules, we have the following:
\begin{itemize}
    \item $\Gamma\vdash$ !e1: tensor (m1+1)! $\Rightarrow$ \tengrad{\Gamma} $\vdash$ \tengrad{\code{e1}}!: tensor (m1+1+n)!
    \item $\Gamma'=\Gamma,$!k:int!$,$!v:tensor m1!
    \item $\Gamma''=$ \tengrad{\Gamma}$,$!k:int!$,$!v:tensor m1!$,$!k':real!$,$!v':tensor (n+m1)!
    \item $\Gamma'\vdash$ !e2: tensor m! $\Rightarrow$
     $\Gamma''\vdash$ \tengrad{\code{e2}}!: tensor (m+n)!
\end{itemize}

Let us name the body of !sum! in its differentiation as !e'!:

\begin{itemize}
    \item \tengrad{\codekw{sum}\code{(<k, v> in e1) e2}} = \\
    !sum(<k,v> in e1) let <k',v'> = <0,!\tengrad{\code{e1}}!(k)> in !\tengrad{e2} = \\
    !sum(<k,v> in e1) e'! $\Rightarrow$ \\
    !e'! = !let <k',v'> = <0,!\tengrad{\code{e1}}!(k)> in !\tengrad{e2}
\end{itemize}



By the induction hypothesis, we know that $\Gamma'\vdash$ !e': tensor (m+n)!.




As a result, we derive $\Gamma\vdash$ !sum(<k,v> in e1) e': tensor (m+n)!, which we needed to prove.
\end{proof}

The rest of this section goes into more mathematical depth and can be skipped by non-expert readers. Its aim is to formalize the argument that tensorized FAD computes gradients.
We first interpret programs of Logical \lang{} as functions and define what it means to compute derivatives and gradients for such functions. We then extend the definition of differentiability and differentials beyond the usual one for functions between Cartesian spaces. We then use a logical relations argument to show that \fvadttang{\code{e}} denotes the differential of !e! and therefore produces correct gradients.

\subsection{Semantics and Differentiability}

\smartpara{Non-smooth semantics}
We can give a denotational semantics to Logical \lang{} as a special case of the one given in \cite{DBLP:journals/pacmpl/ShaikhhaHSO22}. The type !real! is interpreted as reals $\RR$, and programs !x: real! $\vdash$ !e: real! as functions $\RR\to\RR$. We use semantic brackets $\sem{-}$ for the denotational semantics of types and terms.
We interpret !bool! as the booleans $\BB$, and !int! as the natural numbers $\NN$.
The semantics of !tensor n! is defined inductively, with $\llbracket$!{ int -> D }!$\rrbracket$ given by the tensor product of the free $\RR$ vector space on the countable set $\NN=\llbracket$!int!$\rrbracket$ with the real vector space $\llbracket$!D!$\rrbracket$. This space is only defined up to isomorphism of real vector spaces, and we are free to choose $\llbracket$!tensor n! $\rrbracket$ to be the free real vector space on !I!$\times ... \times $!I!, the $n$ fold Cartesian product of !I! with itself. 
By interpreting contexts 
$\sem{\Gamma} = \{$!x1: T1!, $\ldots$, !xn:Tn!$\}$ as $\sem{\Gamma}= \llbracket$!T1!$\rrbracket\times\ldots \times\llbracket$!Tn!$\rrbracket$, programs $\Gamma \vdash$ !e: T! denote functions $\llbracket$!e!$\rrbracket:\sem{\Gamma}\to \llbracket$!T!$\rrbracket$.

The correctness of the rewrite rules from Figure~\ref{fig:trans} is immediate, as they can all be seen as special cases of optimizations given in related work \citep{DBLP:journals/pacmpl/ShaikhhaHSO22,schleich2022optimizing}, and they rely on the semi-ring structure of our dictionaries.

 \smartpara{Smooth semantics}
The limitation of the non-smooth semantics is that it does not inform us of the differentiability of our programs. A function $f:\RR^n\to \RR^m$ is smooth if it is differentiable and its derivative is smooth.
We can refine our semantics of Logical \lang{} by interpreting our types as generalizations of Cartesian spaces, called diffeological spaces \cite{iglesias2013diffeology}, and our functions as morphisms of diffeological spaces, which can be thought of as generalized smooth functions between these spaces, closely following~\cite{huot2020correctness}.
More precisely, a diffeological space is a pair $(X,\mathcal{P}^X)$ of a set $X$ and sets of functions $\mathcal{P}_U^X\subseteq X^U$ for all subsets $U$ of Euclidean spaces $\RR^n$, called plots. Each such set of plots $\mathcal{P}_U^X$ must contain all constant functions, be closed under pre-composition by smooth functions, and be closed under countable glueing. The latter means that if $(f_i: U_i \subseteq \RR^n\to X)_{i\in\NN}$ are a countable family of plots such that for each $i,j\in \NN$, $f_i$ and $f_j$ agree on $U_i\cap U_j$, then $f:\bigcup_i U_i \to X, x\in U_i\mapsto f_i(x)$ is a plot. Morphisms $f:(X,\mathcal{P}^X)\to (Y,\mathcal{P}^Y)$ between diffeological spaces are functions $f:X\to Y$ such that for all $p\in \mathcal{P}_U^X$, $f\circ p\in  \mathcal{P}_U^Y$. $\NN$ or $\BB$ with constant functions as plots form diffeological spaces. $\RR$ with smooth functions $U\to \RR$ as plots $\mathcal{P}_U^\RR$, is a diffeological space. Morphisms of diffeological spaces $(\RR,\mathcal{P}^\RR)\to(\RR,\mathcal{P}^\RR)$ are exactly smooth functions $\RR\to\RR$.
Diffeological spaces are closed under products, where a plot $f:U\to  X\times Y$ is given by a pair of plots $(f_1,f_2)$, with $f_1\in \mathcal{P}_U^X, f_2\in  \mathcal{P}_U^Y$. We can now revise our denotational semantics by interpreting types as diffeological spaces and terms as morphisms.
For the interpretation of types, the only remaining case is that of dictionary types. The underlying set is the same as the set interpretation, and the plots are given by 
\[\mathcal{P}_U^{\sem{\{I \to D\}}}= \{f:U\to \sem{\{I \to D\}}~\mid~\exists g\in \mathcal{P}_U^{\sem{D}^n},i:\sem{D}^n\to  \sem{\{I \to D\}} \text{ linear injective}, f= i \circ g\}\]
One can easily show that all the semantics of all the constructs of the language preserve plots. 
As a corollary, we obtain the following result:

\begin{proposition}
Well-typed programs !x1: real!, $\ldots$, !xn: real! $\vdash$ !e: real! denote smooth functions $\RR^n\to \RR$. 
\end{proposition}

\smartpara{Differentials}
For functions $f:\RR \to \RR$, the differential of a differentiable function $f$ is a function $df: \RR \times \RR \to \RR$ such that for any $x\in\RR$, $df(x,-):\RR\to \RR$ is the best linear approximation of $f$ at $x$, which is characterized by its slope $f'(x)$. Thus, $df(x,v)= f'(x) \cdot v$.
More generally, the differential of a differentiable function $f:\RR^n\to \RR^m$ is a function $df: \RR^n \times \RR^n\to \RR^m$ such that $df(x,v) = J_xf \cdot v$, where $\cdot$ is now a matrix-vector product. $J_xf$ is the Jacobian of $f$ at each, i.e. the $n\times m$ matrix of partial derivatives of $f$ at $x$: $J_xf = \Big(\frac{\partial}{\partial x_j}f_i(x)\Big)_{1\leq j\leq n, 1\leq i\leq m}$, where $f_i$ is the projection of $f$ on its $i$-th component. 

\subsection{Correctness of Scalar Forward-Mode}

\smartpara{Correctness on a first-order sublanguage} 
Consider the simple first-order sublanguage of Logical \lang{} that does not involve dictionaries. On this sublanguage, semantically, the forward-mode algorithm sends $f$ to its differential $df$, up to a reordering of the inputs. More precisely, if 
 !x1: real!, $\ldots$, !xn: real! $\vdash$ !e: real!, and
  !x1: real!, !x1': real!, $\ldots$, !xn: real!, !xn': real! $\vdash$ \fadttang{\code{e}}!: real!, then $\llbracket$ \fadttang{\code{e}}$\rrbracket(x_1,x_1',\ldots,x_n,x_n') = J_{(x_1,\ldots, x_n)} \llbracket$!e!$\rrbracket \cdot (x_1',\ldots, x_n')^T$, where $(-)^T$ is transposition. This can be proved by routine induction on the structure of terms.

\smartpara{Beyond differentials} Because of tensors, the semantics of our programs do not represent functions $\RR^n\to \RR^m$, and thus lack a definition of differentials. Even if we were only interested in programs of the form !x1: real!, $\ldots$, !xn: real! $\vdash$ !e: real!, !e! may use tensors internally, and a correctness proof by simple induction would fail. Thus, we need to define differentials of functions whose input or outputs can be tensors. We call $\pi:\llbracket$!tensor m!$\rrbracket\to \RR^k$ a canonical projection if it is a projection onto the free-vector space $\RR^k$ generated by a subset of size $k$ of !I!. The adjoint linear maps $i:\RR^k\to\llbracket$!tensor m!$\rrbracket$ are called canonical injections. A canonical injection embodies the intuition of tensors with a fixed shape, i.e., a guarantee on their size if they were to be stored as dense tensors. When we differentiate with respect to a tensor, we implicitly mean for a tensor of fixed shape. Likewise, when a program outputs a tensor in our language, its size cannot depend on real inputs, as we do not have non-continuous primitives from the reals to integers. Therefore, if all the integers inputs of a program are fixed, its output tensor will have fixed shape, which can be captured by a canonical projection. Another way to say it is that for programs !x: tensor m! $\vdash$ !e: tensor n!, there is a canonical projection that is injective on the image of $\llbracket$!e!$\rrbracket$. Therefore, we define differentials for functions outputting tensors as follows:

\begin{definition}
The function $dg: \RR^n \times \RR^n \to  \llbracket$!tensor m!$\rrbracket$ is the differential of $f:\RR^n\to \llbracket$!tensor m!$\rrbracket$ if, for all canonical projections $\pi:\llbracket$!tensor m!$\rrbracket\to \RR^k$, $\pi \circ dg$ is the differential of $\pi \circ f$.
\end{definition}

We can extend this definition to functions taking a tensor as input:

\begin{definition}
The function $df: \llbracket$!tensor n!$\rrbracket \times \llbracket$!tensor n!$\rrbracket \to  \llbracket$!tensor m!$\rrbracket$ is the differential of $f: \llbracket$!tensor n!$\rrbracket\to \llbracket$!tensor m!$\rrbracket$ if, for all canonical injections $i:\RR^k \to\llbracket$!tensor n!$\rrbracket$,  $dg\circ (i\times i)$ is the differential of $f\circ i$, where $i \times i: (x,y)\mapsto (i(x),i(y))$.
\end{definition}

Similarly to the case of usual differentiability, we can extend the definition to functions with multiple inputs and outputs in a pointwise style. That is: 
\begin{definition}
\label{defn:extension-differentials}
A function $dg:\sem{T}\times \sem{T}\to \sem{T1} \times \sem{T2} $ is the differential of a function $f: \sem{T} \to \sem{T1}\times \sem{T2}$ iff $\pi_i \circ dg$ is the differential of $\pi_i \circ f$. A function $dg:\sem{T1}\times \sem{T1}\times \sem{T2}\times \sem{T2}$ is the differential of a function $f: \sem{T1}\times \sem{T2} \to \sem{T}$ iff for all $x\in \sem{T1}$, $dg(x,0,-,-)$
is the differential of $f(x,-): \sem{T2} \to \sem{T}$, and for all $y\in \sem{T2}$, $dg(-,-,y,0)$ is the differential of $f(-,y): \sem{T1} \to \sem{T}$.
\end{definition}

We can now state that FAD computes differentials:
\begin{theorem}
\label{thm:simple-fad-correct}
 Let $\Gamma \vdash$! e: D!, where $\Gamma$ only contains tensor types. 
 Then, $\llbracket$\fadttang{\code{e}}$\rrbracket$ is the differential of $\llbracket$!e!$\rrbracket$.
\end{theorem}

The proof now follows a simple logical relations argument, as in \cite{huot2020correctness}. The logical predicates are defined by

\begin{tabular}{rcl}
    $\mathcal{R}_{\codekw{real}}$ & $=$ & $\{(f:\RR\to \RR,g:\RR\to\RR)~\mid~f\in \mathcal{P}_\RR^\RR, g=f'\}$ \\
    $\mathcal{R}_{\codekw{int}}$ & $=$  
    & $\{(f:\RR\to \NN,g:\RR\to\RR)~\mid~f\text{ constant}, g=\lambda x.0\}$ \\
    $\mathcal{R}_{\codekw{bool}}$ & $=$  
    & $\{(f:\RR\to \BB,g:\RR\to\RR)~\mid~f\text{ constant}, g=\lambda x.0\}$\\
    $\mathcal{R}_{\text{\code{\{I -> D\}}}}$ & $=$ & $\big\{(f: \RR\to \llbracket\text{\code{\{I -> D\}}}\rrbracket,g: \RR\to \llbracket\text{\code{\{I -> D\}}}\rrbracket~\mid~f\in\mathcal{P}_\RR^{\sem{\{I \to D\}}}, g=df(-,1), $\\
    && \quad$df$ differential of $f \}$
\end{tabular}

After showing the fundamental lemma for logical relations, where one needs to show that the interpretation of the forward differentiation of every construct of the language preserves the logical predicate (by post-composition), we conclude that the interpretation of every program preserves the logical predicate. We conclude the theorem as a corollary of the fundamental lemma in the same way as in \cite{huot2020correctness}. More precisely, for the predicate $\mathcal{R}_{\codekw{real}}$, we choose $f=id$. For the predicate $\mathcal{R}_{\text{\code{\{I -> D\}}}}$, for 
any linear injection $\RR^k\to\llbracket$\code{\{I -> D\}}$\rrbracket$, and any $x_1,\ldots,x_{k-1}\in\RR^{k-1}$,
we take $f_i = i\circ g_i$ for $g_i=\lambda x.(x_1,\ldots,x,\ldots,x_{k-1})$, where $x$ is at position $i$.


\subsection{Correctness of Tensorized Forward-Mode AD}

\smartpara{A change of perspectives: $\tau$-differentials} Given a differentiable function $f:\RR^n\to \RR^m$, we can equivalently see its differential as a higher-order function $df:\RR^n \times (\RR \multimap \RR^n) \to (\RR \multimap \RR^m)$, where $A \multimap B$ is the real vector space of $\RR$-linear functions from $A$ to $B$.
Given bases for the spaces $A,B$, we identify linear functions $A\multimap B$ with matrices. 
Seen this way, we have $df(x, v)= Jf_x \cdot v$, where $v$ is an $n\times 1$ matrix, and $\cdot$ is now matrix-matrix multiplication. 
Let $\tau =\RR^k$. We can generalize the previous definition by seeing the differential of $f:\RR^n\to \RR^m$ as a function $df:\RR^n \times (\tau \multimap \RR^n) \to (\tau \multimap \RR^m)$, still defined by $df(x, v)= Jf_x \cdot v$. We call these functions $\tau$-differentials. The functions $f\in \tau \multimap \RR^n$ can be seen as the $(k,1)$ velocities from \cite{betancourt2018geometric}.
For finite dimensional vector spaces $A, B$, the vector space $A\multimap B$ is isomorphic to $A \otimes B$, where $A \otimes B$ is the tensor product of vector spaces.
Therefore, an equivalent way to define the $\tau$-differential of $f:\RR^n\to \RR^m$ is as a function $df:\RR^n \times (\RR^n\otimes \tau) \to (\RR^m \otimes \tau)$.
This is, in essence, what our tensorized forward-mode AD computes. 
For instance, if $\tau=\RR^n$, by choosing $v$ to be the identity matrix, we have that $df(x,v)= J_xf$ is the entire Jacobian of $f$, whereas usual forward-mode only computes a column of the Jacobian matrix, by choosing $v$ to be a hot vector. Such a view on differentials extends to functions $f$ 
with a tensor input or output, as follows:
\begin{definition}
    The function $df: \RR^n \times (\tau \multimap \RR^n) \to (\tau \multimap \llbracket$!tensor m!$\rrbracket)$ is the $\tau$-differential of $f:\RR^n\to \llbracket$!tensor m!$\rrbracket$ if, for all canonical projections $\pi:\llbracket$!tensor m!$\rrbracket\to \RR^k$, and canonical injections $i:\RR^p\to \tau$, $\lambda (x, l). \pi\circ df(x, l\circ i)$ is the $\RR^p$-differential of $\pi\circ f$.

    Likewise, The function  $df: \llbracket$!tensor n!$\rrbracket \times (\tau \multimap \llbracket$!tensor n!$\rrbracket) \to (\tau \multimap \llbracket$!tensor m!$\rrbracket)$ is the $\tau$-differential of $f:\llbracket$!tensor n!$\rrbracket\to \llbracket$!tensor m!$\rrbracket$ if, for all canonical projections $\pi:\llbracket$!tensor m!$\rrbracket\to \RR^k$, and canonical injections $i:\RR^p\to \tau,j:\RR^q\to \llbracket$!tensor n!$\rrbracket$, $\lambda (x, l). \pi\circ df(j(x), i^*\circ l\circ i)$ is the $\RR^p$-differential of $\pi\circ f\circ j$, where $i^*$ is the linear adjoint of $i$.
\end{definition}

It then further extends to functions with multiple inputs or outputs, following Definition~\ref{defn:extension-differentials}.

We can now state that tensorized forward-mode AD computes differentials:
\begin{theorem}
Let $\tau$ be a tensor type.
 Let $\Gamma \vdash$! e: D!, where $\Gamma$ only contains tensor types. Then, $\llbracket$\fvadttang{\code{e}}$\rrbracket$ is the $\tau$-differential of $\llbracket$!e!$\rrbracket$.
\end{theorem}

The proof follows a similar logical relations argument as Theorem~\ref{thm:simple-fad-correct}, where the logical predicates are slightly changed to take the type $\tau$ into account. For instance
\begin{align*}
    \mathcal{R}_{\codekw{real}}^\tau = & \big\{ f: \sem{\tau} \to \RR, g: \RR\times (\sem{\tau}\multimap \sem{\tau}) \to (\sem{\tau}\multimap \RR)~\mid~f\text{ morphism of diffeological spaces}, \\
 & g\text{ $\tau$-differential of }f \big\}
\end{align*}

\section{Experimental Results}
\label{sec:exp}


In this section, we see the effectiveness of \system in practice by using real-world and synthetic tensors. We answer the following research questions:

\begin{itemize}
    \item How does \system{} perform for tensor kernels over real-world sparse matrices in comparison with the state-of-the-art AD frameworks? Does \system{} scale to large sparse matrices?
    \item What is the impact of different physical storage formats on the performance of \system{}?
    \item What is the impact of different optimizations on the performance of \system{}?
    \item For which densities does it make sense to use sparse representations and compute the gradients using \system{}?
\end{itemize}

\subsection{Experimental Setup}
The experiments were conducted on a MacBook Pro featuring a dual-core Intel Core i7 CPU clocked at 3.5 GHz with 16 GB of LPDDR3 RAM (2133MHz), running macOS Big Sur 11.5.2. To compile the generated C++ code, we employed CLang 1205.0.22.11 and Python 3.7.4 for executing the Python code. We used TensorFlow 2.11.0 (XLA-enabled) and PyTorch 1.13.1.

\begin{table}
  \caption{Real-world matrices used in the end-to-end experiments.}
  \vspace*{-1em}
  \label{tbl:datasets}
  \begin{tabular}{l l r r l}
    \hline
    Matrix    & Dimensions                      & Density             & \# NNZ & Description \\ \hline \hline
    bcspwr10  & 5.3K $\times$ 5.3K              & $2^{-9.3}$  & 22K & Power Network \\
    nopoly    & 10K $\times$ 10K                & $2^{-9.5}$  & 70K & Undirected Weighted Graph \\
    pdb1HYS   & 36K $\times$ 36K                & $2^{-8.2}$  & 2.19M  & Undirected Weighted Graph      \\
    rma10     & 46K $\times$ 46K                & $2^{-9.8}$  & 2.37M  & Computational Fluid Dynamics      \\
    cant      & 62K $\times$ 62K                & $2^{-9.9}$  & 2.03M  & Finite Element Method      \\
    consph    & 83K $\times$ 83K                & $2^{-10.1}$  & 3.05M  & Finite Element Method      \\
    cop20k\_A & 121K $\times$ 121K              & $2^{-12.4}$  & 1.36M  & Finite Element Method      \\
    \hline
  \end{tabular}

\end{table}

We consider the following sparse matrix kernels:

\begin{itemize}
    \item \textbf{BATAX.} The BATAX kernel computed using the following formula~\cite{DBLP:journals/toms/NelsonBSJN15}: $f(j) = \Sigma_{i,k} \beta \cdot A(i,j)\cdot A(i,k)\cdot X(k)$. We consider its gradient with respect to the vector: $\frac{\partial f}{\partial X}$.
    \item \textbf{SMMM.} The summation of elements of matrix-matrix-multiplication~\cite{schleich2022optimizing}: $f = \Sigma_{i,j,k} A(i,k) \cdot B(k,j)$. We consider its gradient with respect to the second matrix: $\frac{\partial f}{\partial B}$.
    \item \textbf{SMVM.} The summation of the elements of a matrix-vector-multiplication~\cite{DBLP:journals/pacmpl/SmedingV23}: $f = \Sigma_{i,j} A(i,j) \cdot X(j)$. We consider its gradient with respect to the vector: $\frac{\partial f}{\partial X}$.
\end{itemize}

\noindent The first two kernels produce a matrix result, while the last one produces a vector. 

We also consider the following vector kernels from \dfsmooth{} (\cite{DBLP:journals/pacmpl/ShaikhhaFVJ19}):

\begin{itemize}
    \item \textbf{VVA.} The addition of two vectors: $f(i) = V_1(i) + V_2(i)$. We consider its gradient with respect to the first vector: $\frac{\partial f}{\partial V_1}$.
    \item \textbf{VVD.} The dot product of two vectors: $f = \Sigma_i V_1(i) \cdot V_2(i)$. We consider its gradient with respect to the first vector: $\frac{\partial f}{\partial V_1}$.
    \item \textbf{VSM.} The vector-scalar multiplication: $f(i) = V(i) \cdot s \cdot s$. The scalar value was multiplied by itself to make the differentiation more complicated. We consider its gradient with respect to the scalar value: $\frac{\partial f}{\partial s}$.
\end{itemize}

\noindent The first kernel produces a matrix result, while the last two produce a vector. 

As competitors for the matrix kernels, we consider the following systems:

\begin{itemize}
    \item \textbf{\system{} (dict)}: Generated C++ code by \system from tensorized forward-mode AD that uses a nested dictionary.
    \item \textbf{\system{} (arr)}: Generated C++ code by \system from tensorized forward-mode AD that also benefits from the physical storage specified by the sparse representations.
    \item \textbf{TensorFlow}: The reverse-mode AD of the TensorFlow framework, using its \texttt{gradient} and \texttt{jacobian} APIs.
    \item \textbf{PyTorch}: The reverse-mode AD of the PyTorch framework, using its \texttt{jacobian} API.
\end{itemize}

For the vector kernels, in addition to the above competitors, we consider the following systems:\footnote{The generated code were made available in \dfsmooth here: https://github.com/amirsh/autodiff2/tree/master/icfp19}

\begin{itemize}
    \item \textbf{Tapenade (R)}: Generated C code by Tapenade using reverse-mode AD.
    \item \textbf{Tapenade (F)}: Generated C code by Tapenade using forward-mode AD.
    \item \textbf{\dfsmooth}: Generated C code by \dfsmooth using its forward-mode AD.
    \item \textbf{\dfsmooth + DPS}: Generated C code by \dfsmooth using its forward-mode AD with the DPS optimization.
\end{itemize}

Both real-world and synthetic datasets were utilized in our study. To obtain the former, we gathered seven sparse matrices from the SuiteSparse Matrix Collection~\cite{DBLP:journals/toms/DavisH11}. Table~\ref{tbl:datasets} summarizes these datasets. For synthetic data, we created random matrices and vectors with different density and dimension configurations.

In all these benchmarks, we consider the gradient with respect to a dense variable. All experiments use a single-core. We take the average time of running five runs.

\begin{figure}[t]
\setlength\tabcolsep{.5pt}
    \centering
    \begin{tabular}{ccc}
         \includegraphics[width=0.33\linewidth]{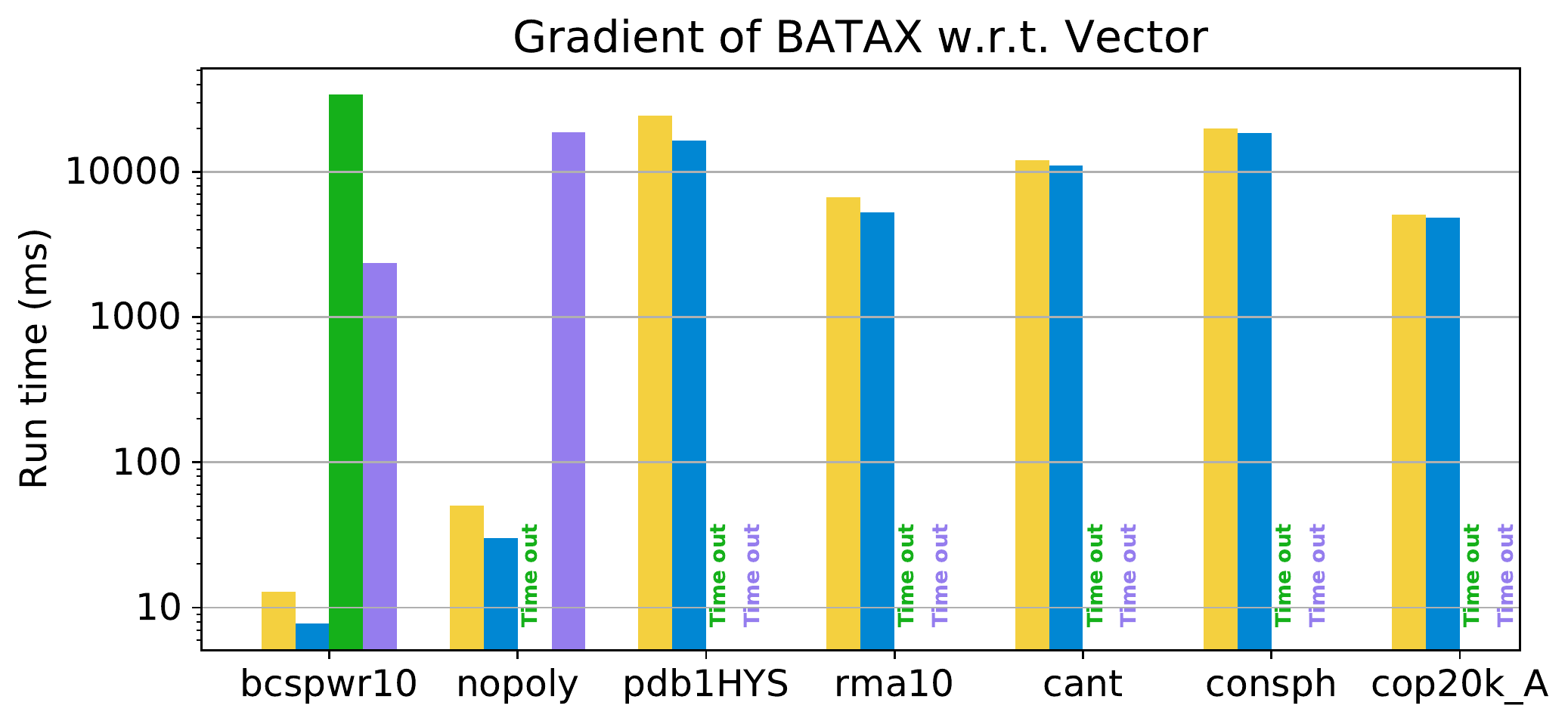} & \includegraphics[width=0.33\linewidth]{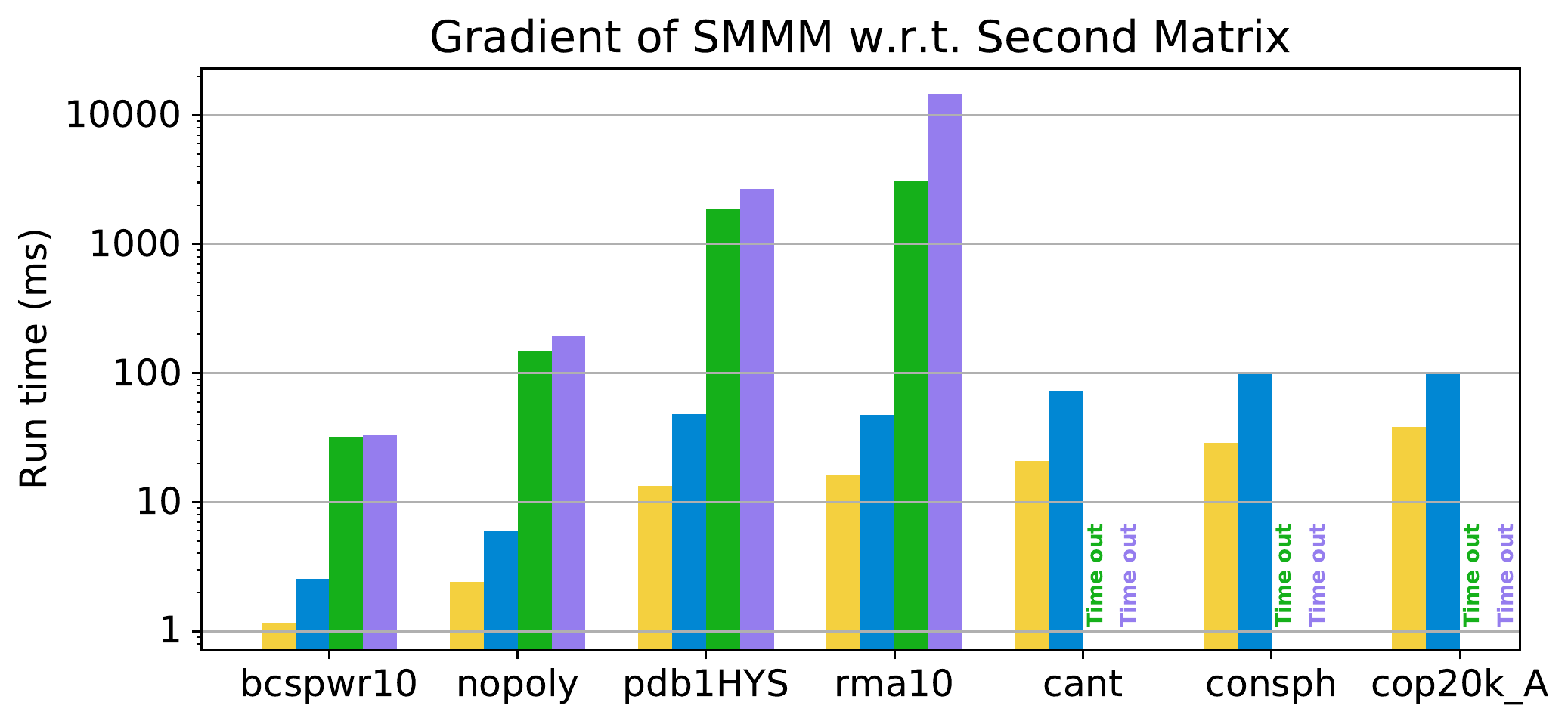} & \includegraphics[width=0.33\linewidth]{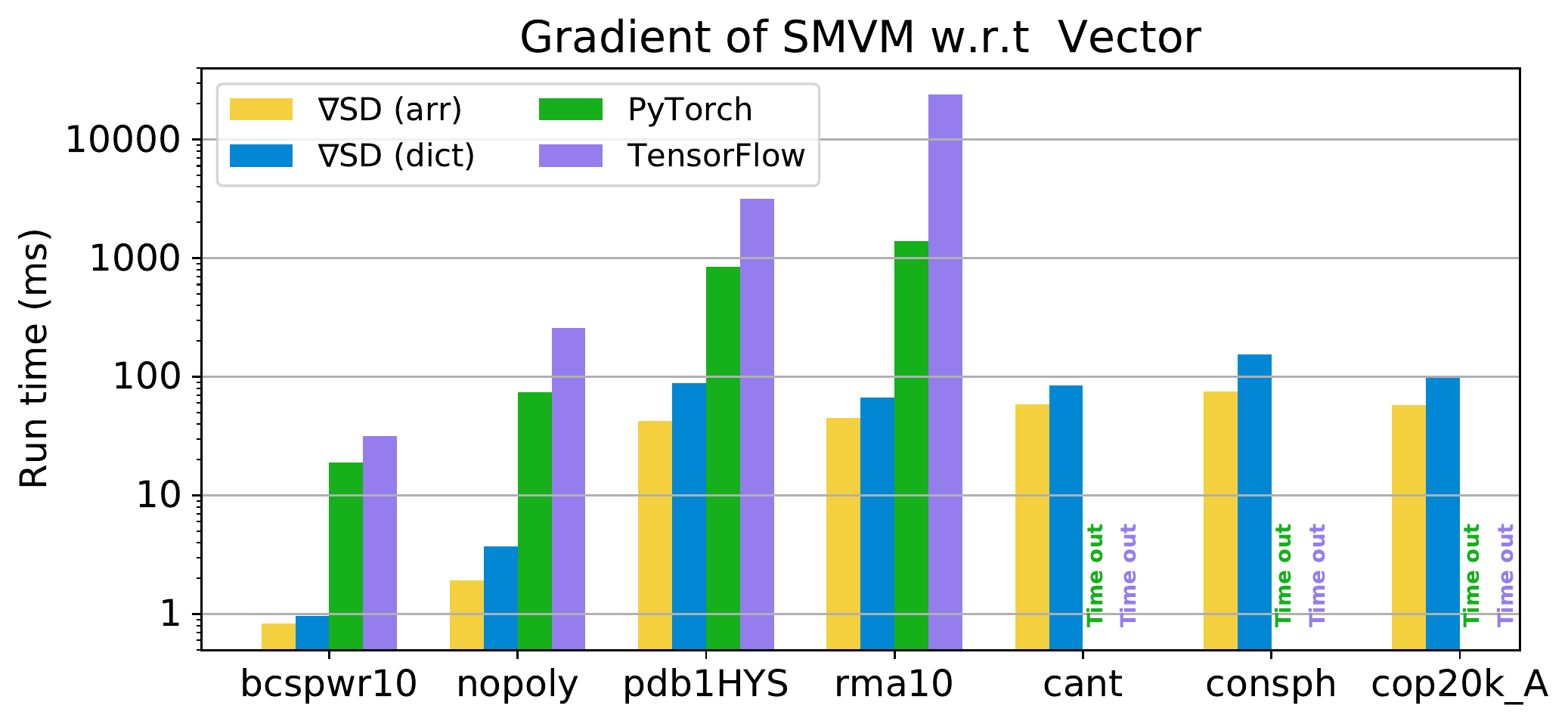}
    \end{tabular} 
    \vspace{-0.5cm}
    \caption{Performance results for differentiation of kernels over real-world sparse matrices. \system{} uses tensorized forward mode over sparse matrices, while TensorFlow and PyTorch use reverse-mode over dense matrices.}
    \label{fig:e2e:exp}
\end{figure}

\subsection{Benchmarks over Real-world Datasets}
In this section, we consider real-world sparse matrices. We compare \system with TensorFlow and PyTorch. Although both systems support sparse operations, none support AD over sparse matrices.

Figure~\ref{fig:e2e:exp} shows the results for the real-world sparse matrices. We make the following observations. First, for all kernels \system{} outperforms TensorFlow and PyTorch thanks to leveraging sparse representations. Second, the generated code by \system that leverages the array-based physical representation can run faster than the version that uses nested dictionaries in SMVM and SMMM. Finally, \textit{\system scales to large sparse matrices}, as opposed to TensorFlow and PyTorch, which do not manage to process such matrices due to storing the entire matrix including the zero elements. 

\subsection{Benchmarks over Synthetic Datasets}
In this section, we use synthetic datasets to better analyse the impact of the sparsity and dimensions. We consider two scenarios: (1) fixing the dimension and varying the sparsity, and (2) fixing the sparsity and varying the dimension.

\begin{figure}[t]
\setlength\tabcolsep{.5pt}
    \centering
    \begin{tabular}{ccc}
         \includegraphics[width=0.33\linewidth]{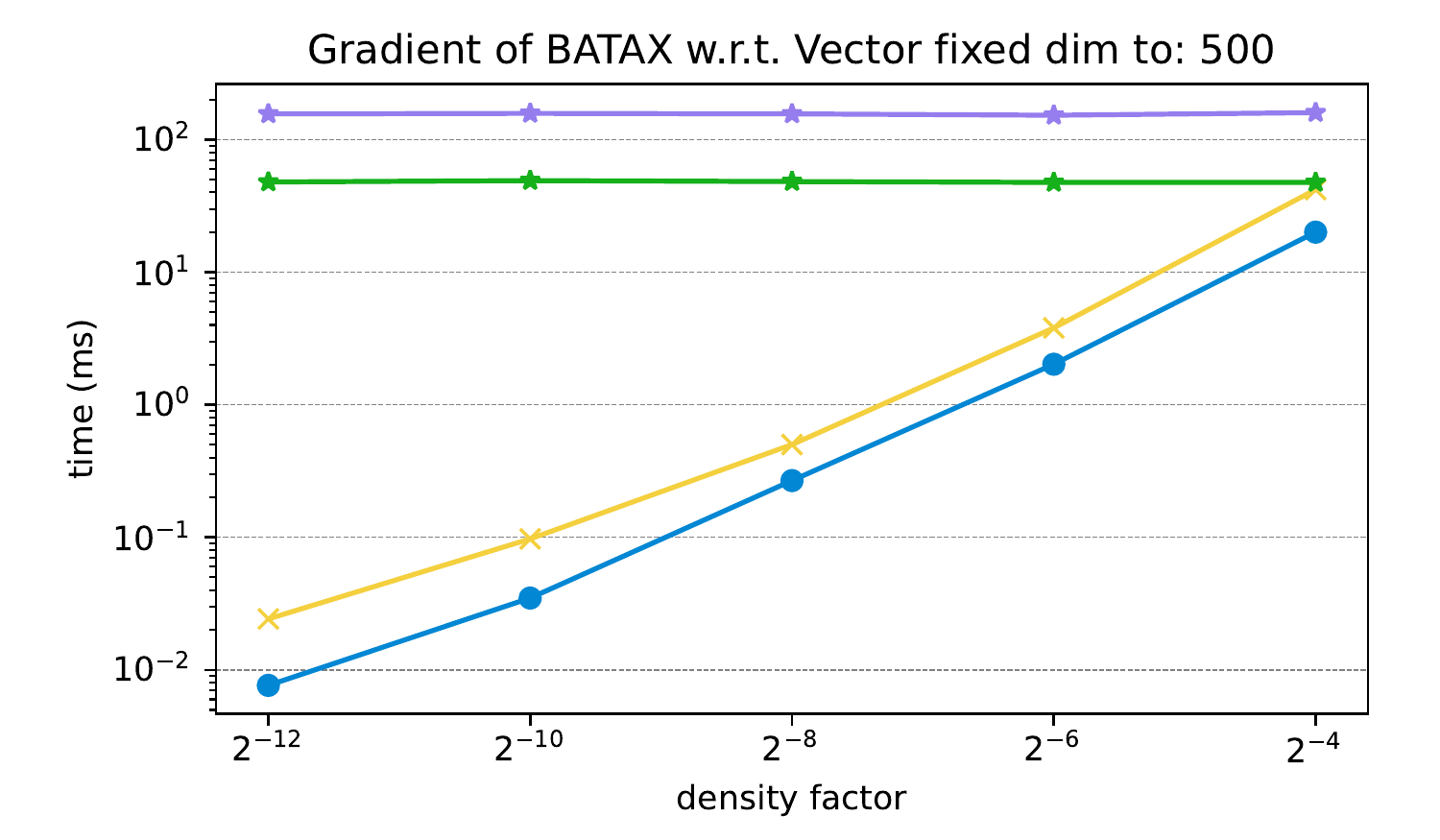} & \includegraphics[width=0.33\linewidth]{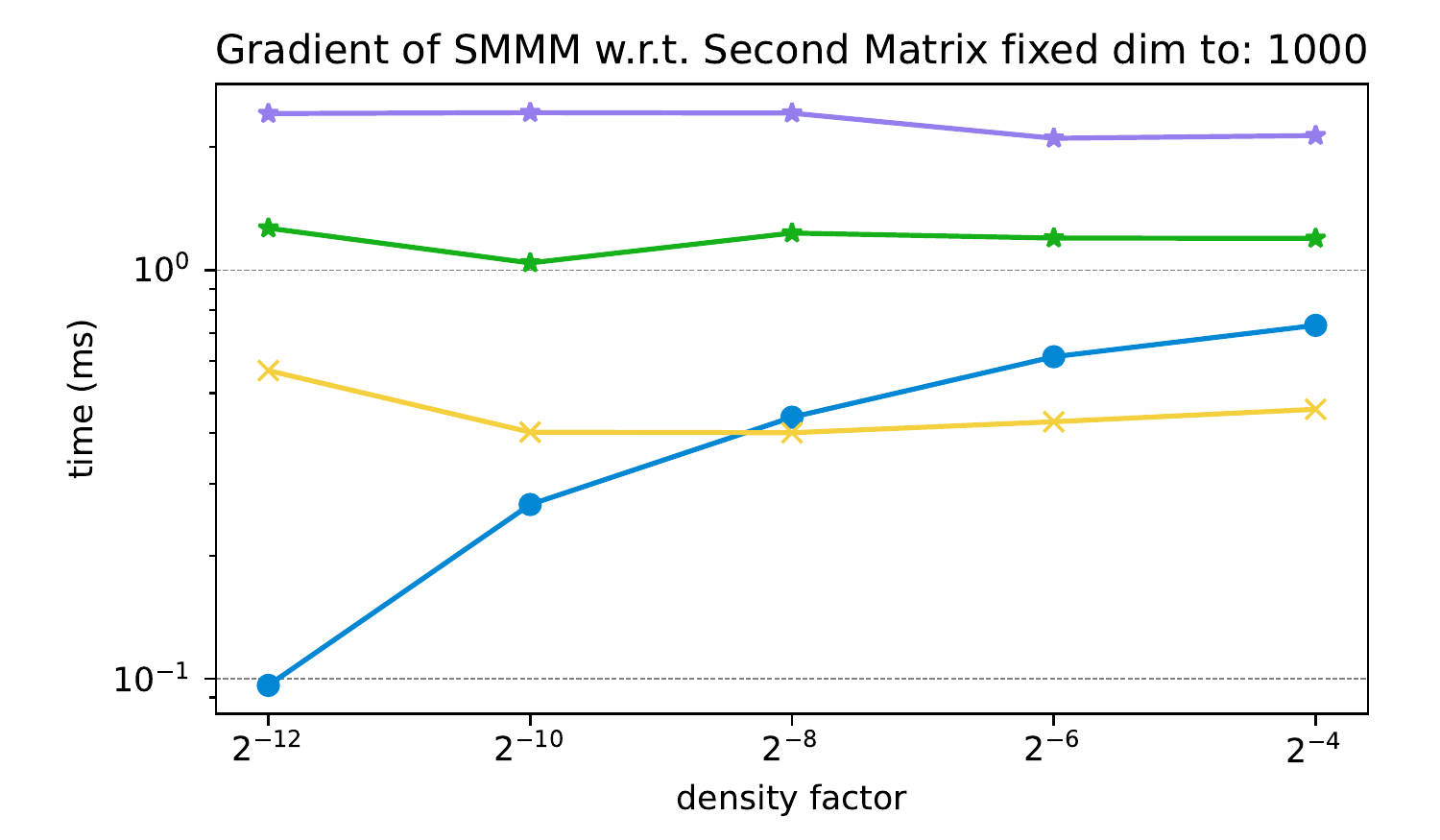} & \includegraphics[width=0.33\linewidth]{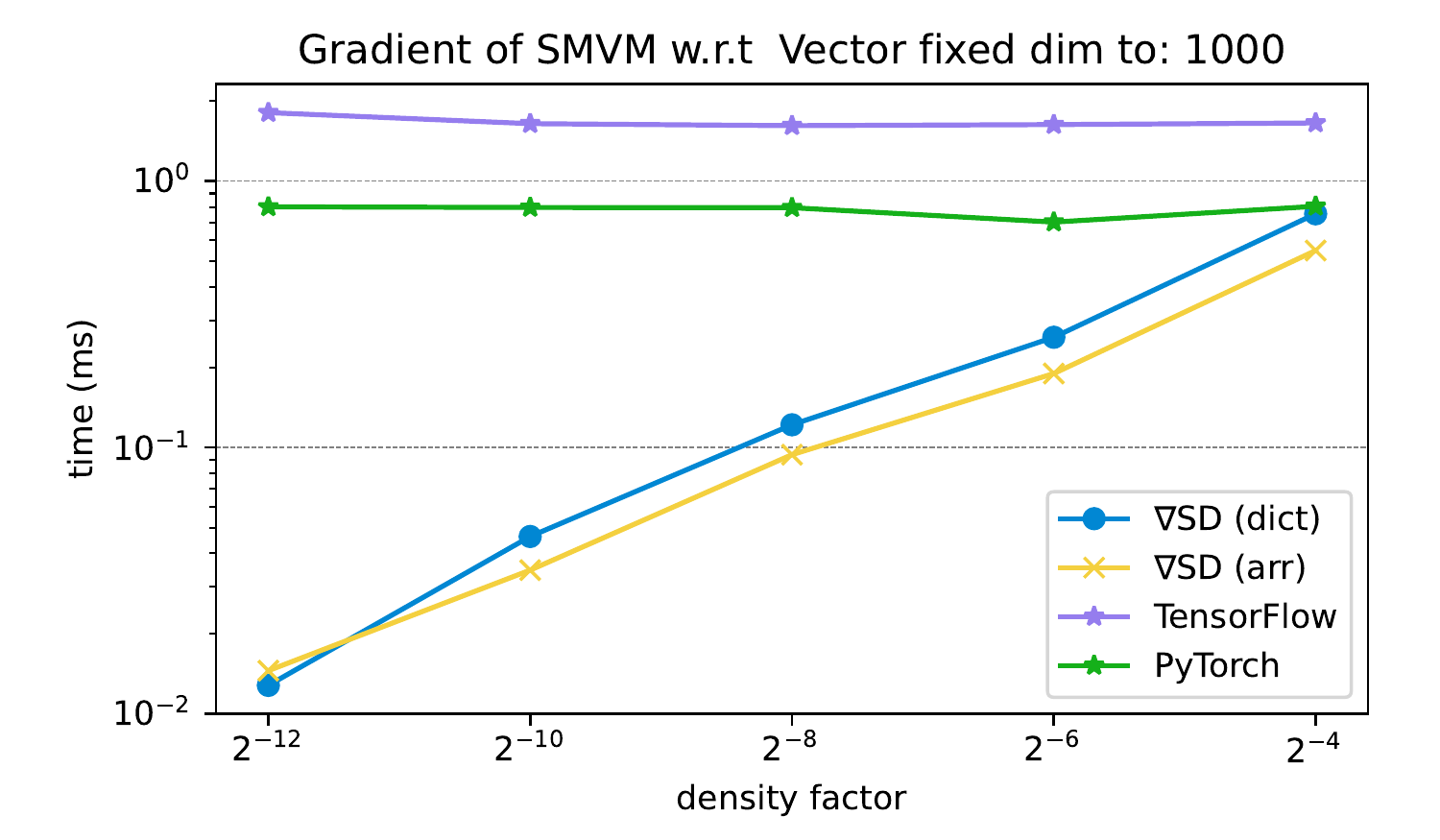} \\
         \includegraphics[width=0.33\linewidth]{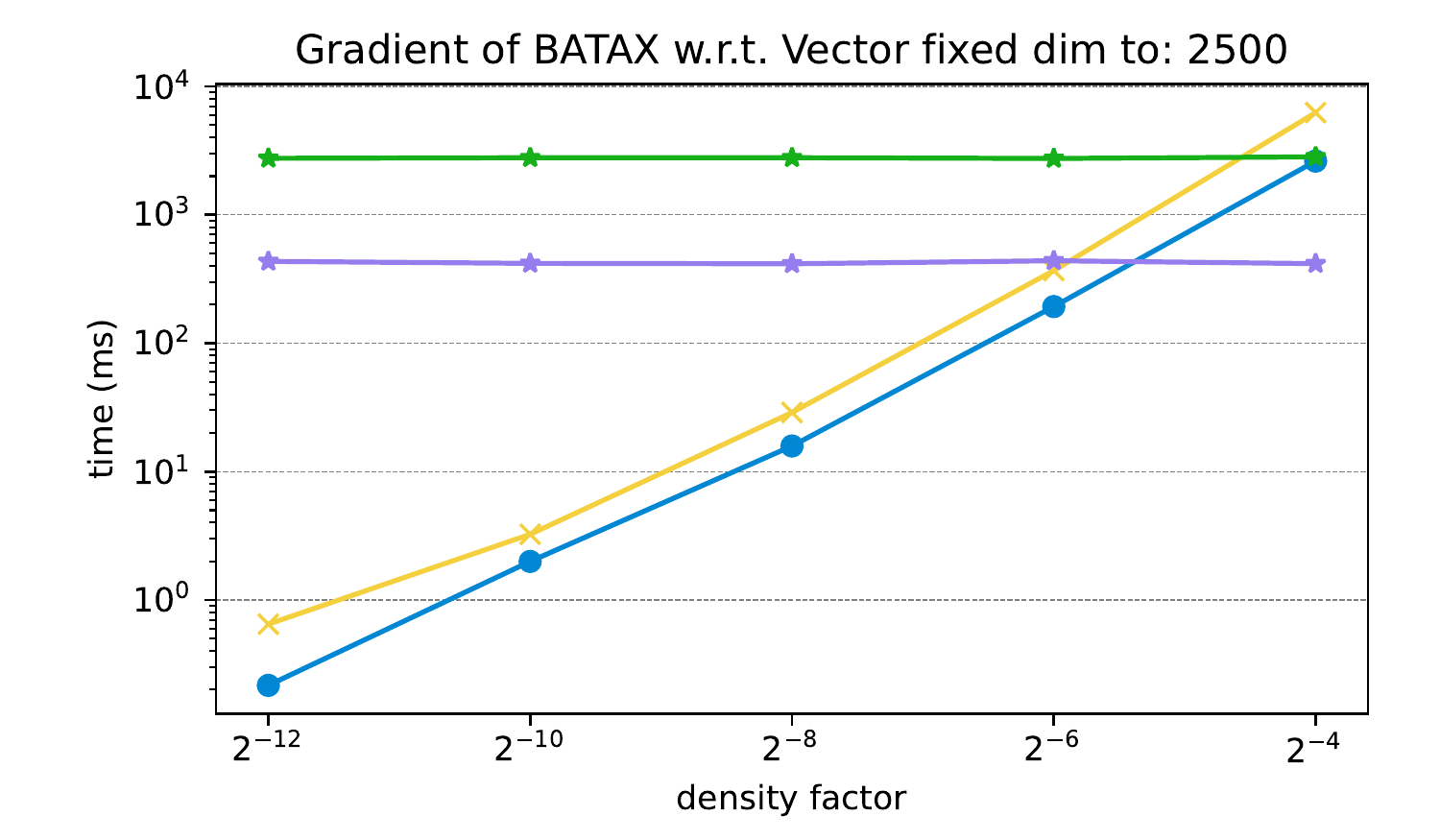} & \includegraphics[width=0.33\linewidth]{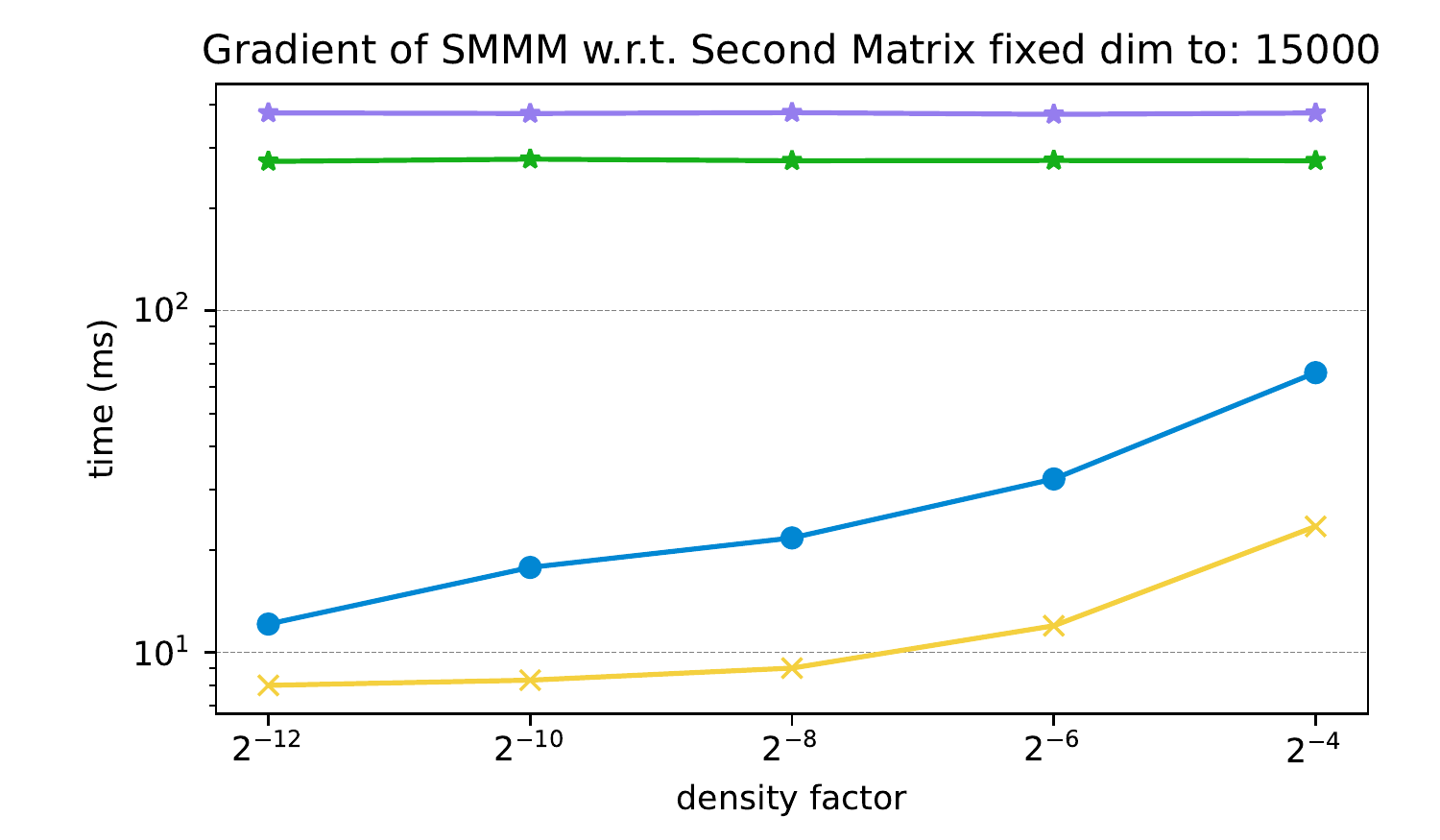} & 
         \includegraphics[width=0.33\linewidth]{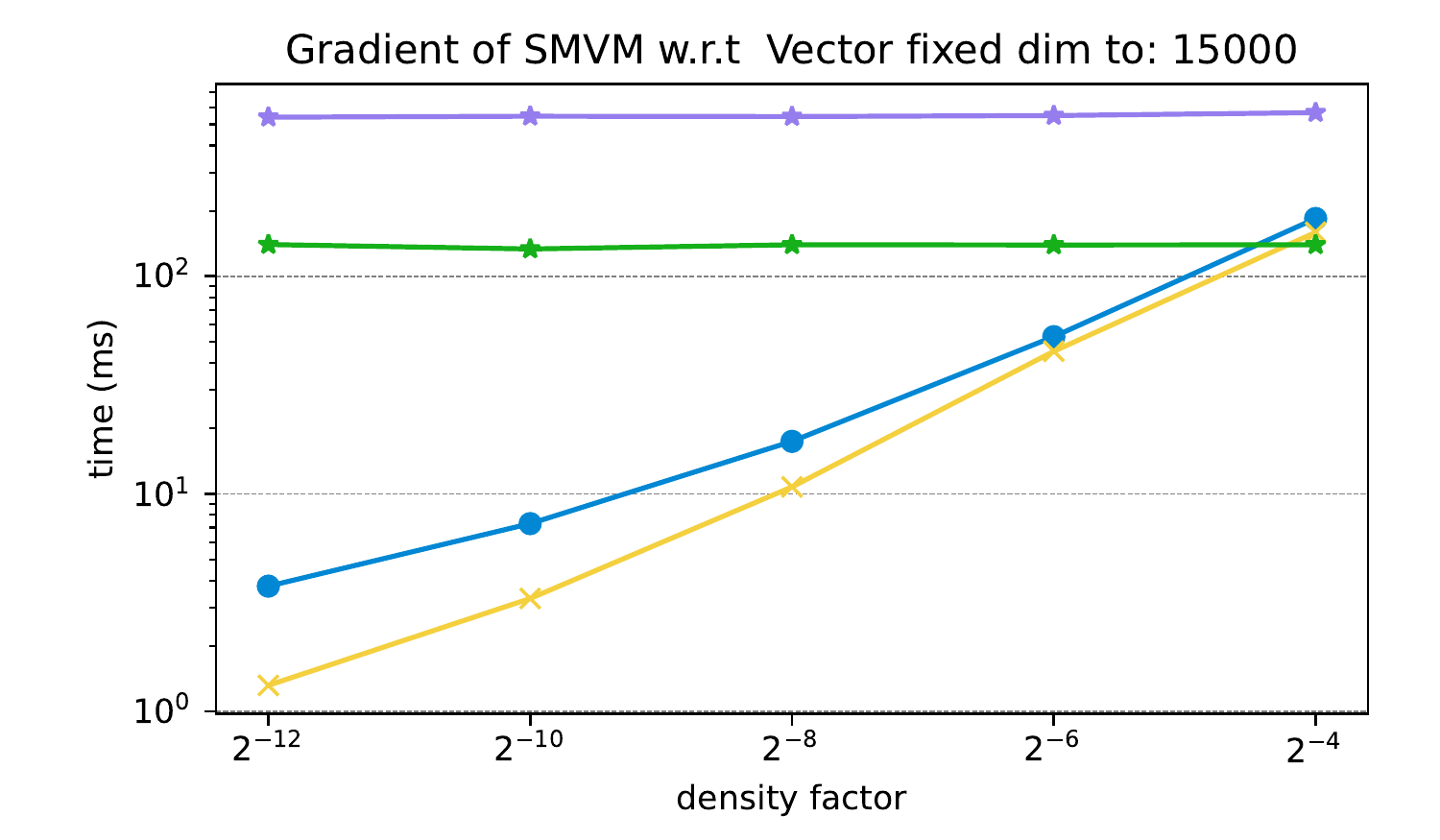}
    \end{tabular} 
    \vspace{-0.5cm}
    \caption{Performance results for differentiation of different sparse matrix kernels by varying sparsity.}
    \label{fig:matrixkernels:sp}
\end{figure}

\begin{figure}[t]
\setlength\tabcolsep{.5pt}
    \centering
    \begin{tabular}{ccc}
         \includegraphics[width=0.33\linewidth]{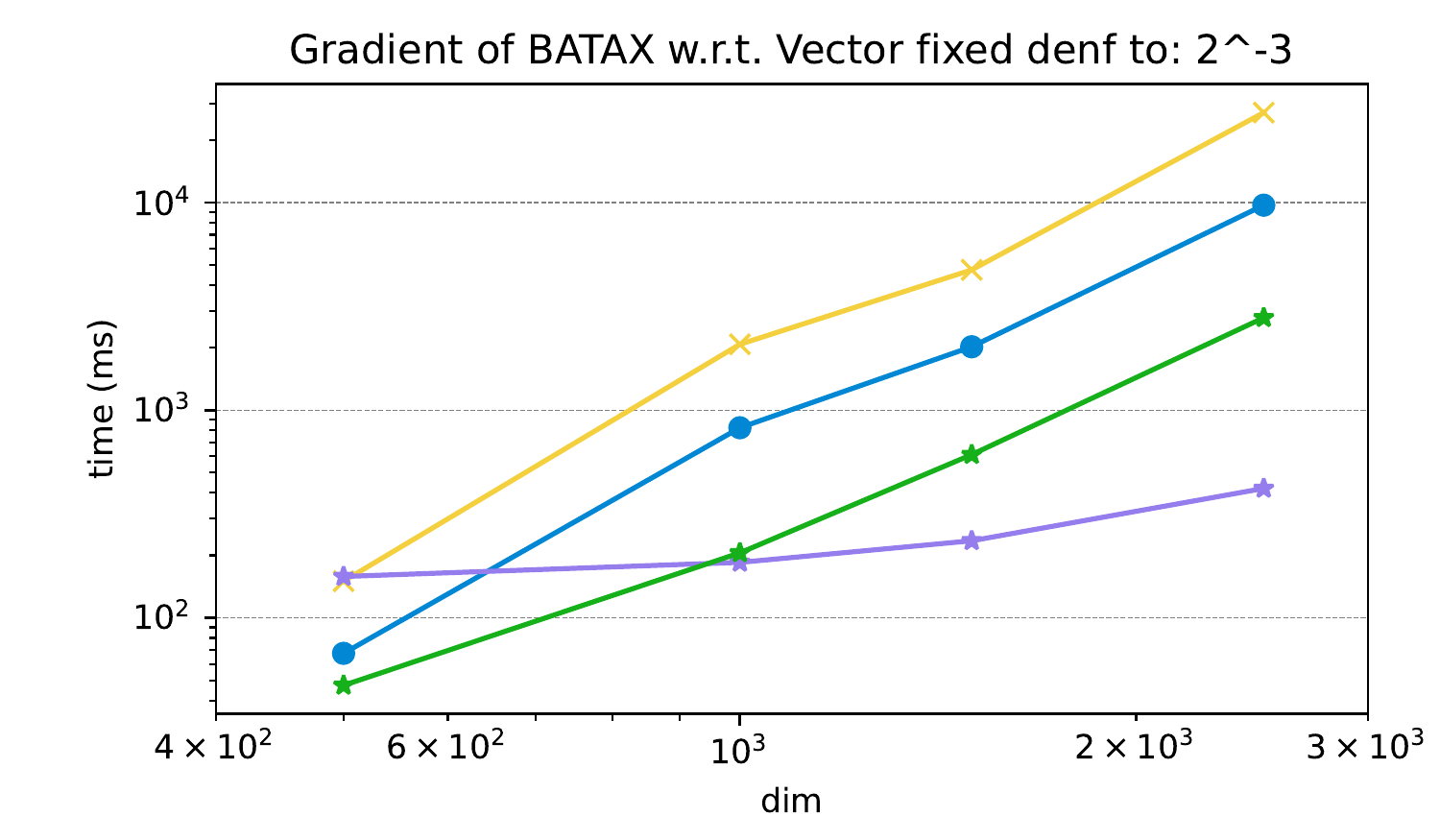} & \includegraphics[width=0.33\linewidth]{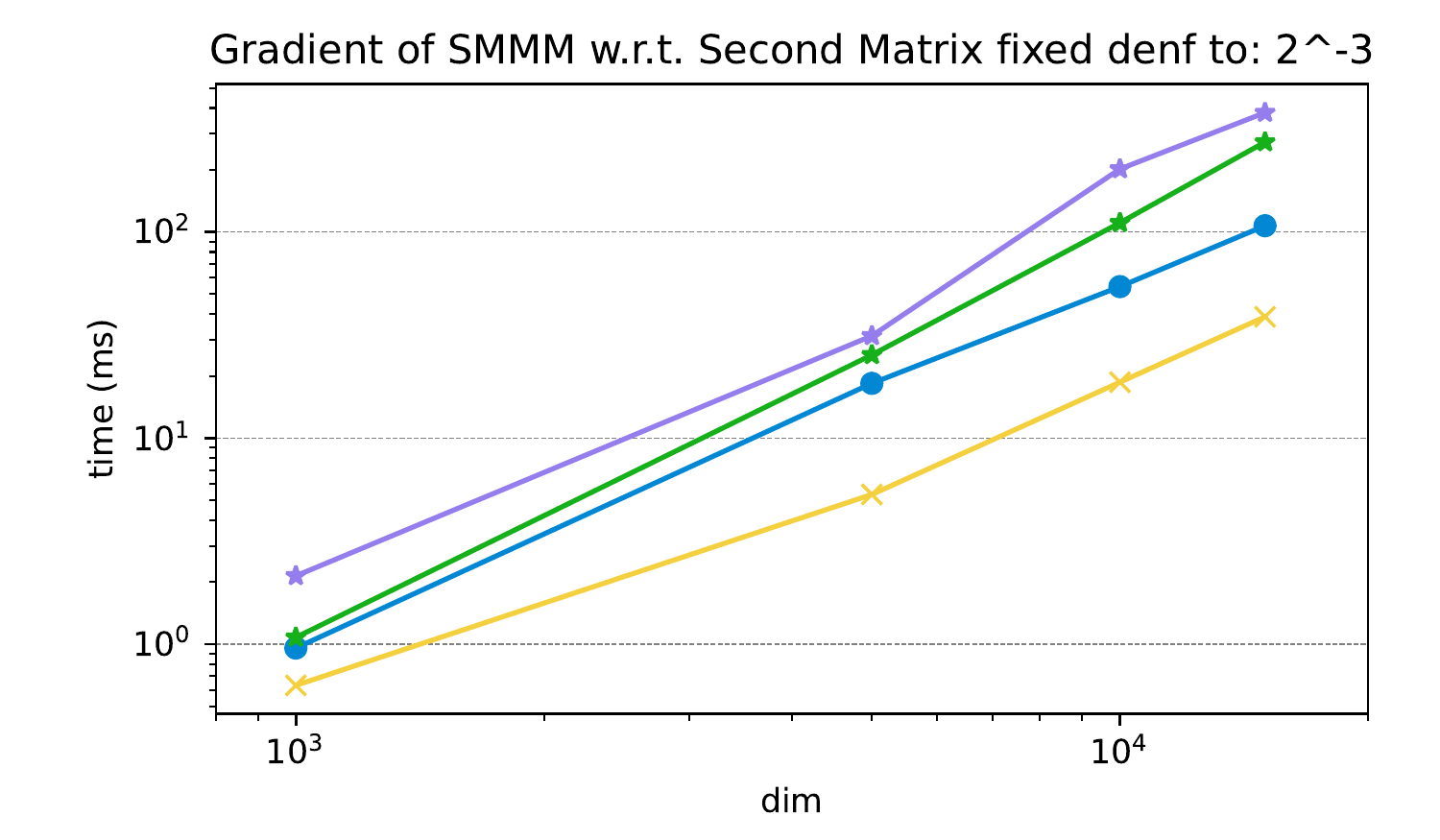} & \includegraphics[width=0.33\linewidth]{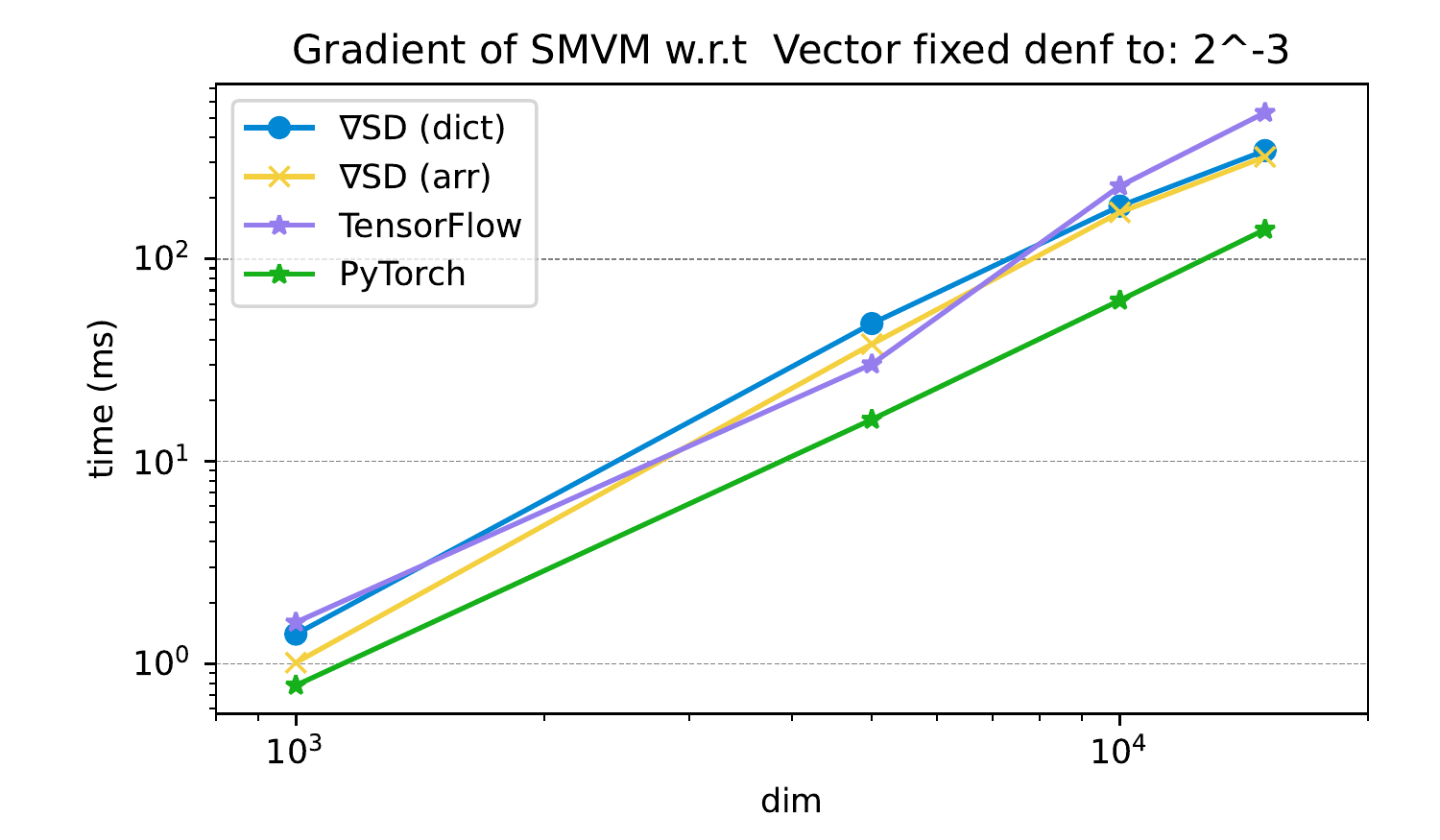} \\
         \includegraphics[width=0.33\linewidth]{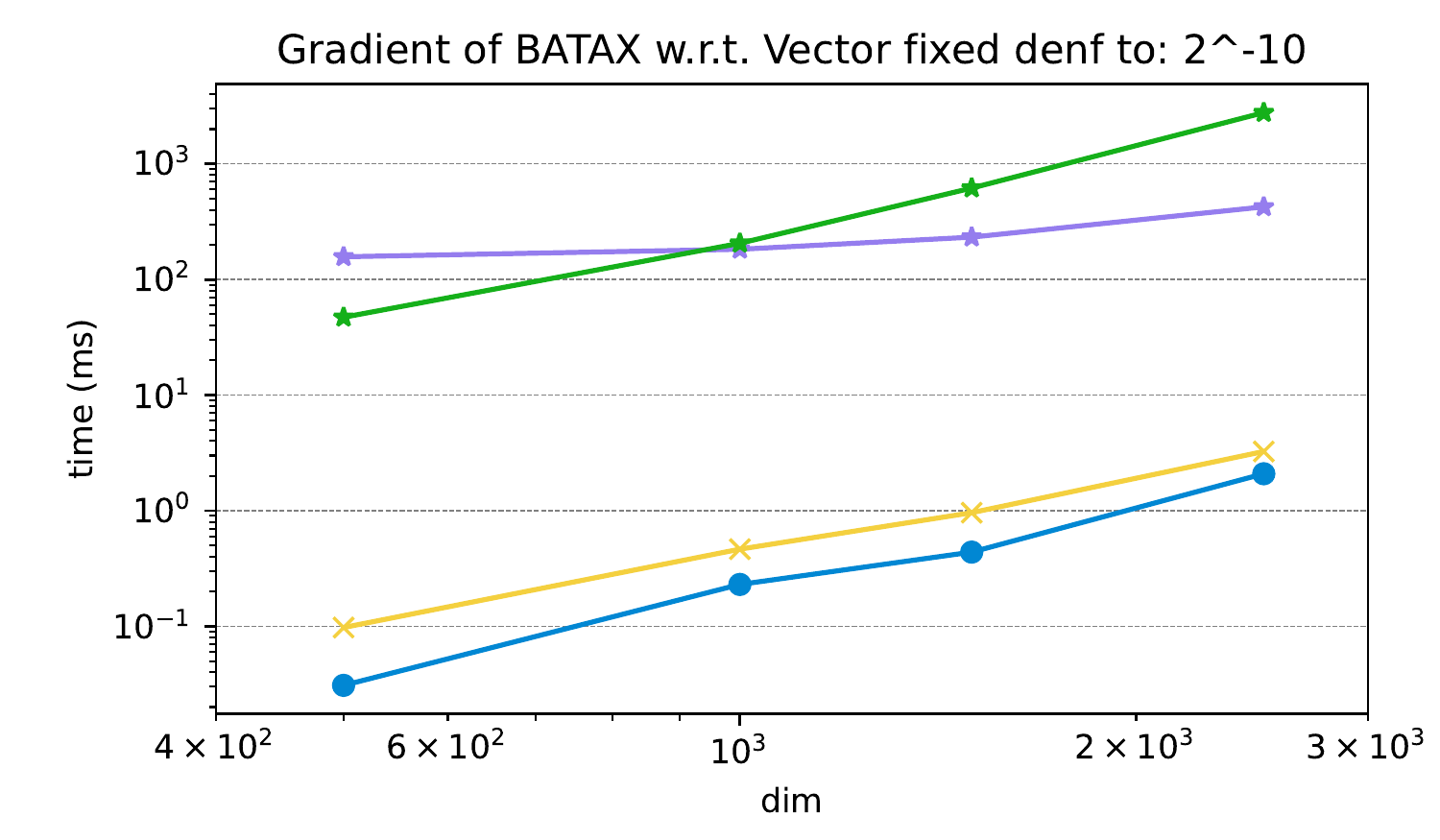} & \includegraphics[width=0.33\linewidth]{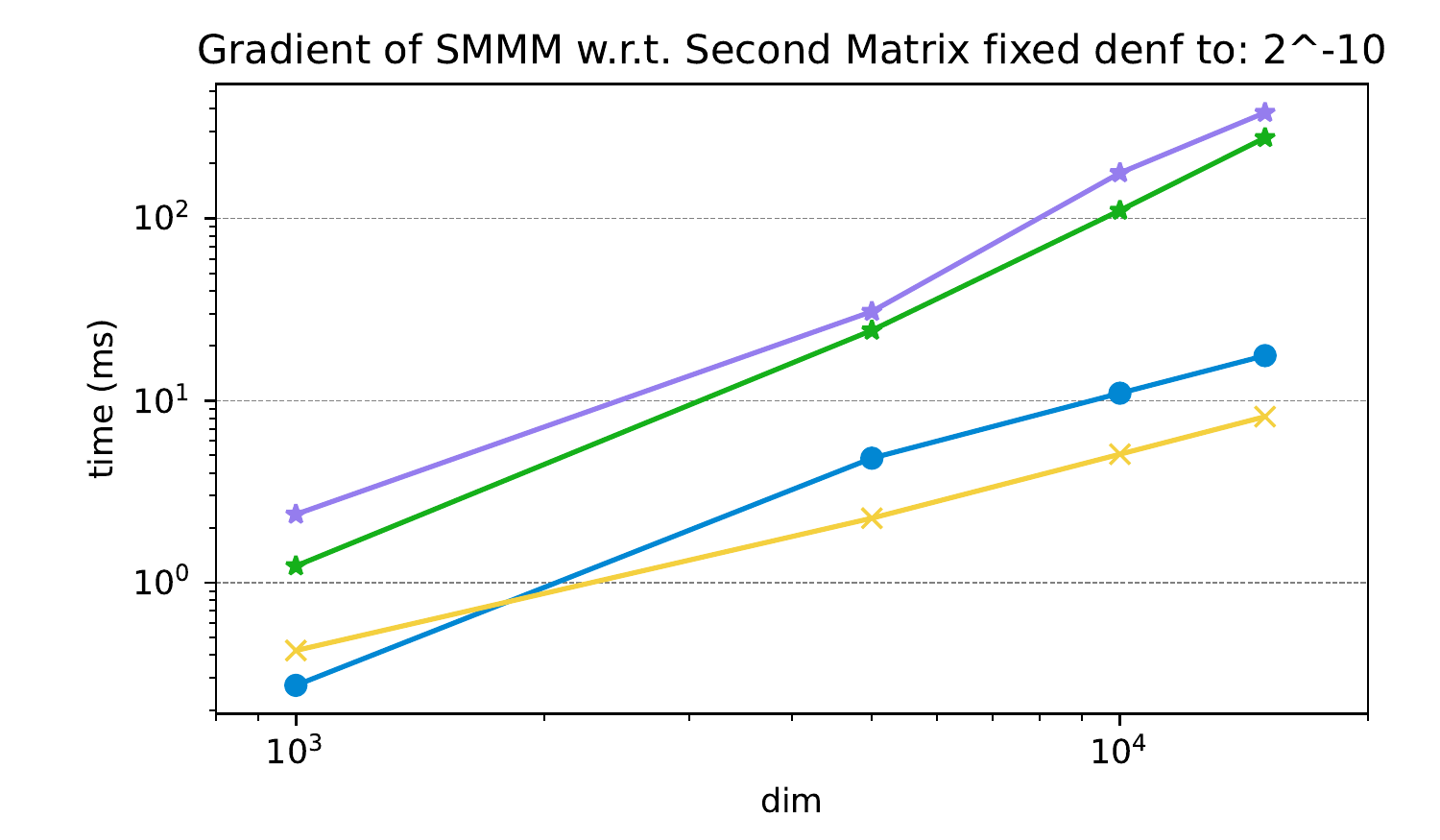} & 
         \includegraphics[width=0.33\linewidth]{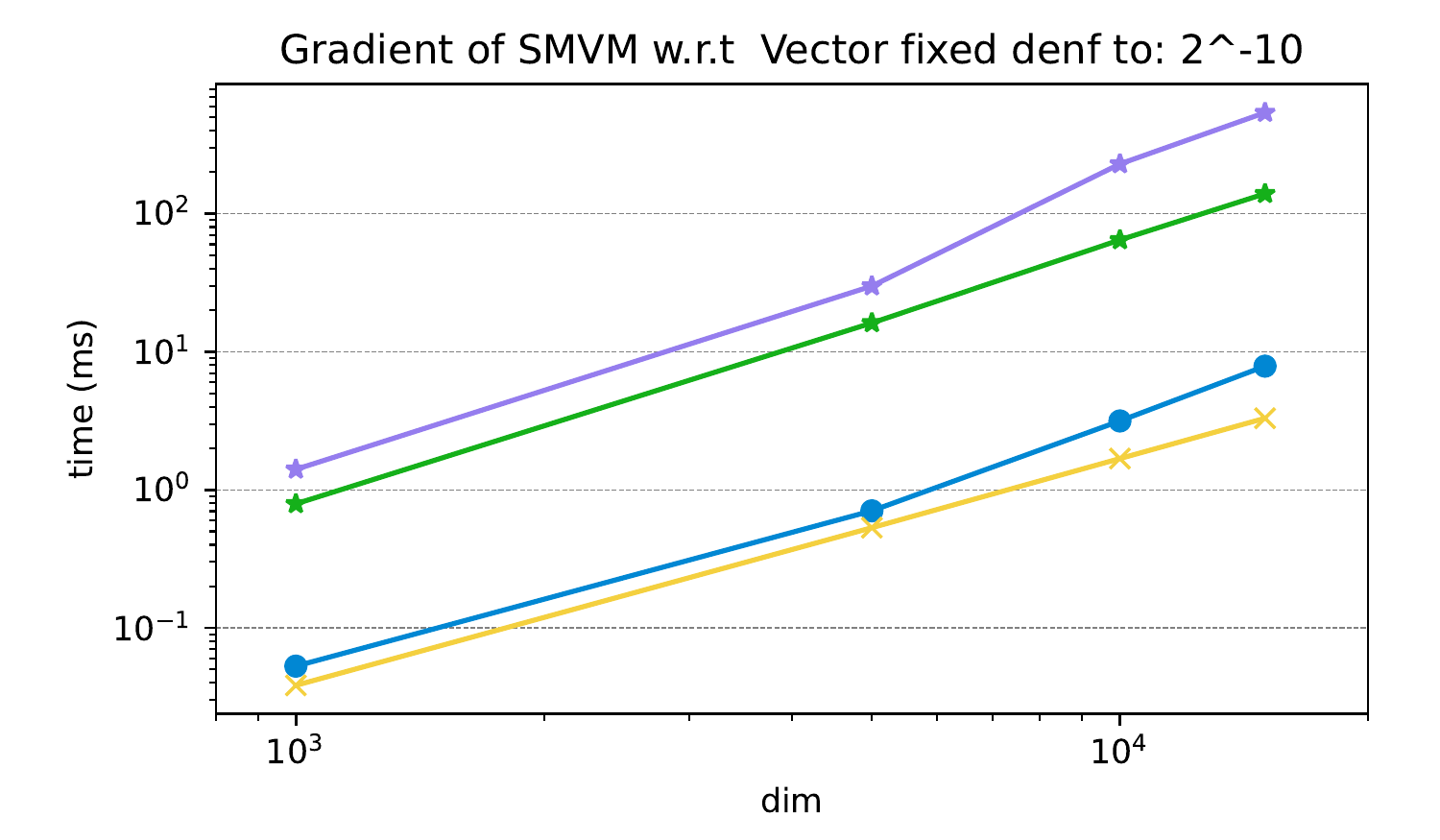}
    \end{tabular} 
    \vspace{-0.5cm}
    \caption{Performance results for differentiation of different sparse matrix kernels by varying dimensions.}
    \label{fig:matrixkernels:dim}
\end{figure}

\smartpara{Matrix Kernels} Figure~\ref{fig:matrixkernels:sp} shows the results for matrix kernels by varying the sparsity for small and large matrices. We make the following observations. First, for low densities, which are in the same range as the real-world datasets, we see a clear advantage for \system{} over TensorFlow and PyTorch. Second, the performance of TensorFlow and PyTorch is not dependent on the sparsity, as expected. As the matrices get denser, the gap between these two frameworks and \system{} becomes smaller. Finally, for the BATAX kernel, there is an advantage for the dictionary-based representation over the array-based one. However, for the other two kernels, especially for larger matrices, the array-based representation performs better.

Figure~\ref{fig:matrixkernels:dim} shows the results by varying the dimension for two different sparsities. Matrices with high densities show a better performance for TensorFlow and PyTorch over \system{}. The gap even widens for larger dimensions. However, we observe the opposite impact for a lower density.

\begin{figure}
\setlength\tabcolsep{.5pt}
    \centering
    \begin{tabular}{ccc}
         \includegraphics[width=0.33\linewidth]{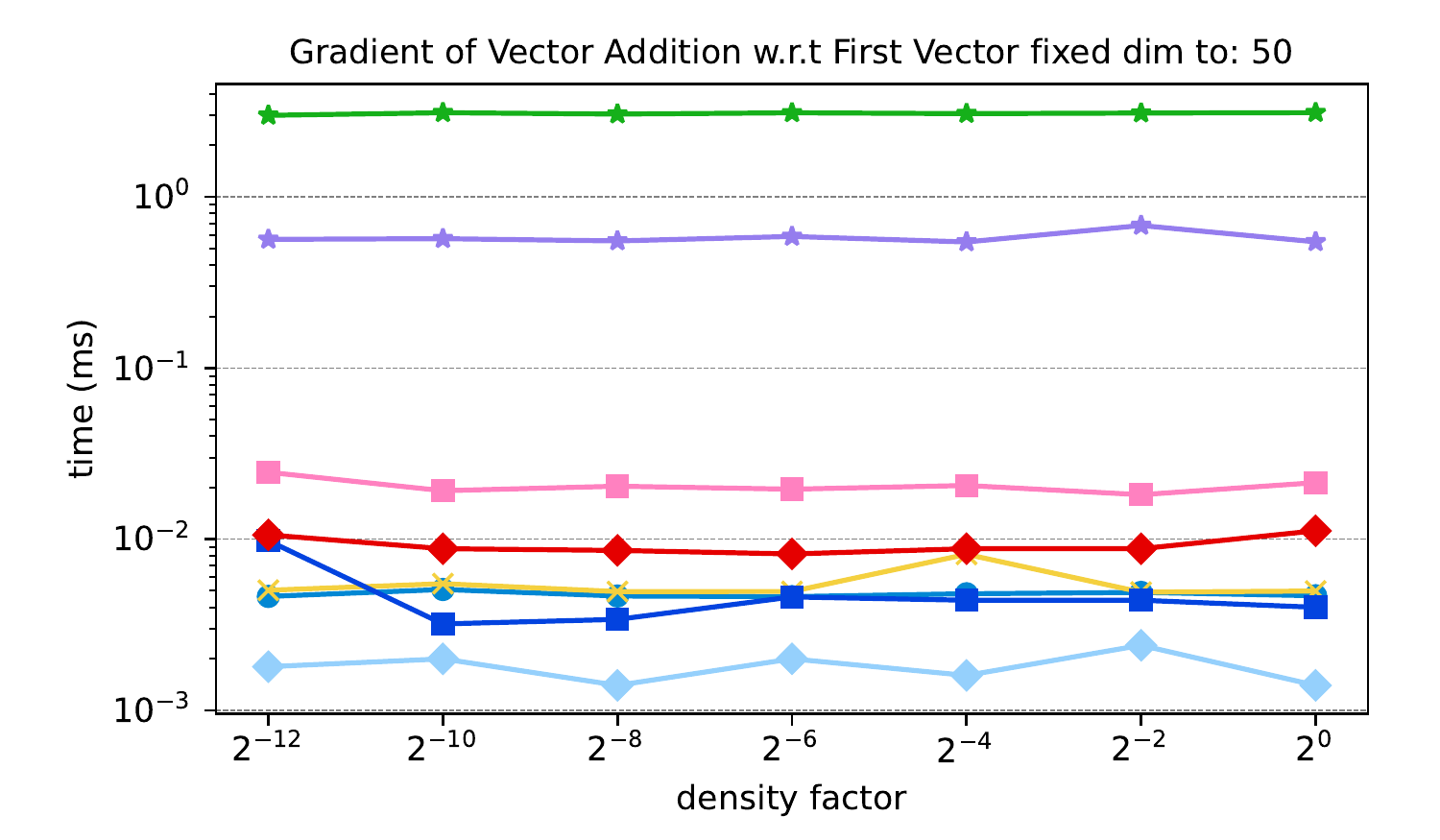} & \includegraphics[width=0.33\linewidth]{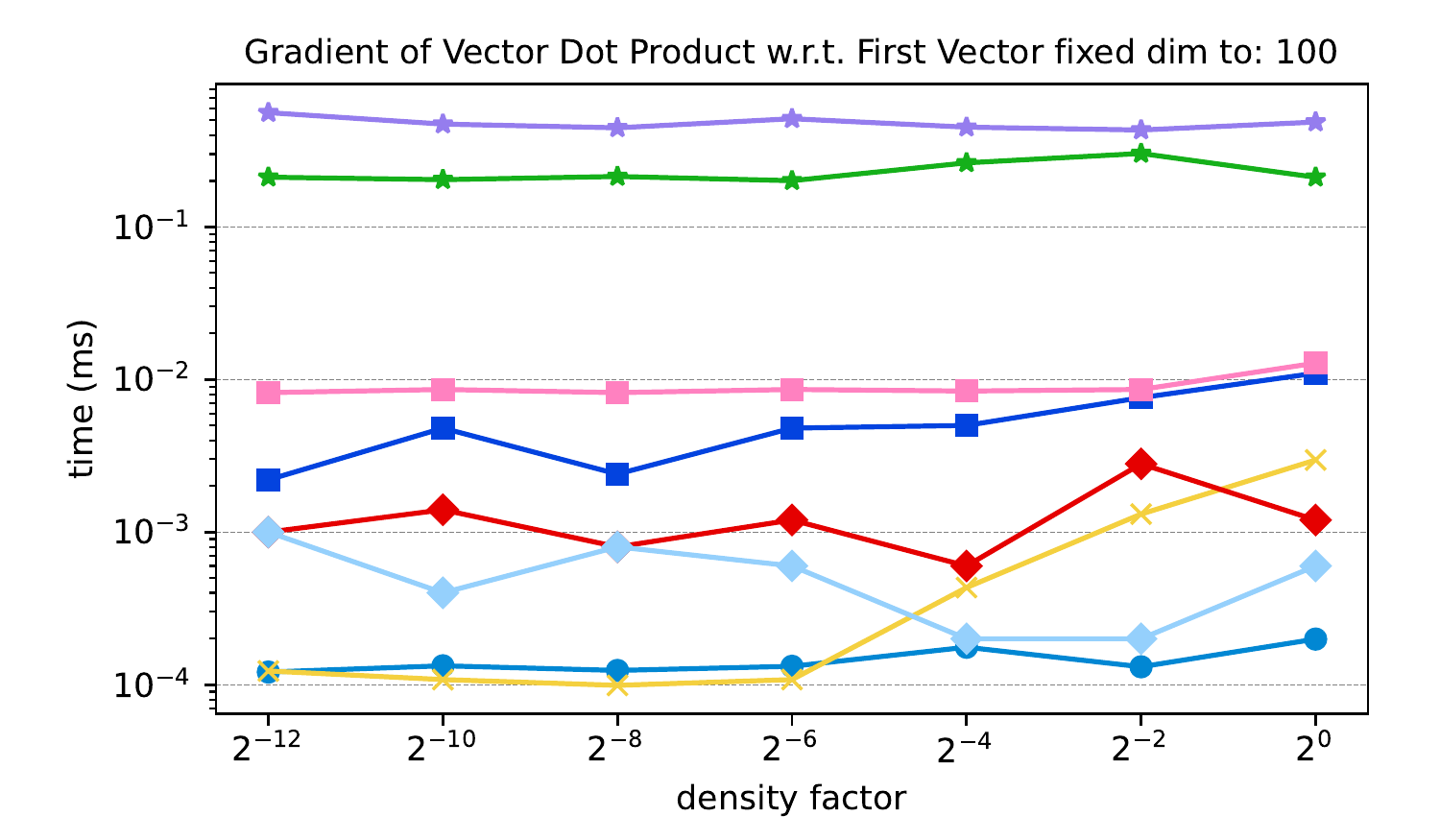} & \includegraphics[width=0.33\linewidth]{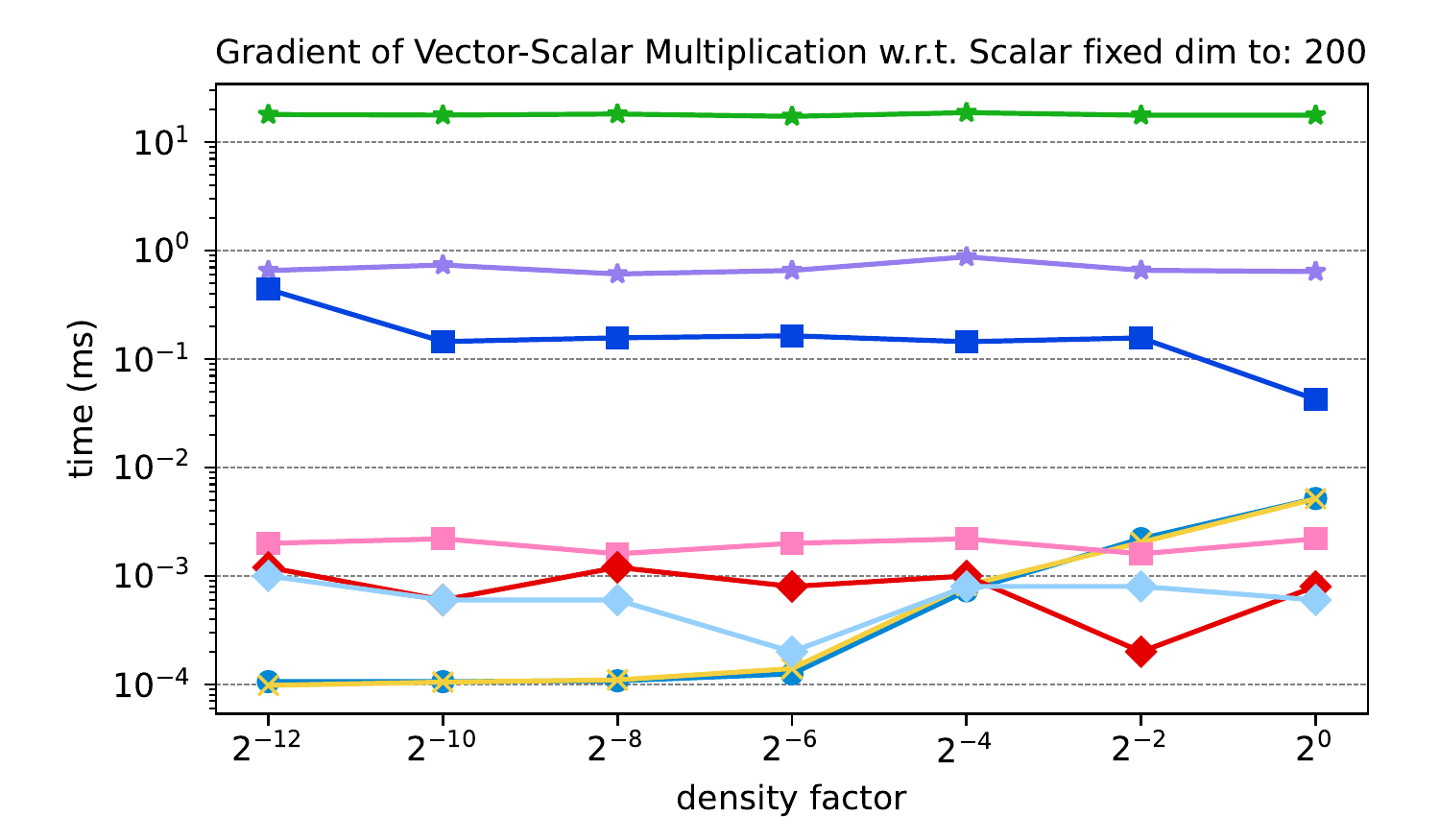} \\
         \includegraphics[width=0.33\linewidth]{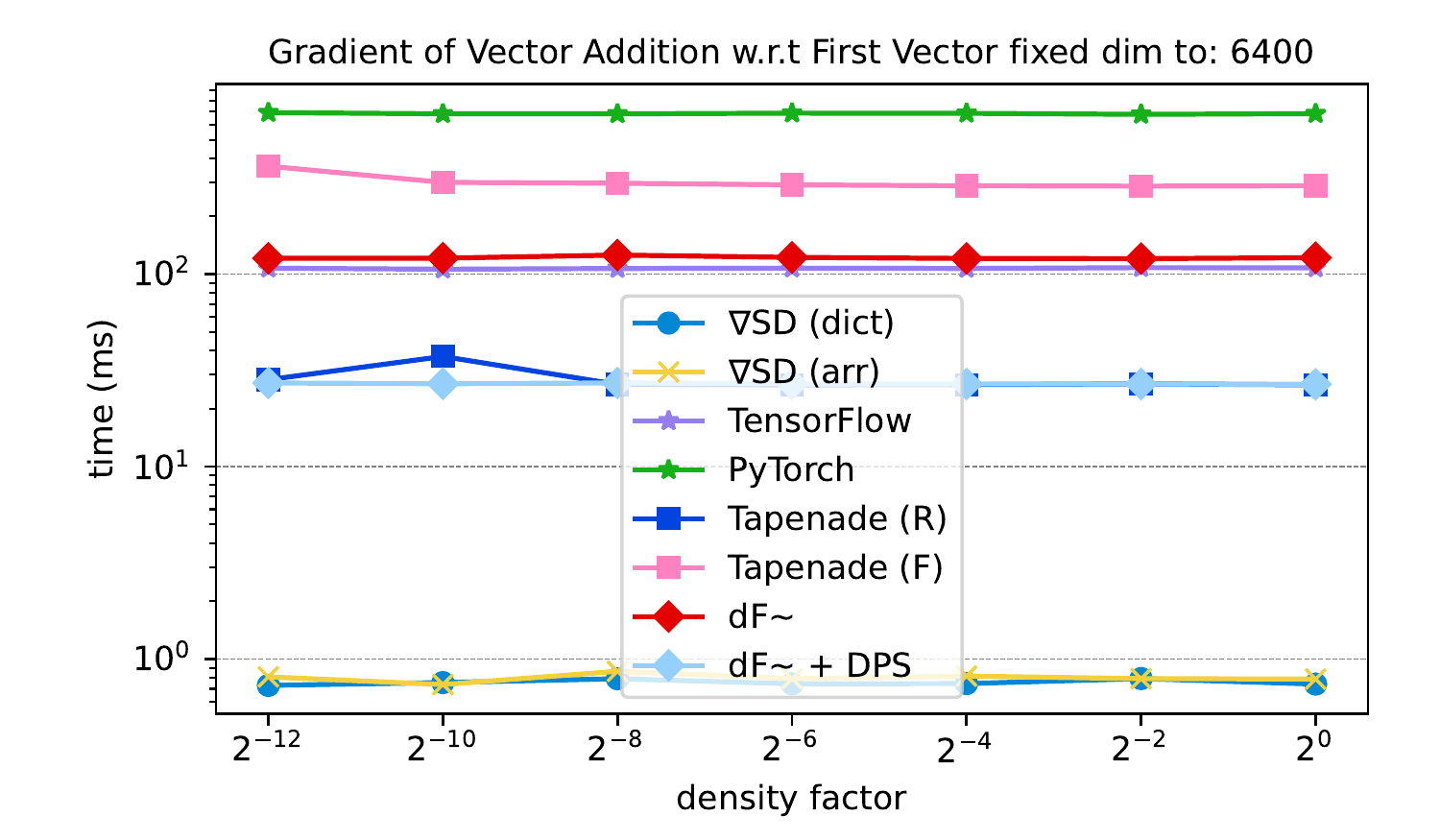} & \includegraphics[width=0.33\linewidth]{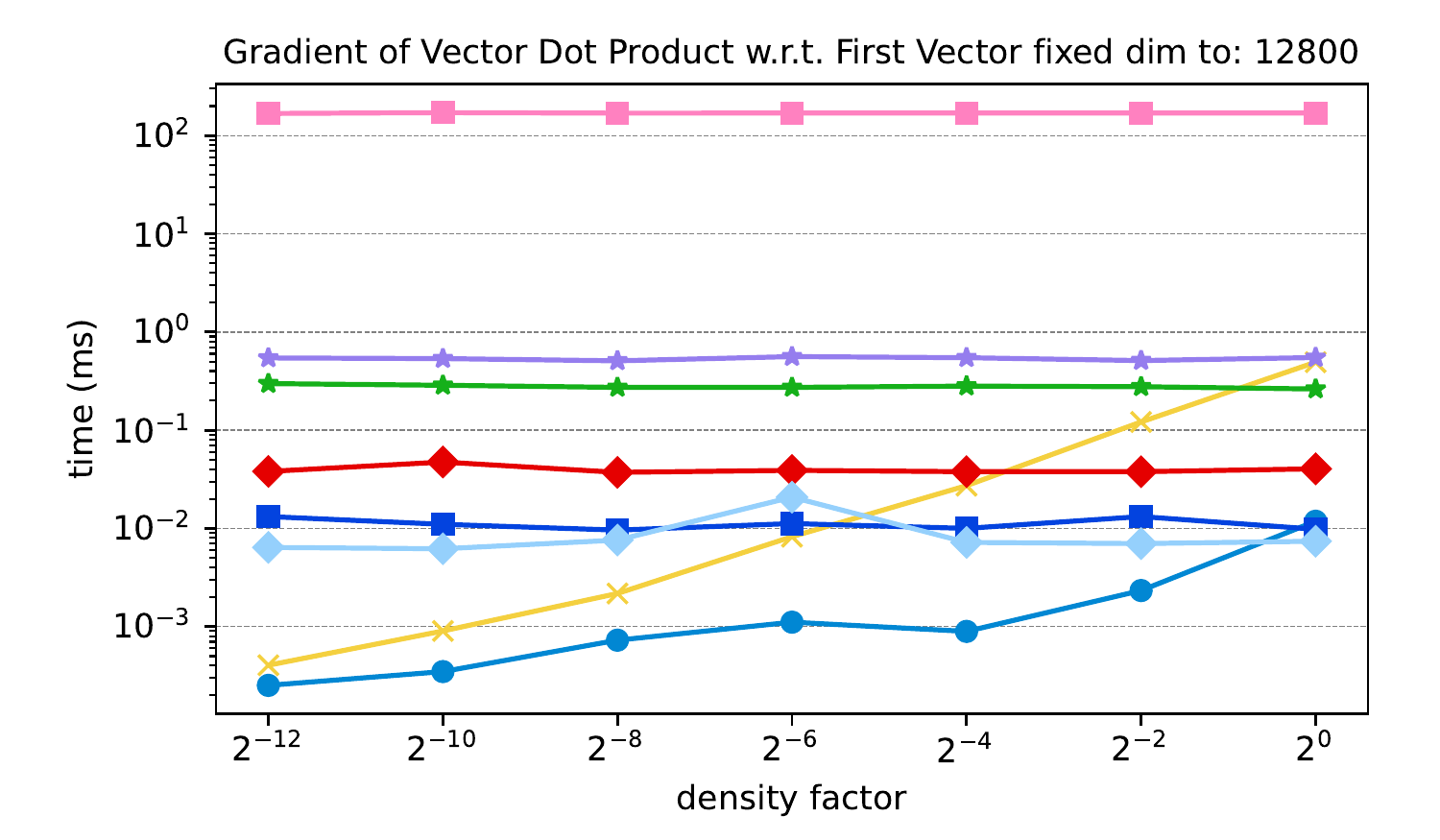} & 
         \includegraphics[width=0.33\linewidth]{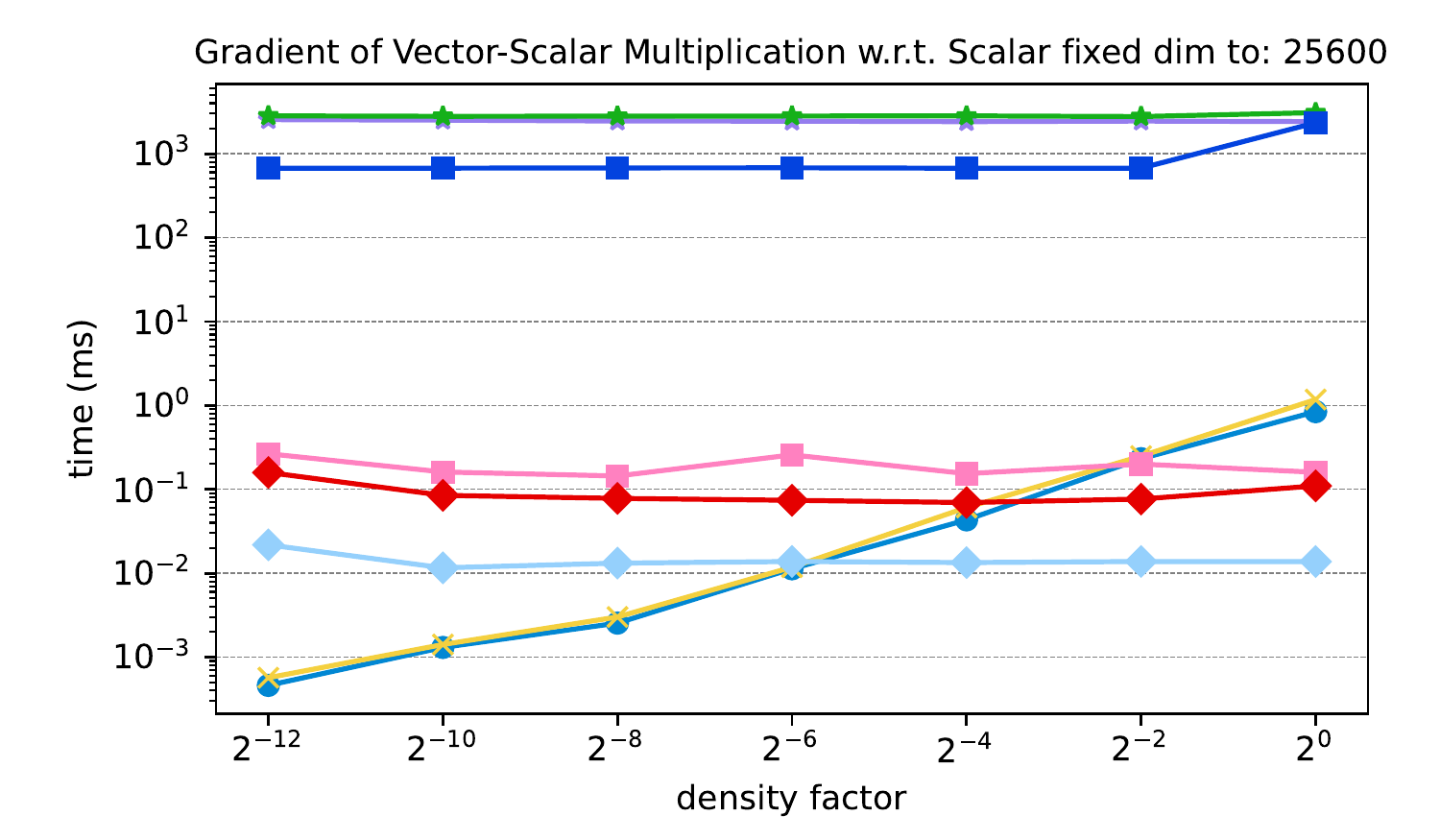}
    \end{tabular} 
    \vspace{-0.5cm}
    \caption{Performance results for differentiation of different sparse vector kernels by varying sparsity.}
    \label{fig:vectorkernels:sp}
\end{figure}

\begin{figure}
\setlength\tabcolsep{.5pt}
    \centering
    \begin{tabular}{ccc}
        \includegraphics[width=0.33\linewidth]{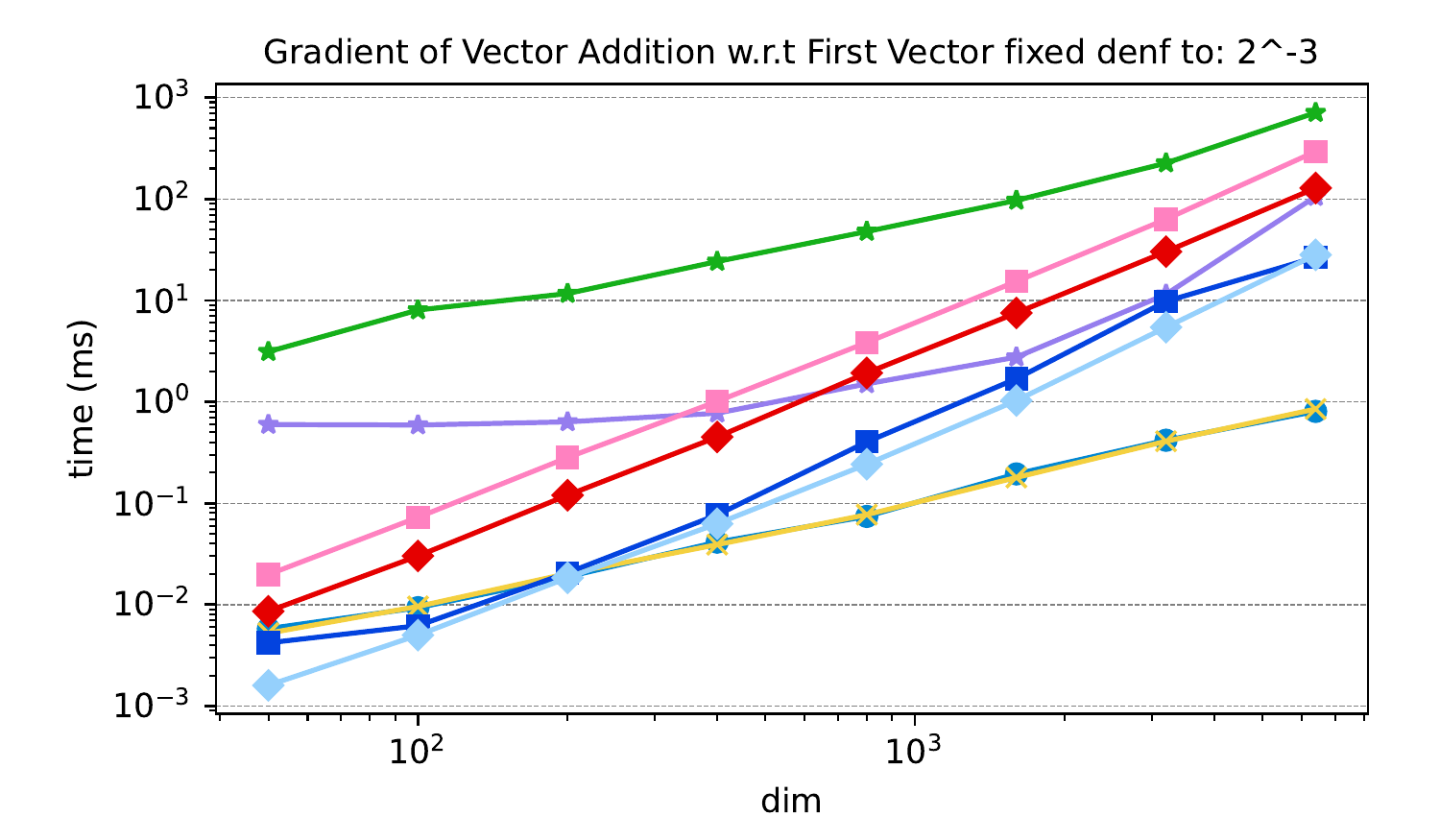} & \includegraphics[width=0.33\linewidth]{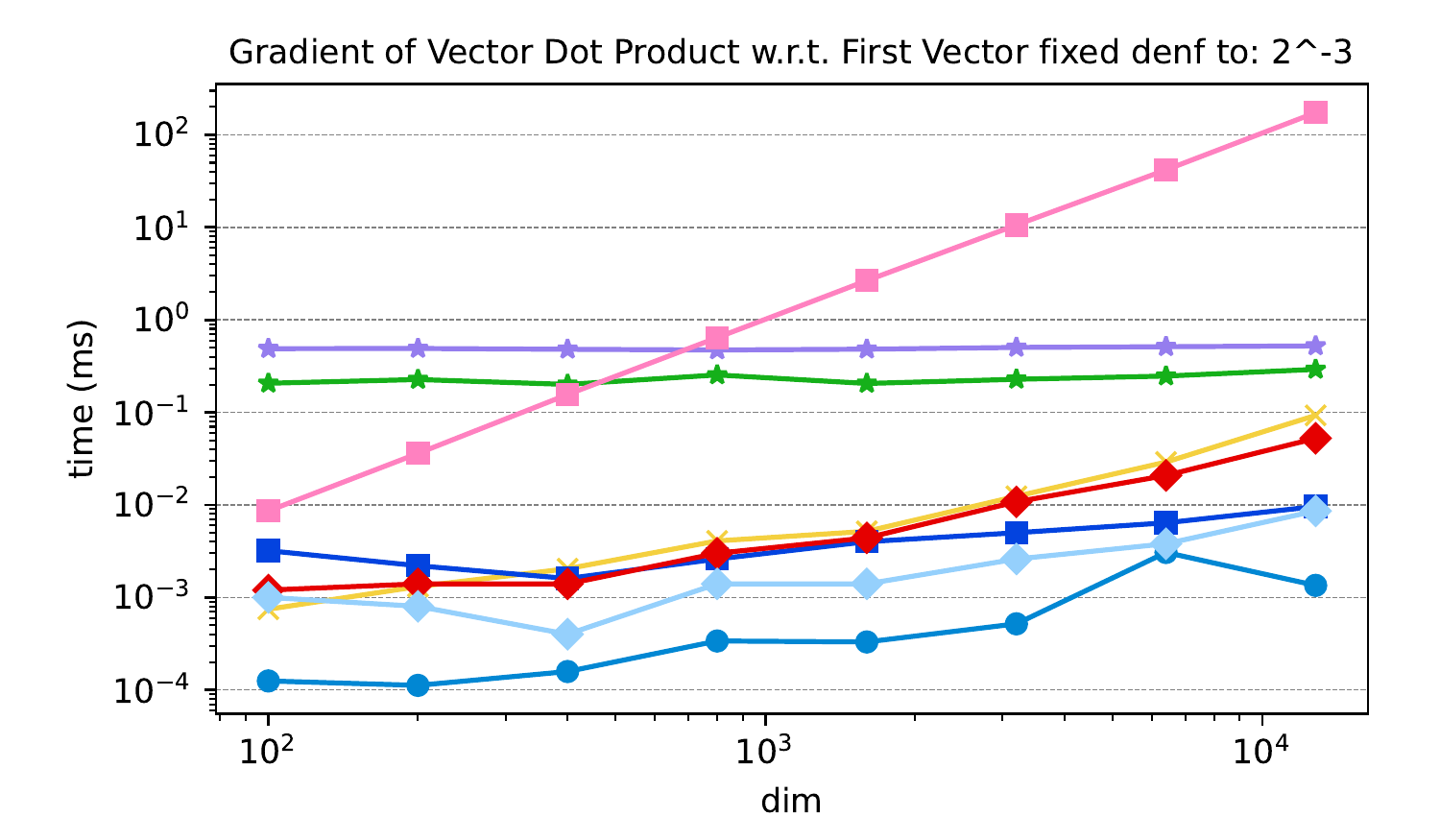} & \includegraphics[width=0.33\linewidth]{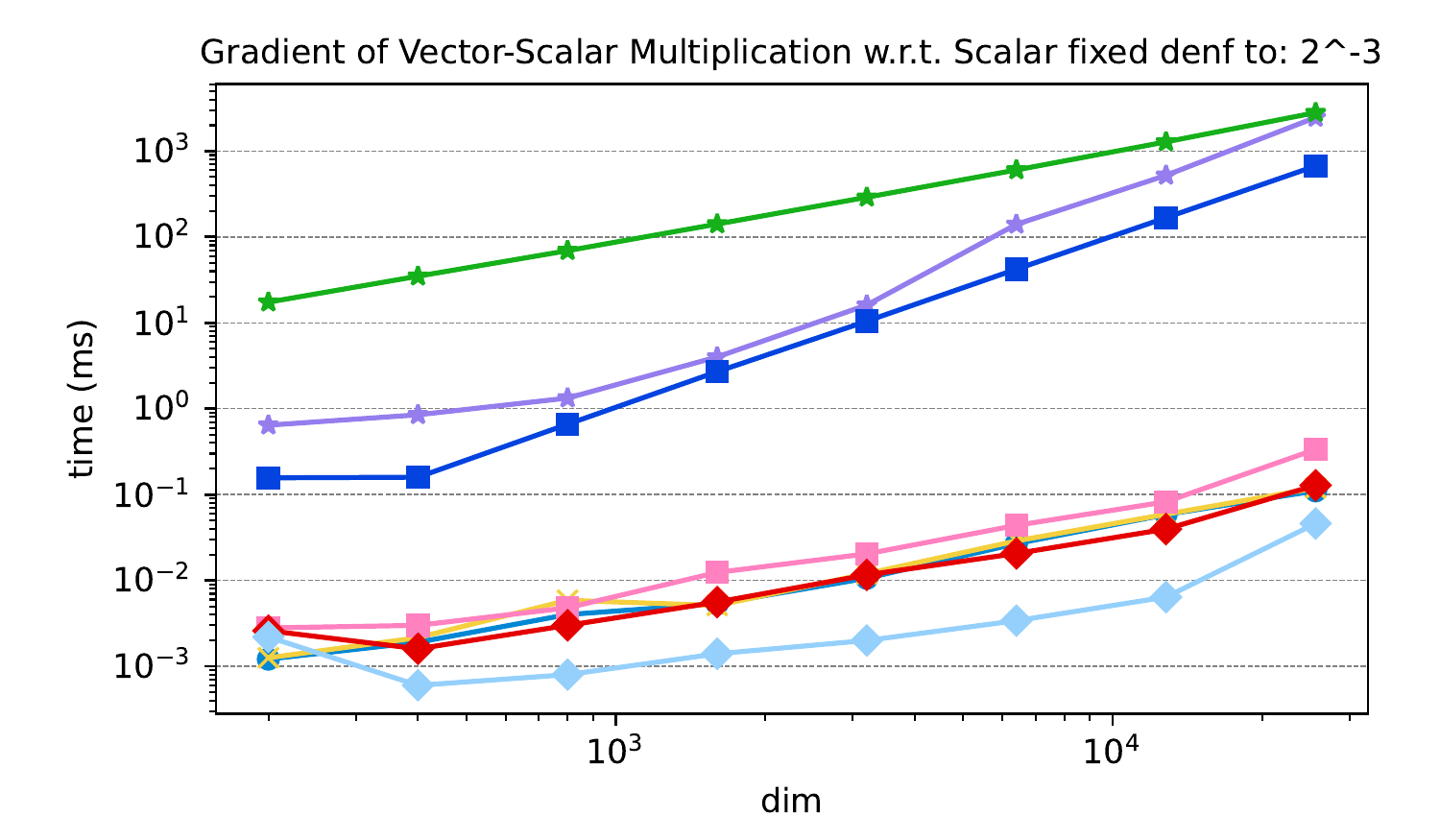} \\
         \includegraphics[width=0.33\linewidth]{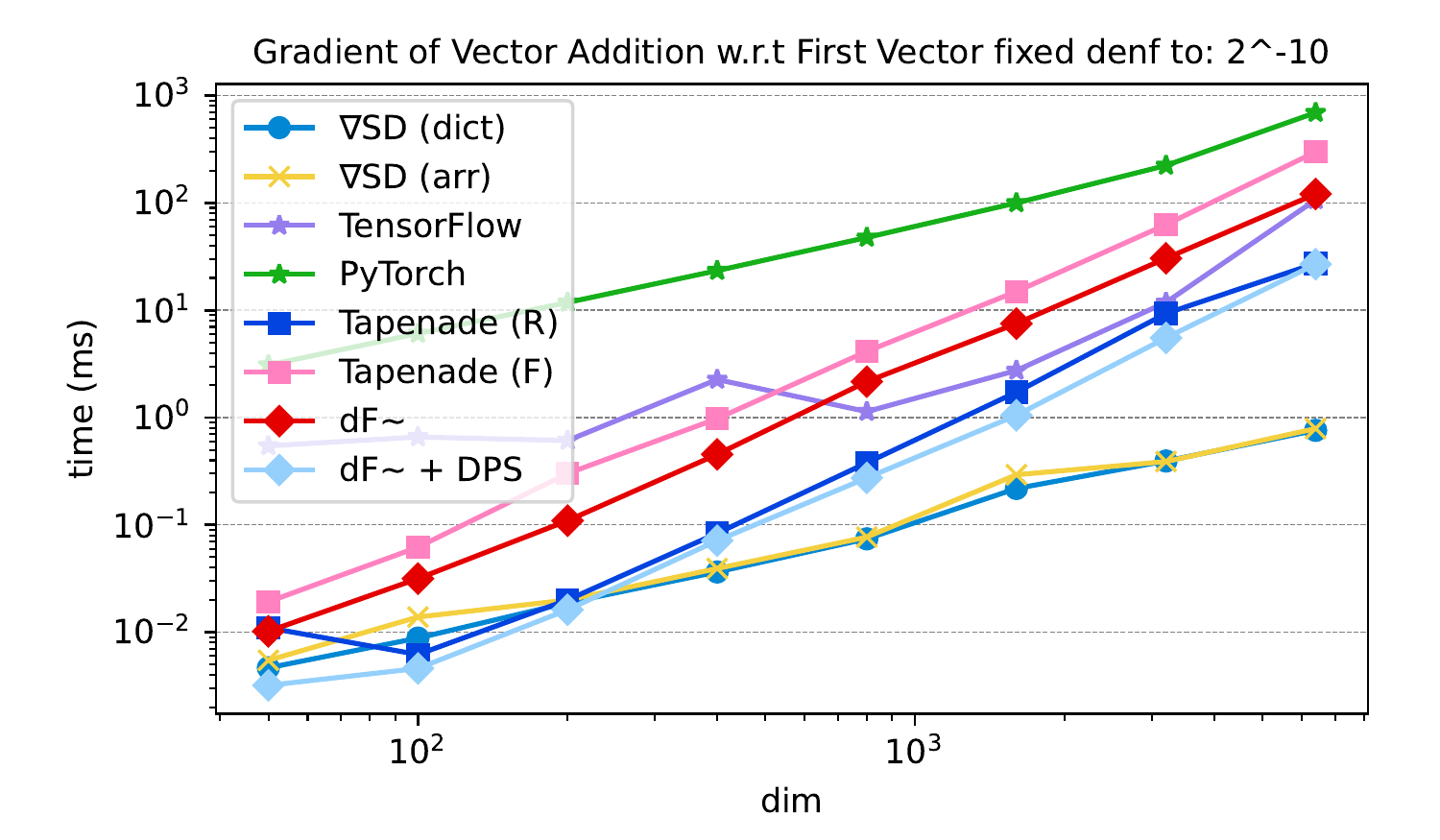} & \includegraphics[width=0.33\linewidth]{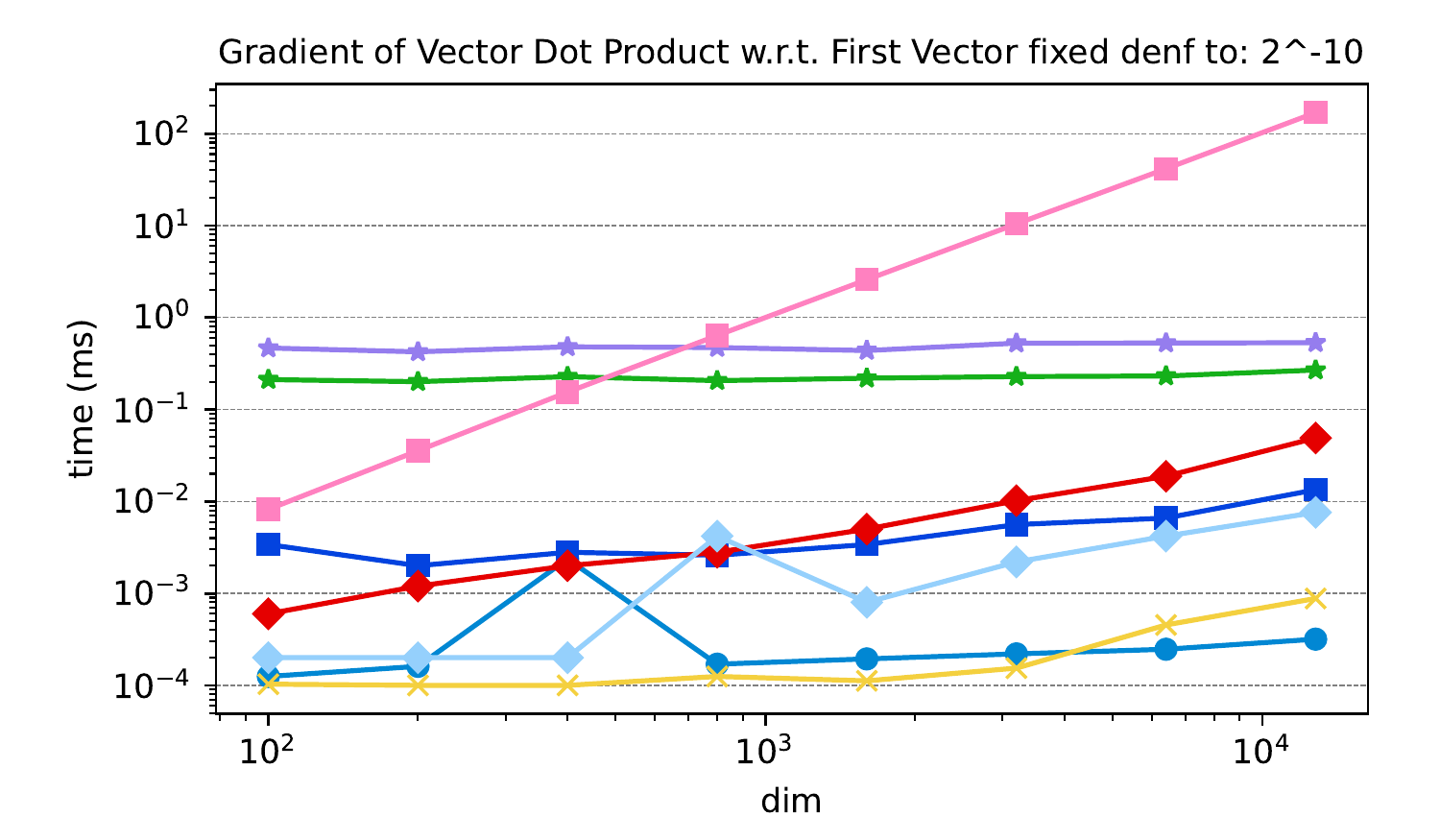} & \includegraphics[width=0.33\linewidth]{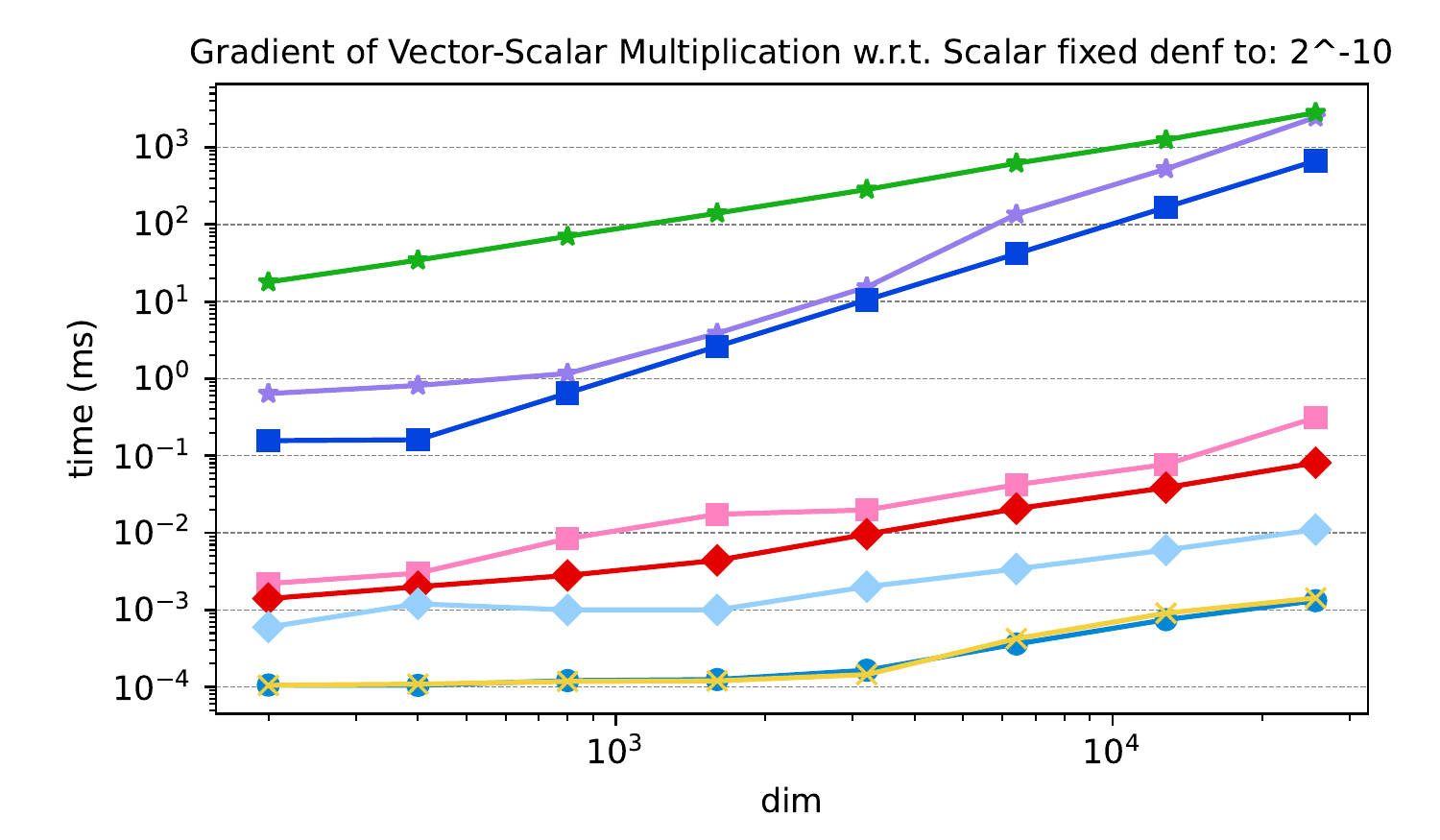} \\
    \end{tabular} 
    \vspace{-0.5cm}
    \caption{Performance results for differentiation of different sparse vector kernels by varying dimensions.}
    \label{fig:vectorkernels:dim}
    \vspace{-0.3cm}
\end{figure}

\smartpara{Vector Kernels} Figure~\ref{fig:vectorkernels:sp} and Figure~\ref{fig:vectorkernels:dim} show the results for vector kernels by varying the sparsity and density. For smaller dimensions, there is no clear advantage for \system{}; for VVA \dfsmooth performs better than \system{}. However, for larger dimensions, we observe a clear advantage for \system{}, especially for matrices with lower densities. For VSM, we observe a clear advantage for forward-mode-based systems over reverse-mode-based ones; for lower densities \system{} outperforms the rest, whereas for higher densities \dfsmooth{} is the most performant system.

\begin{figure}
\setlength\tabcolsep{.5pt}
    \centering
    \begin{tabular}{cc}
        \includegraphics[width=0.48\linewidth]{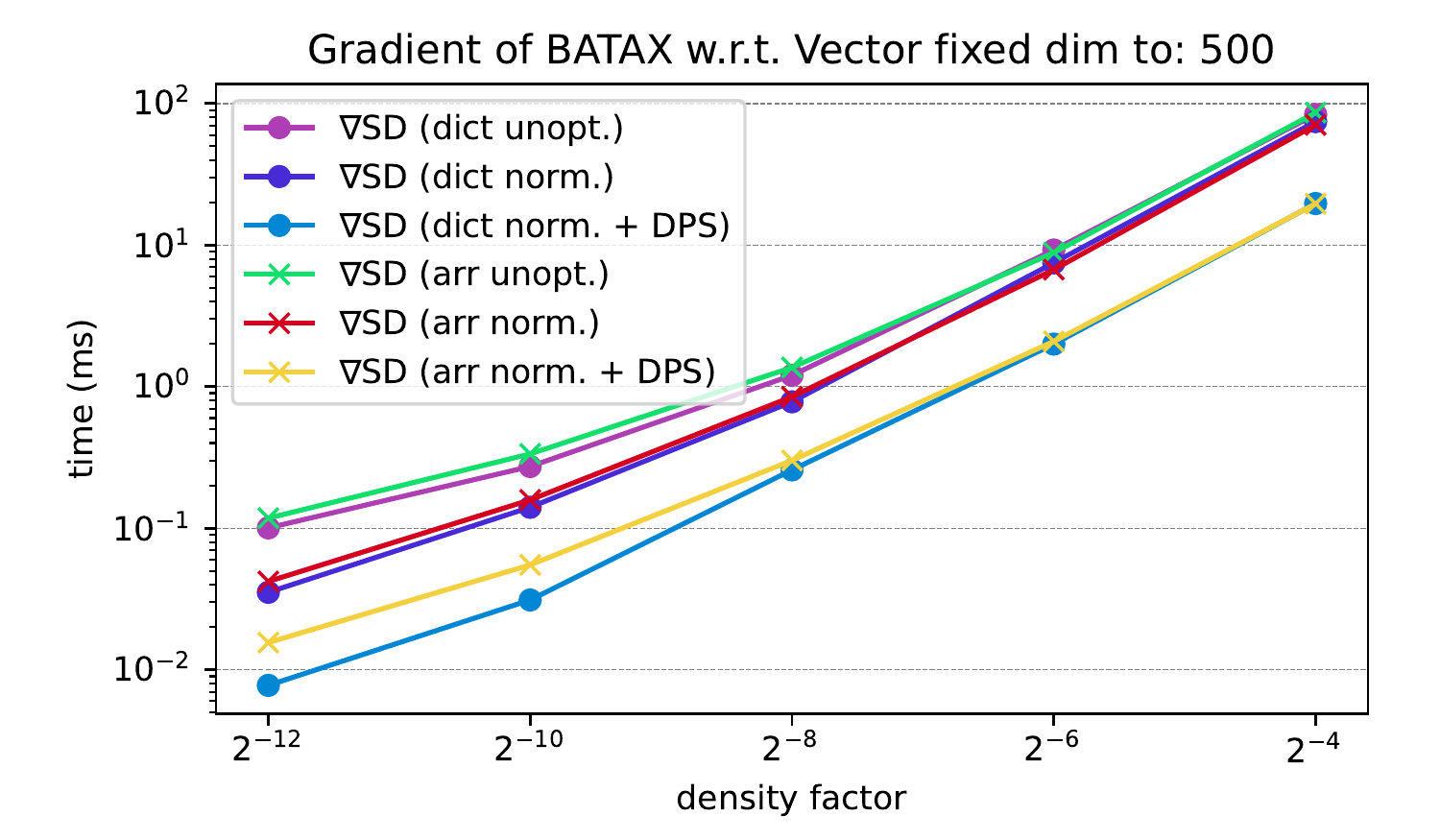} & \includegraphics[width=0.48\linewidth]{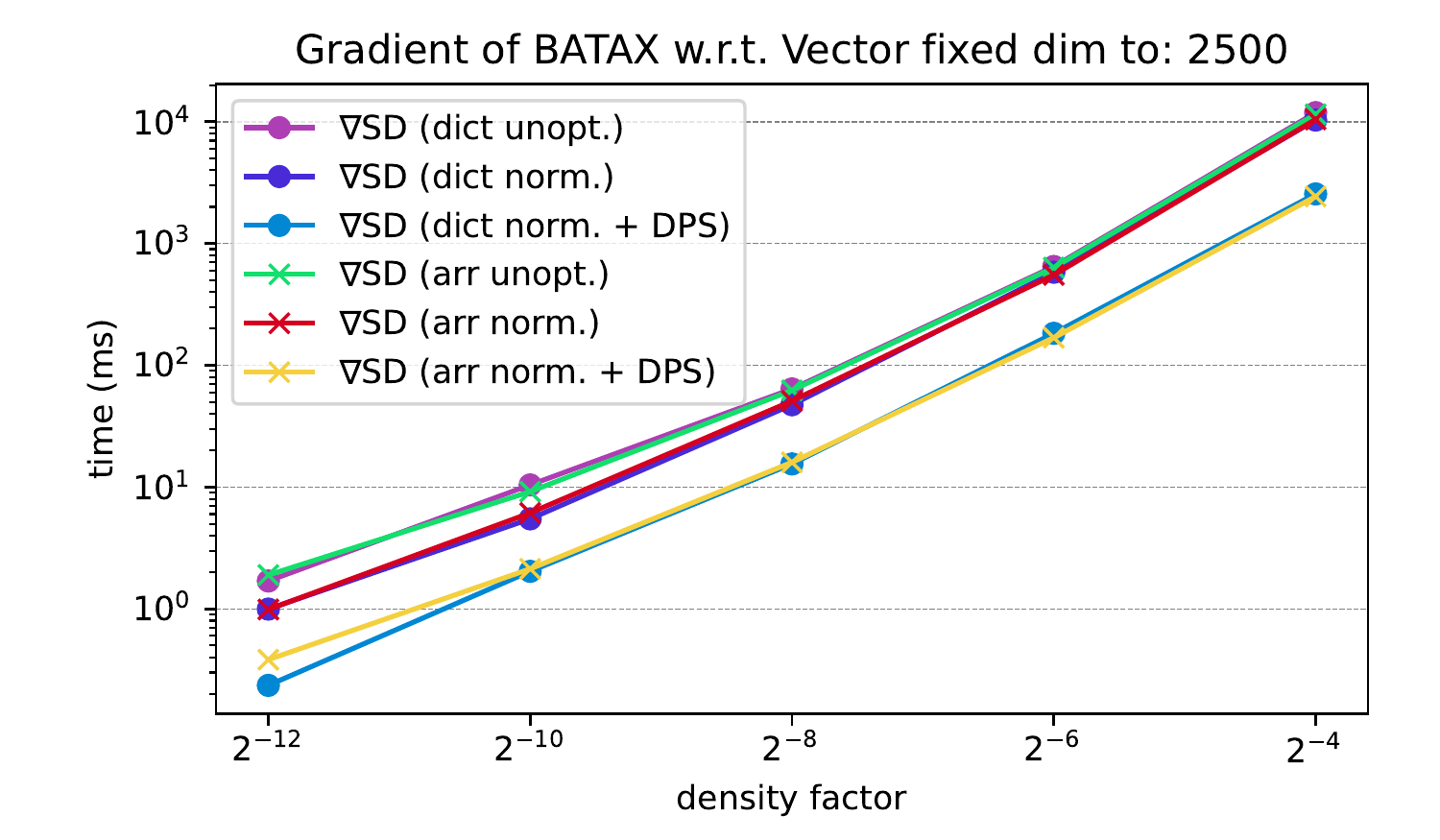} \\
        \includegraphics[width=0.48\linewidth]{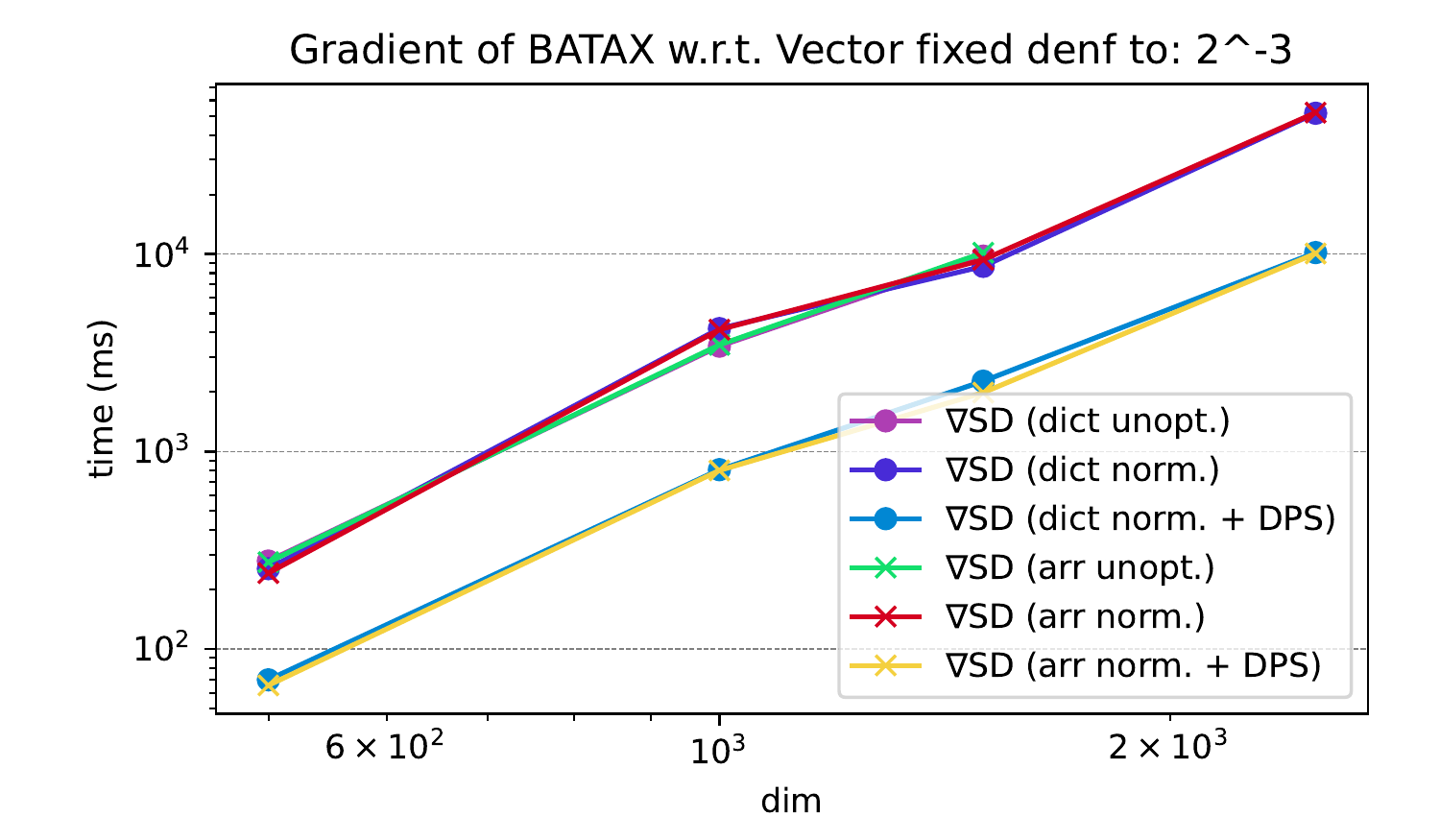} & \includegraphics[width=0.48\linewidth]{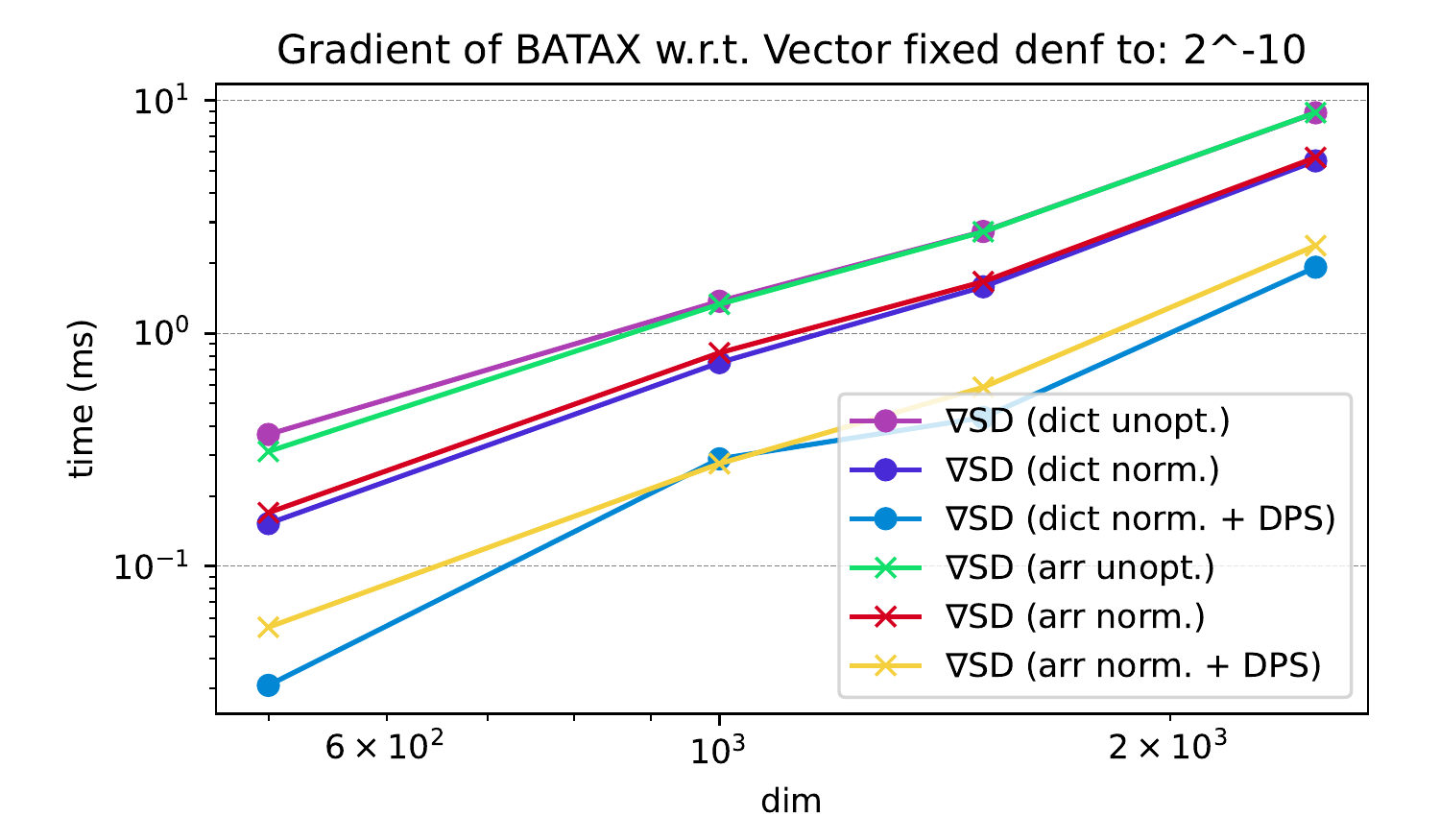} \\
    \end{tabular} 
    \vspace{-0.5cm}
    \caption{Impact of the normalization and DPS over the performance of \system on the BATAX kernel.}
    \label{fig:batax:opt}
    \vspace{-0.3cm}
\end{figure}

\subsection{Impact of Optimizations}
In this section, we investigate the impact of optimizations on the performance of \system{}. For each data representation, we consider the following alternatives for the generated C++ code: (1) without low-level transformations, (2) with multiplication normalization, and (3) with multiplication normalization and removal of intermediate tensors in the nested loops using DPS.

Figure~\ref{fig:batax:opt} shows the results for the BATAX kernel. Similarly to the previous sets of experiments, we see a slight advantage for the dictionary-based representation over the array-based one. This is due to the intermediate dictionary created in the array-based version (cf. Section~\ref{sub:dps-and-cpp}). Furthermore, each of the optimizations has a positive impact on performance. Especially for lower densities, the impact of optimizations becomes more pronounced.

\section{Related Work}


\smartpara{Automatic Differentiation}
There are several existing automatic differentiation (AD) frameworks and libraries for imperative and functional programming languages. ADIFOR~\cite{DBLP:journals/sp/BischofCCGH92} and Tapenade~\cite{DBLP:journals/toms/HascoetP13} perform AD for Fortran and C programs respectively, while Adept~\cite{DBLP:journals/toms/Hogan14} and ADIC~\cite{DBLP:journals/procedia/NarayananNW10} perform AD for C++ using expression templates. ForwardDiff~\cite{DBLP:journals/corr/RevelsLP16} uses vector forward-mode AD for differentiating Julia programs, while DiffSharp~\cite{DBLP:journals/corr/BaydinPS15} is an AD library implemented in F\# that provides both forward-mode and reverse-mode. Stalingrad~\cite{DBLP:journals/toplas/PearlmutterS08} is an optimizing compiler for a dialect of Scheme with a first-class AD operator and supports both forward mode and reverse mode of AD. Similarly, Karczmarczuk~\cite{DBLP:conf/icfp/Karczmarczuk98} presents a Haskell implementation for both forward and reverse mode AD, and Elliott~\cite{DBLP:journals/pacmpl/Elliott18} provides a generalization of AD based on category theory for implementing both forward and reverse-mode AD. There has been recent efforts on providing correct and asymptotically efficient reverse-mode AD for functional languages~\cite{DBLP:journals/pacmpl/KrawiecJKEEF22,DBLP:journals/pacmpl/SmedingV23,DBLP:journals/pacmpl/RadulPFJM23}, the ideas of which are implemented in JAX~\cite{jax2018github,frostig2018compiling} and Dex~\cite{paszke2021getting}.

Machine learning libraries like Tensorflow and PyTorch are implemented based on tensor abstractions. These systems come with a predefined set of efficient kernels for manipulating tensors and can use compilation backends for further optimization. Lantern~\cite{DBLP:journals/pacmpl/WangZDWER19} uses multi-stage programming to perform reverse-mode AD. However, none of the mentioned frameworks supports AD for sparse data structures with irregular storage formats~\cite{tensorflow:sp:issue,jax:sp:issue,pytorch:sp:issue}. There have been efforts on statically incorporating sparsities, however, this requires manually specifying the sparsity patterns by the programmers~\cite{10.1145/3578360.3580259} and do not work for large sparse matrices with arbitrary patterns~\cite{ghorbani2022compiling}.

\smartpara{Sparse Tensor Algebra} Sparse tensor algebra has been the focus of much research and development in recent years, leading to the emergence of several frameworks and systems designed to support it. TACO~\cite{kjolstad:2017:taco,DBLP:journals/pacmpl/ChouKA18} is a system capable of handling both sparse and dense computations over tensor algebra. Another noteworthy framework is the sparse polyhedral framework~\cite{DBLP:journals/pieee/StroutHO18}, which extends the capabilities of polyhedral compilation to support sparse tensor algebra. In addition, packages such as SciPy~\cite{DBLP:journals/corr/abs-1907-10121}, TensorFlow, PyTorch, and the MATLAB Tensor Toolbox~\cite{DBLP:journals/siamsc/BaderK07} offer support for various sparse matrix representations, enabling efficient computation on sparse tensor data structures commonly found in scientific and engineering applications.

Despite the progress made in this area, automatic differentiation for sparse tensors is still not widely supported. One of the few recent efforts~\cite{nytko2022optimized} provides manual gradients for a limited set of kernels. The primary challenge in differentiating sparse tensors is the irregular data representation, making differentiation a complex process. To address this issue, we propose separating the logical sparse tensor representation from its physical storage format, allowing for more efficient and effective differentiation.

\section{Conclusion \& Outlook}
In this paper, we present \system{}, the first framework that supports automatic differentiation for sparse tensors. Our main insight is to separate the logical concerns from the physical data storage representations. We provide a tensorized forward-mode transformation over the logical fragment of \lang, a recently introduced language that fuses the physical storage of sparse tensors with the logical specification of kernels. By benefiting from the algebraic optimizations globally applied using equality saturation we further improve the performance of differentiated programs and fuse their physical representations. We show experimentally that our framework outperforms the state-of-the-art differentiable tensor frameworks over both real-world and synthetic datasets.

For the future, we plan to add support for reverse-mode AD (by borrowing ideas from~\cite{DBLP:journals/pacmpl/RadulPFJM23,berg2022forward}). In addition, we plan to add the support for scheduling transforms to add the GPU backend~\cite{senanayake:2020:scheduling}. This way we can use our framework for training deep, yet sparse learning models such as Graph Neural Networks (GNNs)~\cite{nytko2022optimized}. Finally, we aim to add the support for the entire \sdql{}~\cite{DBLP:journals/pacmpl/ShaikhhaHSO22} (including boolean and integer semi-rings required for set and bag construction as well as tupling constructs) in order to provide automatic differentiation for relational databases and hybrid relational-linear algebra workloads.

\section*{Acknowledgments}
The first author thanks Huawei for their support of the distributed data management and
processing laboratory at the University of Edinburgh. The second author is supported by a
Royal Society University Research Fellowship.

\bibliography{refs}

\end{document}